\documentclass[11pt]{article}%
\usepackage{amsfonts}
\usepackage{amsmath}
\usepackage{amssymb}
\usepackage{graphicx}%
\providecommand{\U}[1]{\protect\rule{.1in}{.1in}}
\newtheorem{theorem}{Theorem}
\newtheorem{acknowledgement}[theorem]{Acknowledgement}

\newtheorem{corollary}[theorem]{Corollary}

\newtheorem{lemma}[theorem]{Lemma}

\newtheorem{proposition}[theorem]{Proposition}

\newenvironment{proof}[1][Proof]{\noindent\textbf{#1.} }{\ \rule{0.5em}{0.5em}}
\begin{document}

\title{Quantum Harmonic Analysis of the Density Matrix}
\author{Maurice A. de Gosson\\University of Vienna,\\Faculty of Mathematics (NuHAG)\\Oskar-Morgenstern-Platz 1, 1090 Vienna}
\maketitle
\tableofcontents

\begin{abstract}
We will study rigorously the notion of mixed states and their density
operators (or matrices.) We will also discuss the quantum-mechanical
consequences of possible variations of Planck's constant $h$. This Review has
been written having in mind two readerships: mathematical physicists and
quantum physicists. The mathematical rigor is maximal, but the language and
notation we use throughout should be familiar to physicists.

\end{abstract}

\section{Preface}

\paragraph{Quantum harmonic analysis}

When writing a Review on the density matrix you have several options. You can
use the standard approach found in physics graduate textbooks, but then you
can't go very far, not only because the important mathematical tools are
lacking, but also because the arguments are often flawed and incorrect: most
of them strictly speaking only apply when the underlying Hilbert space is
finite-dimensional (I am thinking here, among other things, about
\textquotedblleft trace taking\textquotedblright\ procedures for obtaining the
averages of observables which are most of the time mathematically undefendable
in infinite dimension in the absence of extra conditions on these
observables). Then, at the other extreme, you can use the $C^{\ast}$-algebraic
approach. This is certainly the most beautiful, elegant, and intellectually
satisfying way to introduce the mixed states of quantum mechanics and their
density matrices; after having listed some definitions and properties from the
theory of operator algebras the GNS construction then allows the passage from
this abstract theory to the usual Hilbert-space picture of quantum mechanics.
Unfortunately, this approach requires lots of preparatory work to become
accessible, and it is often difficult for the uninitiated to grasp the
physical meaning of the tools that are used. I have chosen here a third way,
which has the advantage of being both mathematically rigorous and intuitive.
It consists in using functional analysis together with tools from harmonic
analysis on phase space, especially Weyl--Wigner--Moyal theory (which is one
way to \textquotedblleft quantize\textquotedblright\ classical observables).
We will call this way of doing things the \textquotedblleft quantum harmonic
analysis\textquotedblright\ approach. One of its advantages is that it fully
justifies from a mathematical viewpoint the introduction of the Wigner
function of a density matrix, which otherwise appears as an \textit{ad hoc}
object pulled out of thin air. There is another major advantage to this
approach: it highlights the sensitivity of the theory of density matrices on
the choice of the value of Planck's constant, and this is precisely one of the
main themes we want to address. We will mostly deal with continuous-variable
systems with an infinite-dimensional Hilbert space described by observables
with continuous eigenspectra.

\paragraph{Variability of Planck's constant}

One topic we address in depth in this Review is the sensitivity of quantum
states to possible variations of Planck's constant. Physically this is a very
controversial question. Setting aside for a moment the debate on whether
Planck's constant can vary or not, consider the following situation: we have
an unknown quantum state, on which we perform a quorum of measurements in
order to determine its density matrix $\widehat{\rho}$. Now, one should be
aware of the fact that this density matrix will not be determined directly by
these experiments; what one does is to measure by, say, a homodyne quantum
tomography, certain properties of that system. Quantum homodyne tomography
originates from the observation by Vogel and Risken \cite{Vogel} that the
probability distributions determined by homodyne detection are just the Radon
transforms of the Wigner function of the density matrix, and that the latter
allows us to infer the density matrix using the Weyl correspondence. The
method works as follows: suppose now that we have been able to determine a
statistical function $\rho(x,p)$ of the position and momentum variables
(called \textquotedblleft quadratures\textquotedblright\ in this context),
yielding the properties of the quantum system under investigation. If we now
identify this function $\rho(x,p)$ with the Wigner distribution%
\begin{equation}
\rho(x,p)=\left(  \tfrac{1}{2\pi\hbar}\right)  ^{n}\int e^{-\frac{i}{\hbar}%
py}\langle x+\tfrac{1}{2}y|\widehat{\rho}|x-\tfrac{1}{2}y\rangle d^{n}y
\label{rhoxp1}%
\end{equation}
of the corresponding density matrix, this relation can be inverted and yields
$\widehat{\rho}$; mathematically speaking $\widehat{\rho}$ is just, up to a
factor, the Weyl operator corresponding to the \textquotedblleft
classical\textquotedblright\ observable $\rho(x,p)$, which can be written for
instance as%
\begin{equation}
\widehat{\rho}\psi(x)=\int\rho(x_{0},p_{0})e^{\frac{2i}{\hbar}p_{0}(x-x_{0}%
)}\psi(2x_{0}-x)d^{n}p_{0}d^{n}x_{0}. \label{rhoxp2}%
\end{equation}
Now, two essential observations. The first is that when using this procedure
we assume quite explicitly that we are using the Weyl--Wigner--Moyal (WWM)
formalism: we identify the function $\rho(x,p)$ with the Wigner function, and
\textquotedblleft quantize\textquotedblright\ it thereafter using the Weyl
transform. This is very good, of course, but one should keep in mind that
there are other possible representations in quantum mechanics; a physically
very interesting one is for instance the Born--Jordan quantization scheme. We
will not investigate the implications of such a choice in this Review, but we
want here to emphasize another problem. Even if we place ourselves in the WWM
framework we are tacitly assuming that $\hbar=h/2\pi$ has a fixed value in
time and space. If it happens that Planck's constant $h$ has another value,
$h^{\prime}$, at another location or at another time, formulas (\ref{rhoxp1})
and (\ref{rhoxp2}) would have to be replaced with the different expressions%
\begin{equation}
\rho^{\prime}(x,p)=\left(  \tfrac{1}{2\pi\hbar^{\prime}}\right)  ^{n}\int
e^{-\frac{i}{\hbar^{\prime}}py}\langle x+\tfrac{1}{2}y|\widehat{\rho}^{\prime
}|x-\tfrac{1}{2}y\rangle d^{n}y \label{rhoxp3}%
\end{equation}
and%
\begin{equation}
\widehat{\rho}^{\prime}\psi(x)=\int\rho(x_{0},p_{0})e^{\frac{2i}{\hbar
^{\prime}}p_{0}(x-x_{0})}\psi(2x_{0}-x)d^{n}p_{0}d^{n}x_{0}. \label{rhoxp4}%
\end{equation}

\paragraph{\textbf{Notation}}

Let $\sigma$ be the standard symplectic form on phase space $\mathbb{R}%
^{2n}\equiv\mathbb{R}_{x}^{n}\times\mathbb{R}_{p}^{n}$: by definition it is
the mapping which to the pairs of vectors $z=(x,p)$ and $z^{\prime}%
=(x^{\prime},p^{\prime})$ associates the number
\[
\sigma(z,z^{\prime})=\sum_{j}p_{j}x_{j}^{\prime}-p_{j}^{\prime}x_{j}%
=px^{\prime}-p^{\prime}x.
\]
The symplectic form can be conveniently written in matrix form as
\[
\sigma(z,z^{\prime})=(z^{\prime})^{T}Jz\text{ \ , \ }J=%
\begin{pmatrix}
0_{n\times n} & I_{n\times n}\\
-I_{n\times n} & 0_{n\times n}%
\end{pmatrix}
\]
($z$ and $z^{\prime}$ being here viewed as column vectors; $J$ is the
\textquotedblleft standard symplectic matrix\textquotedblright).

The inner product of two vectors $\psi,\phi\in L^{2}(\mathbb{R}^{n})$ is given
by
\[
\langle\psi|\phi\rangle=\int\psi^{\ast}(x)\phi(x)d^{n}x
\]
with $d^{n}x=dx_{1}\cdot\cdot\cdot dx_{n}$. The associated norm is
$||\psi||=\sqrt{\langle\psi|\psi\rangle}$.

In addition to the usual $\hbar$-dependent (unitary) Fourier transform%
\[
F\psi(p)=\left(  \tfrac{1}{2\pi\hbar}\right)  ^{n}\int e^{-\frac{i}{\hbar}%
px}\psi(x)d^{n}x
\]
on $\mathbb{R}^{n}$ we will use the symplectic Fourier transform on
$\mathbb{R}^{2n}$: it is the transformation $F_{\sigma}$ which takes a square
integrable function (or more generally a tempered distribution) $a$ on phase
space $\mathbb{R}^{2n}$ to the function (or tempered distribution)
\begin{equation}
F_{\sigma}a(z)=\left(  \tfrac{1}{2\pi\hbar}\right)  ^{n}\int e^{-\frac
{i}{\hslash}\sigma(z,z^{\prime})}a(z^{\prime})d^{2n}z^{\prime}. \label{sft}%
\end{equation}
That $F_{\sigma}$ is just an elementary modification of the $\hbar$-Fourier
transform on $\mathbb{R}^{2n}$ given by
\[
Fa(z)=\left(  \tfrac{1}{2\pi\hbar}\right)  ^{n}\int e^{-\frac{i}{\hslash
}zz^{\prime}}a(z^{\prime})d^{2n}z^{\prime}%
\]
is easy to see: since $\sigma(z,z^{\prime})=(z^{\prime})^{T}Jz=Jz\cdot
z^{\prime}$ ($J$ the standard symplectic matrix) we have $F_{\sigma
}a(z)=Fa(Jz)$. The symplectic Fourier transform is that it is its own inverse:
$F_{\sigma}^{2}=F_{\sigma}F_{\sigma}$ is the identity (\textit{i.e.} the
symplectic Fourier transform is involutive) and it satisfies the modified
Plancherel identity
\begin{equation}
\int a(z)b(z)d^{2n}z=\int F_{\sigma}a(z)F_{\sigma}b(-z)d^{2}z.
\label{Plancherelsig}%
\end{equation}

\section{Introduction\label{sec1}}

We begin by recalling informally the basic definitions from the theory of
quantum states; since the content of this section is common knowledge we do
not provide any particular references (we are following here the terminology
and the exposition in Peres \cite{Peres} with some modifications). The
formalism of density operators and matrices was introduced by John von Neumann
\cite{Neumann} in 1927 and independently, by Lev Landau and Felix Bloch in
1927 and 1946 respectively. Ugo Fano was one of the first to put the theory of
the density matrix in a rigorous form in his review paper \cite{Fano}.

\subsection{Quantum states and observables}

\subsubsection{Pure and mixed states; maximal tests}

A quantum system is said to be in a \emph{pure} state if we have complete
knowledge about that system, meaning we know exactly which state it is in.
Pure states can be prepared using \emph{maximal tests} (\cite{Peres}, \S 2-3):
suppose we are dealing with a quantum system and let $N$ be the maximum number
of different outcomes that can be obtained in a test of that system. If such a
test has exactly $N$ different outcomes, it is called a maximal test.\ The
quantum system under consideration is in a pure state if\ it is prepared in
such a way that it certainly yields a predictable outcome in that maximal
test, the outcomes in any other test having well-defined probabilities which
do not depend on the procedure used for the preparation. A pure state can thus
be identified by specifying the complete experiment that characterizes it
uniquely (Fano \cite{Fano}). One usually writes a pure state using Dirac's ket
notation $|\psi\rangle$; for all practical purposes it is convenient to use
the wavefunction $\psi$ defined by $\psi(x)=\langle x|\psi\rangle$; it is a
normalized element of a certain Hilbert space $\mathcal{H}$, which is usually
identified in the case of continuous variables with $L^{2}(\mathbb{R}^{n})$
(the square integrable functions). When doing this the state is identified
with the linear span of the function $\psi$, that is the ray $\mathbb{C}%
\psi=\{\lambda\psi:\lambda\in\mathbb{C}\}$. It is very important to note that
the pure state $|\psi\rangle$ can be identified with the orthogonal projection
$\widehat{\rho}_{\psi}$ of $\mathcal{H}$ on the subspace $\mathbb{C}\psi$.
This projection, which is of rank one, is denoted by $|\psi\rangle\langle
\psi|$ in quantum mechanics; it is analytically given by the formula
\begin{equation}
\widehat{\rho}_{\psi}\phi=|\psi\rangle\langle\psi|\phi\rangle\label{1}%
\end{equation}
where $\langle\psi|\phi\rangle$ is identified with the inner product in
$\mathcal{H}$. Most tests are however not maximal, and most preparations do
not produce pure states, so we only have partial knowledge of the quantum
system under consideration. The information on such a system is less than a
maximum, with reference to the lack of a complete experiment with a uniquely
predetermined outcome. The state of the system is nevertheless fully
identified by any data adequate to predict the statistical results of all
conceivable observations on the system \cite{Fano}. When this is the case we
say that the system is in a \emph{mixed state}. Mixed states are classical
probabilistic mixtures of pure states; however, different distributions of
pure states can generate physically indistinguishable mixed states (this
possibility will be discussed later). A quantum mixed state can be viewed as
the datum of a set of pairs $\{(\psi_{j},\alpha_{j})\}$ where $\psi_{j}$ is a
(normalized square integrable) pure state and $\alpha_{j}$ a classical
probability; these probabilities sum up to one: $\sum_{j}\alpha_{j}=1$. A
\textit{caveat}: one should not confuse the mixed state $\{(\psi_{j}%
,\alpha_{j})\}$ with the superposition $\psi=\sum_{j}\alpha_{j}\psi_{j}$ which
is a \textit{pure} state!

\subsubsection{The density matrix}

From a mathematical point of view the quantum state can be advantageously
described by a self-adjoint operator on a Hilbert space; this operator is
called the density matrix of the quantum system; in the pure state case this
operator is just the orthogonal projection (\ref{1}) on a subspace of the
Hilbert space, while the density matrix of a mixed state is
\begin{equation}
\widehat{\rho}=\sum_{j}\alpha_{j}\widehat{\rho}_{j}=\sum_{j}\alpha_{j}%
|\psi_{j}\rangle\langle\psi_{j}| \label{trc1}%
\end{equation}
where we have set $\widehat{\rho}_{j}=\widehat{\rho}_{\psi_{j}}$. It should be
kept in mind that the density matrix describes a preparation procedure for an
ensemble of quantum systems whose statistical properties correspond to the
given preparation procedure (this important aspect will be discussed in more
detail below).

\begin{theorem}
\label{ThmA}(i) The density operators on a Hilbert space $\mathcal{H}$ form a
convex subset $\operatorname{Dens}(\mathcal{H)}$ of the space $\mathcal{B}%
(\mathcal{H)}$ of bounded operators on $\mathcal{H}$. (ii) The extreme points
of this set are the rank-one projections, which correspond to the pure states
on $\mathcal{H}$.
\end{theorem}

\begin{proof}
(i) To say that the set of density matrices is convex means that if
$\widehat{\rho}_{1}$ and $\widehat{\rho}_{2}$ are in $\operatorname{Dens}%
(\mathcal{H)}$ then so is $\lambda\widehat{\rho}_{1}+(1-\lambda)\widehat{\rho
}_{2}$ for all real numbers $\lambda$ such that $0\leq\lambda\leq1$. Let
$\{(\psi_{j},\alpha_{j}):j\in K\}$ and $\{(\phi_{j},\beta_{j}):j\in L\}$ be
two mixed states. Relabeling if necessary the indices we may assume that the
sets $K$ and $L$ are disjoint: $K\cap L=\emptyset$. The corresponding density
matrices are%
\[
\widehat{\rho}_{1}=\sum_{k\in K}\alpha_{k}|\psi_{k}\rangle\langle\psi
_{k}|\text{ \ , \ }\widehat{\rho}_{2}=\sum_{\ell\in L}\beta_{\ell}|\phi_{\ell
}\rangle\langle\phi_{\ell}|
\]
and hence we have by linearity
\[
\lambda\widehat{\rho}_{1}+(1-\lambda)\widehat{\rho}_{2}=\sum_{j\in K\cup
L}\gamma_{j}|\chi_{j}\rangle\langle\chi_{j}|
\]
where $\chi_{j}=\psi_{j}$, $\gamma_{j}=\lambda\alpha_{j}$ if $j\in K$ and
$\chi_{j}=\phi_{j}$, $\gamma_{j}=(1-\lambda)\beta_{j}$ if $j\in L$. That
$\widehat{\rho}=\lambda\widehat{\rho}_{1}+(1-\lambda)\widehat{\rho}_{2}$ is a
mixed state now follows from the equality%
\[
\sum_{j\in K\cup L}\gamma_{j}=\lambda\sum_{j\in K}\alpha_{j}+(1-\lambda
)\sum_{j\in L}\beta_{j}=1.
\]
(ii) That a pure state really is \textquotedblleft pure\textquotedblright%
,\textit{ i.e. }that it can never be represented as a mixed state is easily
seen using the following algebraic argument: assume that $\widehat{\rho}%
_{\psi}=|\psi\rangle\langle\psi|$ can be rewritten as a sum $\sum_{j}%
\alpha_{j}|\psi_{j}\rangle\langle\psi_{j}|$ with $\alpha_{j}\geq0$ and
$\sum_{j}\alpha_{j}=1$. In fact, discarding the terms with $\alpha_{j}=0$ we
may assume that $\alpha_{j}>0$ for all indices $j$. Let now $(\mathbb{C}%
\psi)^{\bot}$ be the subspace of $\mathcal{H}$ orthogonal the ray
$\mathbb{C}\psi$: it consists of all vectors $\phi$ in $\mathcal{H}$ such that
$\langle\psi|\phi\rangle=0$. For every $\phi\in(\mathbb{C}\psi)^{\bot}$ we
have $\langle\phi|\widehat{\rho}_{\psi}|\phi\rangle=|\langle\psi|\phi
\rangle|^{2}=0$ and hence also
\[
\langle\phi|\widehat{\rho}_{\psi}|\phi\rangle=\sum_{j}\alpha_{j}|\langle
\phi|\psi_{j}\rangle|^{2}=0;
\]
since we are assuming that $\alpha_{j}>0$ this equality implies that we must
have $\langle\phi|\psi_{j}\rangle=0$ for all $\phi\in S_{\psi}$ and every
$\psi_{j}$, hence $\psi_{j}\in((\mathbb{C}\psi)^{\bot})^{\bot}$. The
orthogonality relation for subspaces being reflexive we have $((\mathbb{C}%
\psi)^{\bot})^{\bot}=\mathbb{C}\psi$, which means that each $\psi_{j}$ belongs
to the ray $\mathbb{C}\psi$, that is $\psi_{j}=\lambda\psi$ for some complex
number $\lambda$ with $|\lambda|=1$; the vectors $\psi_{j}$ thus define the
state $\langle\psi|$. It follows from this argument that the pure state
density matrices are the extreme points of $\operatorname{Dens}(\mathcal{H)}$.
This means that if $\widehat{\rho}_{\psi}$ is a pure state density matrix,
then the relation $\widehat{\rho}_{\psi}=\lambda\widehat{\rho}_{1}%
+(1-\lambda)\widehat{\rho}_{2}$ with $\lambda\neq0$ and $\lambda\neq1$ implies
$\widehat{\rho}_{\psi}=\widehat{\rho}_{1}=\widehat{\rho}_{2}$. That this is
the case immediately follows from the argument above.
\end{proof}

\subsubsection{The $C^{\ast}$-algebraic approach\label{secalgebraic}}

The notion of quantum state and density matrices can be very concisely
described using the language of $C^{\ast}$-algebras; this approach is useful
when one wants to give a rigorous axiomatic description of the theory of
quantum systems. Its inception goes back to early work by John von Neumann on
operator algebras around 1930; the modern theory of $C^{\ast}$-algebras was
developed one decade later by Gel'fand and Naimark. We recommend Rieffel's
paper \cite{Rieffel} for a survey of quantization using this approach;
Landsman's book \cite{landsman} is another excellent source, perhaps somewhat
closer to physicists' background. Let $\mathcal{A}$ be a non-commutative
unital $C^{\ast}$-algebra (\textit{i.e}. a unital complex Banach algebra
equipped with an isometric involution $A\longmapsto A^{\ast}$ compatible with
complex conjugation, and such that $||A^{\ast}A||=||A||^{2}$). If $A=A^{\ast}$
we call $A$ an observable; the set of observables is denoted by $\mathcal{O}$.
If in addition the spectrum of $A\in\mathcal{O}$ consists of numbers
$\lambda\geq0$ we say that $A$ is positive ($A\geq0$). A linear map
$\rho:\mathcal{O}\longrightarrow\mathbb{C}$ is positive ($\rho\geq0$) if
$A\geq0$ implies $\rho(A)\geq0$; if in addition $\rho(1)=1$ then $S$ is called
a \emph{quantum state}; quantum states form the state space $\mathcal{S}%
(\mathcal{A)}$ of $\mathcal{A}$; it is a convex cone whose extreme points are
the \textit{pure states} while the other elements of $\mathcal{S}%
(\mathcal{A)}$ are \textit{mixed states}. Now, the Gel'fand--Naimark theorem
\cite{dix,gena,landsman} implies that for every unital $C^{\ast}$-algebra
$\mathcal{A}$ there exists an isomorphism $\pi:A\longmapsto\widehat{A}$ of
$\mathcal{A}$ onto a separable Hilbert space $\mathcal{H}$ such that two of
the Dirac--von Neumann axioms (\cite{takta}, pp. 66--67) are satisfied:
$\rho\in\mathcal{S}(\mathcal{A)}$ if and only if there exists a positive
$\widehat{\rho}\in\mathcal{B}(\mathcal{H)}$ with $\operatorname{Tr}%
(\widehat{\rho})=1$ and $\rho(A)=\operatorname{Tr}(\widehat{\rho}\widehat{A})$
for every $A\in\mathcal{A}$; moreover the space of observables $\mathcal{O}$
is identified with the self-adjoint elements of $\mathcal{B}(\mathcal{H)}$ and
the number $\operatorname{Tr}(\widehat{\rho}\widehat{A})$ is then the expected
value of $A\in\mathcal{O}$. The $C^{\ast}$-algebraic approach is very
appealing because it helps in unifying classical and quantum states: when the
system is classical, the algebra of observables becomes an abelian $C^{\ast}%
$-algebra and the states become ordinary probability measures.

\subsection{\textquotedblleft Taking traces\textquotedblright: heuristics}

\subsubsection{The average of an observable}

Let $\widehat{A}$ be an observable, \textit{i.e.} a Hermitian operator defined
on some subspace $D_{\widehat{A}}$ of $\mathcal{H}$ (the domain of
$\widehat{A}$); the average value of $\widehat{A}$ in the pure state
$|\psi\rangle$ is by definition%
\begin{equation}
\langle\widehat{A}\rangle_{\psi}=\langle\psi|\widehat{A}|\psi\rangle
=\langle\psi|\widehat{A}\psi\rangle=\langle\widehat{A}\psi|\psi\rangle.
\label{average1}%
\end{equation}
This relation is usually rewritten using the notion of trace
$\operatorname{Tr}$ of an operator, defined as the sum of the eigenvalues of
that operator. In fact, one of the most used formulas in quantum mechanics is,
no doubt,
\begin{equation}
\langle\widehat{A}\rangle_{\psi}=\operatorname{Tr}(\widehat{\rho}_{\psi
}\widehat{A}). \label{trc2}%
\end{equation}
To see this, we observe that since $\widehat{\rho}_{\psi}\widehat{A}%
=|\psi\rangle\langle\psi|\widehat{A}$ and $\widehat{\rho}_{\psi}%
^{2}=\widehat{\rho}_{\psi}$ we have, by cyclicity,%
\[
\operatorname{Tr}(\widehat{\rho}_{\psi}\widehat{A})=\operatorname{Tr}%
(\widehat{\rho}_{\psi}^{2}\widehat{A})=\operatorname{Tr}(\widehat{\rho}_{\psi
}\widehat{A}\widehat{\rho}_{\psi})
\]
that is, in \textquotedblleft bra-c-ket notation\textquotedblright, and using
the homogeneity of the trace,
\begin{align*}
\operatorname{Tr}(\widehat{\rho}_{\psi}\widehat{A})  &  =\operatorname{Tr}%
(|\psi\rangle\langle\psi|\widehat{A}|\psi\rangle\langle\psi|)\\
&  =\langle\psi|\widehat{A}|\psi\rangle\operatorname{Tr}(|\psi\rangle
\langle\psi|)\\
&  =\langle\psi|\widehat{A}|\psi\rangle
\end{align*}
since $\operatorname{Tr}(|\psi\rangle\langle\psi|)=\operatorname{Tr}%
(\widehat{\rho}_{\psi})=1$ (the only eigenvalues of $\widehat{\rho}_{\psi}$
are $1$ and $0$). Formula (\ref{trc2}) is immediately generalized by linearity
of the trace to arbitrary mixed states, and one finds, that the average value
of the observable $\widehat{A}$ in a mixed state with density operator
(\ref{trc1}) is
\begin{equation}
\langle\widehat{A}\rangle_{\widehat{\rho}}=\operatorname{Tr}(\widehat{\rho
}\widehat{A}). \label{trc3}%
\end{equation}

It turns out that several mixtures can lead to the same density matrix. Since
formula (\ref{trc3}) only involves the density operator itself, such
equivalent mixtures cannot be distinguished by measurement of observables
alone because all observable results can be predicted from the density matrix,
without needing to know the ensemble that was used to construct it.

\subsubsection{Precautions}

So far, so good. However, the manipulations made above can -- at best -- be
justified when the Hilbert space $\mathcal{H}$ is finite dimensional because
in this case the notion of trace class operators (and their properties)
involves elementary calculations of finitely-dimensional matrices. This is
almost always implicitly assumed in textbooks, where the reader is invited to
study \textit{ad} \textit{nauseam }the same elementary examples in low
dimensions. Unfortunately the definitions above, and the proof of the formulas
(\ref{trc2})--(\ref{trc3}) are \emph{incorrect }in the general case of an
infinite-dimensional $\mathcal{H}$, and formulas like (\ref{trc3}) are not
justifiable unless one makes very restrictive hypotheses on the observable
$\widehat{A}$ because of convergence problems for the involved integrals.
These issues are seldom questioned and treated sloppily in most of the
physical literature (in this context it is very instructive to read the
excellent paper \cite{daube} by Daubechies which warns of such
misconceptions). The situation is very much similar to what happens when
physicists manipulate the Feynman path integral\footnote{I am not talking here
about the rigorous theory of the Feynman integral, as found for instance in
Albeverio \textit{et al}. \cite{albe} or Nicola \cite{Nicola}.}. It is a
heuristic object which has, in the form used most of the time in physics, no
mathematical justification, but which easily allows one to get new intuitions
and insights (hence its usefulness). Of course, such formal manipulations are
okay when one is only interested in qualitative statements, but they are
certainly not okay when one wants to perform exact calculations (numerical, or theoretical).

Here is another popular (and questionable!) way of calculating traces. Let
$\widehat{\rho}$ be a density matrix on $L^{2}(\mathbb{R}^{n})$; assume that
the kernel of $\widehat{\rho}$ is a function $K(x,y)$:%
\[
\widehat{\rho}\psi(x)=\int K(x,y)\psi(y)d^{n}y.
\]
It is customary (especially in the physical literature) to calculate the trace
of $\widehat{\rho}$ by integrating the kernel along the diagonal, that is by
using the formula
\begin{equation}
\operatorname*{Tr}(\widehat{A})=\int K(x,x)d^{n}x \label{trak}%
\end{equation}
which is obviously an extension to the infinite dimensional case of the usual
definition of the trace of a matrix as the sum of its diagonal elements.
Needless to say, this formula does not follow directly from the definition of
a trace class operator! In fact, even when the integral in (\ref{trak}) is
absolutely convergent, this formula has no reason to be true in general (the
kernel of an operator is defined uniquely only up to a set of measure zero,
and in (\ref{trak}) we are integrating along the diagonal, which precisely has
measure zero!). Brislawn \cite{Brislawn} discusses formulas of the type
(\ref{trak}) in the context of traceable Hilbert--Schmidt operators and gives
sufficient conditions for such formulas to hold using Mercer's theorem and its
variants. As shown in Barry Simon's little monograph \cite{BSimon} on trace
ideals, the right condition for formula (\ref{trak}) to hold in the case $n=1$
is that $K$ must be of positive type, which means that we have
\begin{equation}
\sum_{1\leq j,k\leq N}\lambda_{j}\overline{\lambda_{k}}K(x_{j},x_{k})\geq0
\label{lalak}%
\end{equation}
for all integers $N$, all $x_{j}\in\mathbb{R}$ and all $\lambda_{j}%
\in\mathbb{C}$ (in particular we must thus have $K\geq0)$. On the positive
side, Simon (\textit{ibid}.) notes that if a density matrix has kernel $K$
satisfying $\int|K(x,x)|d^{n}x<\infty$ then we are \textquotedblleft almost
sure\textquotedblright\ that formula (\ref{trak}) holds. Of course this vague
statement is not a charter allowing carefree calculations! Needless to say,
the \textquotedblleft derivations\textquotedblright\ above are formal and one
should be extremely cautious when using the \textquotedblleft
formulas\textquotedblright\ thus obtained. Shubin \cite{sh87}, \S 27,
discusses a (not so easy to use explicitly) step-by-step procedure for
checking such identities. We will come back to these essential points in
Section \ref{secalcul}.

\subsection{Quantum Tomography}

Reconstructing a quantum state (or even a classical state) is an extremely
important problem, and it is, generally speaking, a difficult one. We will
discuss here very briefly this topic; we will mainly be giving references to
past works; such a list is of course incomplete due to the large number of
contributions, which is steadily growing. \ We will mainly focus on quantum
tomography, which is a technique for characterizing a state of a quantum
system by subjecting it to a large number of quantum measurements, each time
preparing the system anew.

\subsubsection{Estimating the density matrix: generalities}

Experiments performed in a laboratory on quantum systems do not lead directly
and precisely to a determination of the density matrix. The problem can be
formulated as a statistical one: can we estimate the density matrix using
repeated measurements on quantum systems that are identically prepared? The
following strategy is used: after having obtained measurements on these
identical quantum systems we can make a statistical inference about the
probability distribution of the measurements, and thus indirectly about the
density matrix of the quantum system \cite{waxu}. This procedure is called
\textit{quantum state tomography}, and is practically implemented using a set
of measurements of a \textit{quorum} of observables, \textit{i.e}., of a
minimal\ \textquotedblleft complete\textquotedblright\ set of non-commuting
observables. See D'Ariano and his collaborators \cite{dariano,damapa} for
various strategies, and also Vogel and Risken \cite{Vogel} or Leonhardt and
Paul \cite{lepa94}. Bu\v{z}ek and his collaborators \cite{buzek} show how to
reconstruct the Wigner function using the Jaynes principle of Maximum Entropy.
In that context, also see the recent paper Thekkadath \textit{et al}.
\cite{the} where a scheme that can be used to directly measure individual
density matrix elements is given; Lvovsky and Raymer \cite{lvra09} give a very
interesting review in the continuous-variable case. The Lecture Notes volume
\cite{pare} edited by Paris and Reh\'{a}\v{c}ek contains a selection of texts
by leading experts presenting various aspects of quantum state estimation.
Mancini \textit{et al}. \cite{mancini} apply symplectic geometry to describe
the dynamics of a quantum system in terms of equations for a purely classical
probability distribution; also see Man'ko and Man'ko \cite{mankomanko}. Ibort
\textit{et al}. \cite{ib} introduce what they call the tomographic picture of
quantum mechanics.

Although the Wigner function cannot be measured directly as a probability
density, all its marginal distributions can. Once we know all the marginal
distributions associated with different quadratures, we can reconstruct the
Wigner function. Historically, the problem can be traced back to Wolfgang
Pauli's question in his 1958 book \cite{Pauli} \textquotedblleft\textit{The
mathematical problem as to whether, for given probability densities }%
$W(p)$\textit{ and }$W(x)$\textit{, the wavefunction }$\psi$\textit{ (...) is
always uniquely determined, has still not been investigated in all its
generality}\textquotedblright. It turns out that it is not possible to
determine a state by knowing only the configuration space and the momentum
space marginal distributions. The answer is negative; here is an elementary
example (Esposito \textit{et al.} \cite{espo}): consider the two normalized
wavefunctions%
\[
\psi_{1}(x)=\left(  \frac{2}{\pi}\operatorname{Re}\alpha\right)
^{1/4}e^{-\alpha x^{2}}\text{ \ , \ }\psi_{2}(x)=\left(  \frac{2}{\pi
}\operatorname{Re}\alpha\right)  ^{1/4}e^{-\alpha^{\ast}x^{2}}%
\]
where $\alpha$ is a complex number whose real part is positive:
$\operatorname{Re}\alpha>0$ and imaginary part nonzero: $\operatorname{Im}%
\alpha\neq0$. These functions become%
\begin{align*}
\phi_{1}(p)  &  =\left(  \frac{2}{\pi}\operatorname{Re}\alpha\right)
^{1/4}\frac{1}{\sqrt{2\alpha}}e^{-\alpha p^{2}}\\
\phi_{2}(p)  &  =\left(  \frac{2}{\pi}\operatorname{Re}\alpha\right)
^{1/4}\frac{1}{\sqrt{2\alpha^{\ast}}}e^{-\alpha p^{2}}%
\end{align*}
in the momentum representation, so that the associated probability
distributions are the same%
\begin{align*}
|\psi_{1}(x)|^{2}  &  =|\psi_{2}(x)|^{2}=\left(  \frac{2}{\pi}%
\operatorname{Re}\alpha\right)  ^{1/2}e^{-2(\operatorname{Re}\alpha)x^{2}}\\
|\phi_{1}(p)|^{2}  &  =|\phi_{2}(p)|^{2}=\left(  \frac{2}{\pi}%
\operatorname{Re}\alpha\right)  ^{1/2}\frac{1}{2|\alpha|}%
e^{-2(\operatorname{Re}\alpha)x^{2}}.
\end{align*}
However the states $|\psi_{1}\rangle$ and $|\psi_{2}\rangle$ are different,
because%
\[
|\langle\psi_{1}|\psi_{2}\rangle|^{2}=\frac{\operatorname{Re}\alpha}{|\alpha
|}\neq1
\]
so that we cannot have $\psi_{1}=c\psi_{2}$ for some complex constant $c$ such
that $|c|=1$.

\subsubsection{The Radon transform}

Let us begin with a short discussion of the classical case followed by a
motivation for the quantum case; we are following the very clear exposition in
\S 12.5 of the recent book \cite{espo} by Esposito \textit{et al}. Let
$\rho=\rho(x,p)$ be some function defined on the phase space $\mathbb{R}^{2}$
(it could be a classical probability, or a quantum quasiprobability). The
Radon transform of this function is defined formally in physics texts by
\begin{equation}
\mathcal{R}\rho(X,\mu,\nu)=\iint\rho(x,p)\delta(X-\mu x-\nu p)dpdx.
\label{Radon1}%
\end{equation}
The function $\mathcal{R}\rho(X,\mu,\nu)$ is called a \emph{tomogram}. The
meaning of this integral is that one integrates the function $\rho$ along the
line $\ell$ with equation $X-\mu x-\nu p=0$ where the variable $X$ has an
arbitrary value. Now, the crucial point is that if we know fully the function
$R\rho$ -- which depends on the three real variables $X,\mu,\nu$ -- then one
can reconstruct the function $\rho$ by using the so-called \textquotedblleft
inverse Radon transform\textquotedblright:%
\begin{equation}
\rho(x,p)=\left(  \tfrac{1}{2\pi}\right)  ^{2}\iiint\mathcal{R}\rho(X,\mu
,\nu)e^{i(X-\mu x-\nu p)}dXd\mu d\nu. \label{Radon2}%
\end{equation}
Let us \textquotedblleft verify\textquotedblright\ this formula, just for the
fun of making formal calculations. Denoting by $A$ the triple integral above
we have%
\begin{align*}
A  &  =\int\rho(x^{\prime},p^{\prime})\left\{  \iint\left[
{\textstyle\int}
e^{i(X-\mu x-\nu p)}\delta(X-\mu x^{\prime}-\nu p^{\prime})dX\right]  d\mu
d\nu\right\}  dp^{\prime}dx^{\prime}\\
&  =\int\rho(x^{\prime},p^{\prime})\left\{  \iint\left[  e^{i(\mu
(x-x)+\nu(p-p^{\prime}))}\int\delta(X-\mu x^{\prime}-\nu p^{\prime})dX\right]
d\mu d\nu\right\}  dp^{\prime}dx^{\prime}\\
&  =\int\rho(x^{\prime},p^{\prime})\left\{  \iint e^{i(\mu(x-x)+\nu
(p-p^{\prime}))}d\mu d\nu\right\}  dp^{\prime}dx^{\prime}\\
&  =(2\pi)^{2}\int\rho(x^{\prime},p^{\prime})\delta(x-x^{\prime}%
)\delta(p-p^{\prime})dp^{\prime}dx^{\prime}%
\end{align*}
and hence $A=(2\pi)^{2}\rho(x,p)$, so we are done.

In the quantum case one proceeds as follows: noting that $\delta(X-\mu x-\nu
p)$ can be rewritten as the Fourier integral%
\[
\delta(X-\mu x-\nu p)=\frac{1}{2\pi\hbar}\int_{-\infty}^{\infty}e^{\frac
{i}{\hbar}k(X-\mu x-\nu p)}dk
\]
and hence (\ref{Radon1}) as%
\[
\mathcal{R}\rho(X,\mu,\nu)=\frac{1}{2\pi\hbar}\iint\rho(x,p)\left(
\int_{-\infty}^{\infty}e^{\frac{i}{\hbar}k(X-\mu x-\nu p)}dk\right)  dpdx
\]
one \emph{defines} the quantum Radon transform by analogy with (\ref{Radon1})
by the formula
\[
\mathcal{R}\rho(X,\mu,\nu)=\frac{1}{2\pi\hbar}\operatorname{Tr}\left(
\widehat{\rho}\int_{-\infty}^{\infty}e^{\frac{i}{\hbar}k(\widehat{X}%
-\mu\widehat{x}-\nu\widehat{p})}dk\right)  .
\]
Doing this one has of course to give a precise meaning to the exponential
$e^{\frac{i}{\hbar}k(\widehat{X}-\mu\widehat{x}-\nu\widehat{p})}$, and one
then \textquotedblleft inverts\textquotedblright\ the formula above, which
yields in analogy with (\ref{Radon2})%
\[
\widehat{\rho}=\frac{1}{2\pi\hbar}\int\mathcal{R}\rho(X,\mu,\nu)e^{\frac
{i}{\hbar}(\widehat{X}-\mu\widehat{x}-\nu\widehat{p})}dXd\mu d\nu.
\]
Of course these manipulations are heuristic and completely formal and do not
have any mathematical sense unless one defines rigorously the involved operators.

\section{Mathematical Theory of the Density Matrix}

To give a precise definition of the trace formulas above, we need to work a
little bit and use some operator theory (the theory of Hilbert--Schmidt and
trace class operators suffices at this point). We will use several times the
generalized Bessel equality%
\begin{equation}
\sum_{j}\langle\psi|\psi_{j}\rangle\langle\phi|\psi_{j}\rangle^{\ast}%
=\langle\psi|\phi\rangle\label{Besseleq}%
\end{equation}
valid for every orthonormal basis of a Hilbert space $\mathcal{H}$. It is an
immediate consequence of the equality $\sum_{j}|\psi_{j}\rangle\langle\psi
_{j}|=I$.

\subsection{Trace class operators}

\subsubsection{General algebraic definition}

Let us now define the notion of density operator in terms of the so important
notion of trace class operators. The starting point is to notice that formula
(\ref{trc1}) implies that a density operator has, to begin with, three basic properties:

\begin{enumerate}
\item[(1)] It is a \emph{bounded operator} on $\mathcal{H}$; in particular
$\widehat{\rho}$ is well-defined for all $\psi\in\mathcal{H}$;

\item[(2)] It is a \emph{self-adjoint operator}: $\widehat{\rho}^{\dag
}=\widehat{\rho}$;

\item[(3)] It is a \emph{positive operator}: $\langle\psi|\widehat{\rho}%
|\psi\rangle\geq0$ for all $\psi\in\mathcal{H}$.
\end{enumerate}

The boundedness of $\widehat{\rho}$ follows from the observation that
\[
||\widehat{\rho}\psi||\leq\sum_{j}\alpha_{j}||\widehat{\rho}_{j}\psi||\leq
\sum_{j}\alpha_{j}||\psi||=||\psi||
\]
where $||\psi||=\sqrt{\langle\psi|\psi\rangle}$; that $||\widehat{\rho}%
_{j}\psi||\leq||\psi||$ follows from the Cauchy--Schwarz inequality since
$\widehat{\rho}_{j}\psi=|\psi_{j}\rangle\langle\psi_{j}|\psi\rangle$. The
self-adjointness of $\widehat{\rho}$ is clear, since it is a linear
combination of the self-adjoint operators $\widehat{\rho}_{j}=|\psi_{j}%
\rangle\langle\psi_{j}|$. Finally, the positivity of $\widehat{\rho}$ is clear
since we have%
\[
\langle\psi|\widehat{\rho}|\psi\rangle=\sum_{j}\alpha_{j}\langle\psi|\psi
_{j}\rangle\langle\psi_{j}|\psi\rangle=\sum_{j}\alpha_{j}|\langle\psi|\psi
_{j}\rangle|^{2}\geq0
\]
(notice that $(3)\Longrightarrow(2)$ when $\mathcal{H}$ is a complex Hilbert
space). There remains the trace issue. As we mentioned above, physicists
define the trace of an operator $\widehat{\rho}$ as the sum of its diagonal
elements, and this definition only makes sense without further restrictions
when $\dim\mathcal{H}<\mathcal{\infty}$. The way to do things correctly
consists in using the mathematical definition of trace class operators; the
latter is very general (and hence very useful) because it involves arbitrary
bases of $\mathcal{H}$. The definition goes as follows: a bounded operator
$\widehat{\rho}$ on a Hilbert space $\mathcal{H}$ (we do not assume
self-adjointness or positivity at this point) is of \emph{trace class} if
there exist two orthonormal bases $(\phi_{j})$ and $(\chi_{j})$ of
$\mathcal{H}$ (indexed by the same set) such that%
\begin{equation}
\sum_{j}|\langle\widehat{\rho}\phi_{j}|\chi_{j}\rangle|<\infty. \label{trc4}%
\end{equation}
Notice that this definition immediately implies that $\widehat{\rho}$ is of
trace class if and only if its adjoint $\widehat{\rho}^{\dag}$ is. Now, a
crucial property is that if condition (\ref{trc4}) holds for one pair of
orthonormal bases, then it also holds for \emph{all} pairs of orthonormal
bases, and that if $(\psi_{j})$ and $(\phi_{j})$ are two such pairs then%
\begin{equation}
\sum_{i}\langle\psi_{j}|\widehat{\rho}\psi_{j}\rangle=\sum_{i}\langle\phi
_{j}|\widehat{\rho}\phi_{j}\rangle\label{trc5}%
\end{equation}
each series being absolutely convergent. The proof of this result is not
difficult; it consists in expanding each base using the vectors of the other;
we refer to Chapter 12 in de Gosson \cite{Birkbis} and to the Appendix 3 in
Shubin's book \cite{sh87} for complete proofs. This being done, we
\emph{define} the trace of $\widehat{\rho}$ by the formula
\begin{equation}
\operatorname{Tr}(\widehat{\rho})=\sum_{k}\langle\psi_{k}|\widehat{\rho}%
\psi_{k}\rangle\label{tr6}%
\end{equation}
where $(\psi_{j})$ is any orthonormal basis of $\mathcal{H}$; that the result
does not depend on the choice of such a basis follows from the identity
(\ref{trc5}). We leave to the reader to verify that as a consequence of this
definition the sum of two trace class operators is again a trace class
operator, and that the trace of their sum is the sum of the traces. Also,
$\operatorname{Tr}(\lambda\widehat{\rho})=\lambda\operatorname{Tr}%
(\widehat{\rho})$ for every complex number $\lambda$, hence density matrices
form a convex cone. Also notice that it immediately follows from definition
(\ref{tr6}) that
\begin{equation}
\operatorname{Tr}(\widehat{\rho}^{\dag})=\operatorname{Tr}(\widehat{\rho
})^{\ast} \label{tr7}%
\end{equation}
hence the trace is real if $\widehat{\rho}$ is self-adjoint.

\subsubsection{Invariance under unitary conjugation}

The following invariance of the trace under unitary conjugation is well known:

\begin{theorem}
Let $\widehat{\rho}$ be a density matrix on the Hilbert space $\mathcal{H}$
and $\widehat{U}$ a unitary operator on $\mathcal{H}$. Then $\widehat{U}%
^{\dag}\widehat{\rho}\widehat{U}$ is also a positive trace class operator and
we have
\begin{equation}
\operatorname*{Tr}(\widehat{U}^{\dag}\widehat{\rho}\widehat{U}%
)=\operatorname*{Tr}(\widehat{A}). \label{truau12}%
\end{equation}

\end{theorem}

\begin{proof}
It is clear that $\widehat{U}^{\dag}\widehat{\rho}\widehat{U}$ is a positive,
bounded, and self-adjoint operator. The operator $\widehat{\rho}$ is of trace
class if and only if $\sum_{j}\langle\psi_{j}|\widehat{\rho}|\psi_{j}%
\rangle<\infty$ for one (and hence every) orthonormal basis $(\psi_{j})_{j}$
of $\mathcal{H}$. Since $\langle\psi_{j}|\widehat{U}^{\dag}\widehat{\rho
}\widehat{U}|\psi_{j}\rangle=\langle\widehat{U}\psi_{j}|\widehat{\rho
}\widehat{U}\psi_{j}\rangle$ and the basis $(\widehat{U}\psi_{j})_{j}$ also is
orthonormal, the operator $\widehat{U}^{\dag}\widehat{\rho}\widehat{U}$ is of
trace class. The trace formula (\ref{truau12}) follows from the orthonormal
basis independence of the trace formula (\ref{tr6}).
\end{proof}

We have not yet related our definition of trace to that used in the previous
section. Let us show that these definitions do coincide. Let
\begin{equation}
\widehat{\rho}=\sum_{j}\alpha_{j}|\psi_{j}\rangle\langle\psi_{j}|
\end{equation}
be a density matrix in the sense of (\ref{trc1}); the vectors $\psi_{j}$ being
normalized we have
\[
\langle\psi_{k}|\widehat{\rho}|\psi_{k}\rangle=\sum_{j}\alpha_{j}|\langle
\psi_{j}|\psi_{k}\rangle|^{2}=\alpha_{k}%
\]
and hence, using definition (\ref{tr6}) of the trace,
\[
\operatorname{Tr}(\widehat{\rho})=\sum_{k}\alpha_{k}=1.
\]

\subsubsection{The spectral theorem}

An especially useful expansion of a density operator is obtained using
elementary functional analysis (the spectral resolution theorem). Recall that
each eigenvalue of a selfadjoint compact operator (except possibly zero) has
finite multiplicity (see any book on elementary functional analysis,
\textit{e.g.} Blanchard and Br\"{u}ning \cite{blabru} or Landsman \cite{land1}).

\begin{theorem}
\label{ThmB}A bounded linear self-adjoint operator $\widehat{\rho}$ on a
complex Hilbert space $\mathcal{H}$ is of trace class if and only if there
exists a sequence of real numbers $\lambda_{j}\geq0$ and an orthonormal basis
$(\psi_{j})$ of $\mathcal{H}$ such that for all $\psi\in\mathcal{H}$%
\begin{equation}
\widehat{\rho}\psi=\sum_{j}\lambda_{j}\langle\psi_{j}|\psi\rangle\psi_{j}
\label{speck1}%
\end{equation}
that is $\widehat{\rho}=\sum_{j}\lambda_{j}\widehat{\rho}_{j}$ where
$\widehat{\rho}_{j}$ is the orthogonal projection on the ray $\mathbb{C}%
\psi_{j}$. In particular, $\widehat{\rho}$ is a density matrix if and only if
$\lambda_{j}\geq0$ for every $j$ and $\sum_{j}\lambda_{j}=1$; the vector
$\psi_{j}$ is the eigenvector corresponding to the eigenvalue $\lambda_{j}$.
\end{theorem}

\begin{proof}
This is a classical result from the theory of compact self-adjoint operators
on a Hilbert space; see any introductory book on functional analysis, for
instance Blanchard and Br\"{u}ning \cite{blabru}. That $\psi_{j}$ \ is an
eigenvector corresponding to $\lambda_{j}$ is clear: since $\langle\psi
_{k}|\psi_{j}\rangle=\delta_{kj}$ we have
\[
\widehat{\rho}\psi_{j}=\sum_{k}\lambda_{k}\langle\psi_{k}|\psi_{j}\rangle
\psi_{k}=\lambda_{j}\psi_{j}.
\]

\end{proof}

An immediate consequence of Theorem \ref{ThmB} is that if $\widehat{\rho}$ is
a density matrix, then $\operatorname{Tr}(\widehat{\rho}^{2})\leq1$ with
equality if and only if $\widehat{\rho}$ represents a pure state: we have
\[
\widehat{\rho}^{2}=\left(  \sum\nolimits_{j}\lambda_{j}\widehat{\rho}%
_{j}\right)  ^{2}=\sum\nolimits_{j,k}\lambda_{j}\lambda_{k}\widehat{\rho}%
_{j}\widehat{\rho}_{k}=\sum\nolimits_{j}\lambda_{j}^{2}\widehat{\rho}_{j}%
\]
the second equality because $\widehat{\rho}_{j}\widehat{\rho}_{k}=0$ if $j\neq
k$ since $\psi_{j}$ and $\psi_{k}$ are then orthogonal, and $\widehat{\rho
}_{j}^{2}=\widehat{\rho}_{j}$. Since $\lambda_{j}^{2}\leq\lambda_{j}\leq1$ we
have%
\[
\operatorname{Tr}(\widehat{\rho}^{2})=\sum\nolimits_{j}\lambda_{j}^{2}\leq1.
\]
The equality $\sum\nolimits_{j}\lambda_{j}^{2}=1$ can only occur if all the
coefficients $\lambda_{j}$ are equal to zero, except one which is equal to
one. Thus $\operatorname{Tr}(\widehat{\rho}^{2})=1$ if and only if
$\widehat{\rho}$ represents a pure state. The number $\operatorname{Tr}%
(\widehat{\rho}^{2})$ is therefore called the \emph{purity }of the quantum
state represented by the density matrix $\widehat{\rho}$. Another way of
measuring the purity of a state is to use the von Neumann entropy
$S(\widehat{\rho})$. By definition:%
\begin{equation}
S(\widehat{\rho})=-\sum_{j}\lambda_{j}\ln\lambda_{j} \label{entropy}%
\end{equation}
(with the convention $0\ln0=0$). One often uses the suggestive notation%
\begin{equation}
S(\widehat{\rho})=-\operatorname{Tr}(\widehat{\rho}\ln\widehat{\rho})
\label{entropy2}%
\end{equation}
but one should then not forget that the right-hand side of this equality is
defined by (\ref{entropy}), and not the other way around! Notice that the von
Neumann entropy $S(\widehat{\rho})$ is zero if and only if $\widehat{\rho}$ is
a pure state.

\subsubsection{Functional properties}

Perhaps the most directly useful property of trace class operators (and hence
of density matrices) is Theorem \ref{ThmTC} below; it says that if we compose
a trace class operator with any bounded operator we obtain again a trace class operator.

We denote by $\mathcal{L}_{1}(\mathcal{H)}$ the set of all trace class
operators on the Hilbert space $\mathcal{H}$ and by $\mathcal{B}(\mathcal{H)}$
the algebra of bounded linear operators on $\mathcal{H}$. The following
important result justifies the trace formula (\ref{trc3}) for bounded observables:

\begin{theorem}
\label{ThmTC}The set $\mathcal{L}_{1}(\mathcal{H})$ of all trace class
operators on $\mathcal{H}$ is both a vector subspace of $\mathcal{B}%
(\mathcal{H)}$ and a two-sided ideal in $\mathcal{B}(\mathcal{H)}$: if
$\widehat{\rho}\in\mathcal{L}_{1}(\mathcal{H})$ and $\widehat{A}\in
\mathcal{B}(\mathcal{H)}$ then $\widehat{\rho}\widehat{A}\in\mathcal{L}%
_{1}(\mathcal{H})$ and $\widehat{A}\widehat{\rho}\in\mathcal{L}_{1}%
(\mathcal{H})$ and we have%
\begin{equation}
\operatorname{Tr}(\widehat{\rho}\widehat{A})=\operatorname{Tr}(\widehat{A}%
\widehat{\rho}). \label{trab1}%
\end{equation}
Formula (\ref{trab1}) applies in particular when $\widehat{\rho}$ is a density
matrix and $\widehat{A}$ a bounded quantum observable on $\mathcal{H}$.
\end{theorem}

\begin{proof}
That $\mathcal{L}_{1}(\mathcal{H})$ is a vector space is clear using formula
(\ref{trc4}): if $\sum_{j}|\langle\widehat{\rho}_{1}\phi_{j}|\chi_{j}%
\rangle|<\infty$ and $\sum_{j}|\langle\widehat{\rho}_{2}\phi_{j}|\chi
_{j}\rangle|<\infty$ then we have, using the triangle inequality,%
\[
\sum_{j}|\langle(\widehat{\rho}_{1}+\widehat{\rho}_{2})\phi_{j}|\chi
_{j}\rangle|\leq\sum_{j}|\langle\widehat{\rho}_{1}\phi_{j}|\chi_{j}%
\rangle|+\sum_{j}|\langle\widehat{\rho}_{2}\phi_{j}|\chi_{j}\rangle|<\infty
\]
so that $\widehat{\rho}_{1}+\widehat{\rho}_{2}\in\mathcal{L}_{1}(\mathcal{H}%
)$; that $\lambda\widehat{\rho}\in\mathcal{L}_{1}(\mathcal{H})$ if
$\widehat{\rho}\in\mathcal{L}_{1}(\mathcal{H})$ and $\lambda\in\mathbb{C}$ is
clear. Let us show that $\widehat{A}\widehat{\rho}\in\mathcal{L}%
_{1}(\mathcal{H})$ if $\widehat{\rho}\in\mathcal{L}_{1}(\mathcal{H})$ and
$\widehat{A}$ is a bounded operator on $\mathcal{H}$. Recall that the
boundedness of $\widehat{A}$ is equivalent to the existence of a number
$C\geq0$ such that $||\widehat{A}\psi||\leq C||\psi||$ for all $\psi
\in\mathcal{H}$. Let now $(\psi_{j})$ and $(\phi_{j})$ be two orthonormal
bases of $\mathcal{H}$; writing $\langle\widehat{A}\widehat{\rho}\psi_{j}%
|\phi_{j}\rangle=\langle\widehat{\rho}\psi_{j}|\widehat{A}^{\dagger}\phi
_{j}\rangle$ and applying Bessel's equality (\ref{Besseleq}) to $\langle
\widehat{\rho}\psi_{j}|\widehat{A}^{\dagger}\phi_{j}\rangle$ we get%
\begin{equation}
\langle\widehat{A}\widehat{\rho}\psi_{j}|\phi_{j}\rangle=\sum_{k}%
\langle\widehat{\rho}\psi_{j}|\phi_{k}\rangle\langle\widehat{A}^{\dagger}%
\phi_{j}|\psi_{k}\rangle^{\ast}; \label{Bessel}%
\end{equation}
using the Cauchy--Schwarz inequality we have, since $||\widehat{A}^{\dagger
}\phi_{j}||\leq C$,%
\[
|\langle\widehat{A}^{\dagger}\phi_{j}|\psi_{k}\rangle^{\ast}|\leq
||\widehat{A}^{\dagger}\phi_{j}||\,||\psi_{k}||\leq C
\]
and hence%
\[
|\langle\widehat{A}\widehat{\rho}\psi_{j}|\phi_{j}\rangle|\leq\sum_{k}%
|\langle\widehat{\rho}\psi_{j}|\phi_{k}\rangle\langle\widehat{A}^{\dagger}%
\phi_{j}|\psi_{k}\rangle^{\ast}|\leq C\sum_{k}|\langle\widehat{\rho}\psi
_{j}|\phi_{k}\rangle|.
\]
Summing this inequality with respect to the index $j$ yields, since
$\widehat{\rho}$ is of trace class
\[
|\sum_{j}\langle\widehat{A}\widehat{\rho}\psi_{j}|\phi_{j}\rangle|\leq
C\sum_{j,k}|\langle\widehat{\rho}\psi_{j}|\phi_{k}\rangle|<\infty
\]
hence $\widehat{A}\widehat{\rho}$ is of trace class as claimed. That
$\widehat{\rho}\widehat{A}$ also is of trace class is immediate noting that we
can write $\widehat{\rho}\widehat{A}=(\widehat{A}^{\dagger}\widehat{\rho
}^{\dagger})^{\dagger}$. There remains to prove the trace equality
(\ref{trab1}). Choosing $(\psi_{j})=(\phi_{j})$ the Bessel equality
(\ref{Bessel}) becomes%
\[
\langle\widehat{A}\widehat{\rho}\psi_{j}|\psi_{j}\rangle=\sum_{k}%
\langle\widehat{\rho}\psi_{j}|\psi_{k}\rangle\langle\widehat{A}^{\dagger}%
\psi_{j}|\psi_{k}\rangle^{\ast}%
\]
and, similarly,%
\begin{align*}
\langle\widehat{\rho}\widehat{A}\psi_{j}|\psi_{j}\rangle &  =\sum_{k}%
\langle\widehat{A}\psi_{j}|\psi_{k}\rangle\langle\widehat{\rho}^{\dagger}%
\psi_{j}|\psi_{k}\rangle^{\ast}\\
&  =\sum_{k}\langle\widehat{\rho}\psi_{k}|\psi_{j}\rangle\langle
\widehat{A}^{\dagger}\psi_{k}|\psi_{j}\rangle^{\ast}.
\end{align*}
Summing this equality over $j$ we get%
\begin{align*}
\operatorname{Tr}(\widehat{\rho}\widehat{A})  &  =\sum_{j}\langle
\widehat{A}\widehat{\rho}\psi_{j}|\psi_{j}\rangle=\sum_{j,k}\langle
\widehat{\rho}\psi_{j}|\psi_{k}\rangle\langle\widehat{A}^{\dagger}\psi
_{j}|\psi_{k}\rangle^{\ast}\\
\operatorname{Tr}(\widehat{A}\widehat{\rho})  &  =\sum_{j}\langle
\widehat{\rho}\widehat{A}\psi_{j}|\psi_{j}\rangle=\sum_{j,k}\langle
\widehat{\rho}\psi_{k}|\psi_{j}\rangle\langle\widehat{A}^{\dagger}\psi
_{k}|\psi_{j}\rangle^{\ast}%
\end{align*}
hence $\operatorname{Tr}(\widehat{\rho}\widehat{A})=\operatorname{Tr}%
(\widehat{A}\widehat{\rho})$ since the sums of both right-hand sides are identical.
\end{proof}

\subsection{Hilbert--Schmidt operators\label{sectionHS}}

We are following here de Gosson \cite{Birkbis}, \S 12.1, and Shubin
\cite{sh87}, Appendix 3.

\subsubsection{The trace norm}

An operator $\widehat{A}$ on a Hilbert space $\mathcal{H}$ is called a
\textit{Hilbert--Schmidt operator} if there exists an orthonormal basis
$(\psi_{j})$ of $\mathcal{H}$ such that%
\begin{equation}
\sum_{j}\langle\widehat{A}\psi_{j}|\widehat{A}\psi_{j}\rangle=\sum
_{j}||\widehat{A}\psi_{j}||^{2}<\infty. \label{HSdef}%
\end{equation}
In particular such an operator is bounded on $\mathcal{H}$. As in the case of
trace class operators one shows that if this property holds for one
orthonormal basis then it holds for all, and that the number $\sum
_{j}||\widehat{A}\psi_{j}||^{2}$ is independent of the choice of such a basis:
Let in fact $(\phi_{j})$ be a second orthonormal basis, and write
$\widehat{A}\psi_{j}=\sum_{k}\langle\phi_{j}|\widehat{A}\psi_{j}\rangle
\phi_{k}$. Then, using the Bessel equality (\ref{Besseleq})
\[
\sum_{j}||\widehat{A}\psi_{j}||^{2}=\sum_{j,k}|\langle\phi_{j}|\widehat{A}%
\psi_{j}\rangle|^{2}=\sum_{j,k}|\langle\widehat{A}^{\dag}\phi_{j}|\psi
_{j}|^{2}%
\]
that is, using again (\ref{Besseleq}),
\[
\sum_{j}||\widehat{A}\psi_{j}||^{2}=\sum_{k}||\widehat{A}^{\dag}\phi_{k}%
||^{2}<\infty.
\]
Observe that this equality shows that if we take $(\psi_{j})_{j}=(\phi
_{j})_{j}$ then $\sum_{k}||\widehat{A}^{\dag}\psi_{k}||_{\mathcal{H}}%
^{2}<\infty$ hence the adjoint $\widehat{A}^{\dag}$ of a Hilbert--Schmidt
operator is also a Hilbert--Schmidt operator; we may thus replace
$\widehat{A}$ by $\widehat{A}^{\dag}$ in the inequality above, which yields
$\sum_{k}||\widehat{A}^{\dag}\phi_{k}||^{2}<\infty$ as claimed.

Hilbert--Schmidt operators form a vector space $\mathcal{L}_{2}(\mathcal{H)}%
$\ and%
\begin{equation}
||\widehat{A}||_{\mathrm{HS}}=\left(
{\textstyle\sum\nolimits_{j}}
||\widehat{A}\psi_{j}||^{2}\right)  ^{1/2} \label{HSnorm}%
\end{equation}
defines a norm on this space ; this norm is associated with the scalar
product
\begin{equation}
\langle\widehat{A}|\widehat{B}\rangle_{\mathrm{HS}}=\operatorname{Tr}%
(\widehat{A}^{\dag}\widehat{B})=%
{\textstyle\sum\nolimits_{j}}
\langle\widehat{A}\psi_{j}|\widehat{B}\psi_{j}\rangle. \label{HSAB}%
\end{equation}

If\textit{ }$\widehat{A}$ and $\widehat{B}$ are Hilbert--Schmidt operators
then $\lambda\widehat{A}$ is trivially a Hilbert--Schmidt operator and
$||\lambda\widehat{A}||_{\text{HS}}=|\lambda|\,||\widehat{A}||_{\text{HS}}$
for every $\lambda\in\mathbb{C}$; on the other hand, by the triangle
inequality,%
\[
\sum_{j}||(\widehat{A}+\widehat{B})\psi_{j}||^{2}\leq\sum_{j}||\widehat{A}%
\psi_{j}||^{2}+\sum_{j}||\widehat{B}\psi_{j}||^{2}<\infty
\]
for every orthonormal basis $(\psi_{j})$ hence $\widehat{A}+\widehat{B}$ is
also a Hilbert--Schmidt operator and we have
\[
||\widehat{A}+\widehat{B}||_{\text{HS}}^{2}\leq||\widehat{A}||_{\text{HS}}%
^{2}+||\widehat{B}||_{\text{HS}}^{2}\leq(||\widehat{A}||_{\text{HS}%
}+||\widehat{B}||_{\text{HS}})^{2}%
\]
and hence%
\[
||\widehat{A}+\widehat{B}||_{\text{HS}}\leq||\widehat{A}||_{\text{HS}%
}+||\widehat{B}||_{\text{HS}}.
\]
Finally, $||\widehat{A}||_{\text{HS}}=0$ is equivalent to $\widehat{A}\psi
_{j}=0$ for every index $j$ that is to $\widehat{A}=0$.

The space $\mathcal{L}_{2}(\mathcal{H})$ is complete for that norm, and hence
a Banach space (it is actually even a Hilbert space when equipped with the
scalar product (\ref{HSAB}). In addition $\mathcal{L}_{2}(\mathcal{H)}$ is
closed under multiplication (and hence an algebra).

It turns out that the space $\mathcal{L}_{2}(\mathcal{H})$ is a two-sided deal
in the algebra of $\mathcal{B}(\mathcal{H})$ of bounded operators: if
$\widehat{A}\in\mathcal{L}_{2}(\mathcal{H})$ and $\widehat{B}\in
\mathcal{B}(\mathcal{H})$ then $\widehat{A}\widehat{B}\in\mathcal{L}%
_{2}(\mathcal{H})$ and $\widehat{B}\widehat{A}\in\mathcal{L}_{2}(\mathcal{H}%
)$. Let us show that $\widehat{B}\widehat{A}\in\mathcal{L}_{2}(\mathcal{H})$.
We have, denoting by $||\widehat{B}||$ the operator norm of $\widehat{B}$,
\[
||\widehat{B}\widehat{A}||_{\text{HS}}^{2}=\sum\nolimits_{j}||\widehat{B}%
\widehat{A}\psi_{j}||^{2}\leq||\widehat{B}||\left(  \sum\nolimits_{j}%
||\widehat{A}\psi_{j}||^{2}\right)  <\infty.
\]
Applying the same argument to $\widehat{A}\widehat{B}=(\widehat{B}^{\dag
}\widehat{A}^{\dag})^{\dag}$ shows that we have $\widehat{A}\widehat{B}%
\in\mathcal{L}_{2}(\mathcal{H})$ as well.

\subsubsection{Hilbert--Schmidt and trace class}

An essential property is that every trace class operator is the product of two
Hilbert--Schmidt operators (and hence itself a Hilbert--Schmidt operator). Let
us glorify this important statement as a theorem:

\begin{theorem}
\label{ThmHS} (i) A bounded operator $\widehat{A}$ on $\mathcal{H}$ is of
trace class if and only if it is the product of two Hilbert--Schmidt
operators: $\mathcal{L}_{1}(\mathcal{H})=(\mathcal{L}_{2}(\mathcal{H))}^{2}$.
(ii) A trace class operator $\widehat{A}$ on $\mathcal{H}$ is itself a
Hilbert--Schmidt operator: $\mathcal{L}_{1}(\mathcal{H})\subset\mathcal{L}%
_{2}(\mathcal{H})$.
\end{theorem}

\begin{proof}
In what follows $(\psi_{j})$ is an orthonormal basis in $\mathcal{H}$. (i)
Assume that $\widehat{A}=$ $\widehat{B}\widehat{C}$ where $\widehat{B}$ and
$\widehat{C}$ are both Hilbert--Schmidt operators. We have, using respectively
the triangle and the Cauchy--Schwarz inequalities,
\[
|\sum_{j}\langle\psi_{j}|\widehat{A}\psi_{j}\rangle|\leq\sum_{j}%
|\langle\widehat{B}^{\dag}\psi_{j}|\widehat{C}\psi_{j}\rangle|\leq\sum
_{j}||\widehat{B}^{\dag}\psi_{j}||\,||\widehat{C}\psi_{j}||;
\]
in view of the trivial inequality
\[
\sum_{j}||\widehat{B}^{\dag}\psi_{j}||\,||\widehat{C}\psi_{j}||\leq\frac{1}%
{2}\left(
{\textstyle\sum\nolimits_{j}}
||\widehat{B}^{\dag}\psi_{j}||^{2}+\,||\widehat{C}\psi_{j}||^{2}\right)
\]
we get, since $\widehat{B}$ and $\widehat{C}$ are both Hilbert--Schmidt
operators,
\[
|\sum_{j}\langle\psi_{j}|\widehat{A}\psi_{j}\rangle|\leq\frac{1}{2}\left(
{\textstyle\sum\nolimits_{j}}
||\widehat{B}^{\dag}\psi_{j}||^{2}+\,||\widehat{C}\psi_{j}||^{2}\right)
<\infty
\]
proving that $\widehat{A}$ is indeed of trace class. Assume, conversely, that
$\widehat{A}\in\mathcal{L}_{1}(\mathcal{H})$. In view of the polar
decomposition theorem there exists a unitary operator $\widehat{U}$ on
$\mathcal{H}$ such that $\widehat{A}=\widehat{U}(\widehat{A}^{\dag}%
\widehat{A})^{1/2}$. Setting $\widehat{B}=\widehat{U}(\widehat{A}^{\dag
}\widehat{A})^{1/4}$ and $\widehat{C}=(\widehat{A}^{\dag}\widehat{A})^{1/4}$
we have $\widehat{A}=\widehat{B}\widehat{C}$; let us show that $\widehat{C}$
and $\widehat{B}$ are Hilbert--Schmidt operators. We have%
\[
\sum_{j}|\langle\widehat{C}\psi_{j}|\widehat{C}\psi_{j}\rangle|=\sum
_{j}|\langle\widehat{C}^{\dag}\widehat{C}\psi_{j}|\psi_{j}\rangle|=\sum
_{j}|\langle(\widehat{A}^{\dag}\widehat{A})^{1/2}\psi_{j}|\psi_{j}%
\rangle|<\infty
\]
because $(\widehat{A}^{\dag}\widehat{A})^{1/2}=\widehat{U}^{\dag}\widehat{A}$
is of trace class (Theorem \ref{ThmTC}), hence $\widehat{C}\in\mathcal{L}%
_{2}(\mathcal{H)}$. It follows that $\widehat{B}=\widehat{U}\widehat{C}%
\in\mathcal{L}_{2}(\mathcal{H)}$ as well. (ii) We have seen that every trace
class operator is a product $\widehat{A}=\widehat{B}\widehat{C}$ of two
Hilbert--Schmidt operators. In view of the algebra property (i) of
$\mathcal{L}_{2}(\mathcal{H)}$ the operator $\widehat{A}$ is
itself\ Hilbert--Schmidt operator.
\end{proof}

\subsubsection{The case of $L^{2}(\mathbb{R}^{n})$}

Let us now specialize to the case where $\mathcal{H}=L^{2}(\mathbb{R}^{n})$.
In this case Hilbert--Schmidt operators are exactly those operators that have
a square integrable kernel; as a consequence a density matrix also has square
integrable kernel.

\begin{theorem}
\label{ThmKernel}Let $\widehat{A}$ be a bounded operator on $L^{2}%
(\mathbb{R}^{n})$. (i) It is a Hilbert--Schmidt operator if and only if there
exists a function $K\in L^{2}(\mathbb{R}^{n}\times\mathbb{R}^{n})$ such that%
\begin{equation}
\widehat{A}\psi(x)=\int K(x,y)\psi(y)d^{n}y. \label{akernel}%
\end{equation}
(ii) Every trace class operator (and hence every density matrix ) on
$L^{2}(\mathbb{R}^{n})$ can be represented by (\ref{akernel}) with kernel
$K\in L^{2}(\mathbb{R}^{n}\times\mathbb{R}^{n})$.
\end{theorem}

\begin{proof}
(i) \textit{The condition is necessary}. Let $(\psi_{j})$ be an orthonormal
basis in $L^{2}(\mathbb{R}^{n})$. Let $\widehat{A}$ $\in\mathcal{L}_{2}%
(L^{2}(\mathbb{R}^{n})\mathcal{)}$; we have $\widehat{A}\psi_{i}=\sum
_{j}\langle\psi_{j}|\widehat{A}\psi_{i}\rangle\psi_{j}$ and hence%
\[
\widehat{A}\psi=\sum_{i}\langle\psi_{i}|\psi\rangle\widehat{A}\psi_{i}%
=\sum_{i,j}\langle\psi_{i}|\psi\rangle\langle\psi_{j}|\widehat{A}\psi
_{i}\rangle\psi_{j}%
\]
which we can rewrite, using the definition
\[
\langle\psi_{i}|\psi\rangle=\int\psi_{i}^{\ast}(y)\psi(y)d^{n}y
\]
of the inner product as
\begin{align*}
\widehat{A}\psi(x)  &  =\sum_{i,j}\langle\psi_{i}|\psi\rangle\langle\psi
_{j}|\widehat{A}\psi_{i}\rangle\psi_{j}(x)\\
&  =\sum_{i,j}\langle\psi_{i}|\psi\rangle\langle\psi_{j}|\widehat{A}\psi
_{i}\rangle\psi_{j}(x)\\
&  =\sum_{i,j}\langle\psi_{j}|\widehat{A}\psi_{i}\rangle\int\psi_{j}%
(x)\psi_{i}^{\ast}(y)\psi(y)d^{n}y.
\end{align*}
This is now (\ref{akernel}) with
\[
K(x,y)=\sum_{i,j}\langle\psi_{j}|\widehat{A}\psi_{i}\rangle\psi_{j}(x)\psi
_{i}^{\ast}(y).
\]
Let us show that $K\in L^{2}(\mathbb{R}^{n}\times\mathbb{R}^{n})$. Remarking
that the tensor products $(\psi_{j}\otimes\psi_{i}^{\ast})$ form an
orthonormal basis in $L^{2}(\mathbb{R}^{n}\times\mathbb{R}^{n})$ we have%
\[
\int|K(x,y)|^{2}d^{n}xd^{n}y\leq\sum_{i,j}|\langle\psi_{j}|\widehat{A}\psi
_{i}\rangle|^{2};
\]
applying the Bessel equality (\ref{Besseleq}) to $|\langle\psi_{j}%
|\widehat{A}\psi_{i}\rangle|^{2}$ we get
\[
\sum_{i,j}|\langle\psi_{j}|\widehat{A}\psi_{i}\rangle|^{2}=\sum_{i}%
|\langle\widehat{A}\psi_{i}|\widehat{A}\psi_{i}\rangle|^{2}<\infty
\]
since $\widehat{A}$ is a Hilbert--Schmidt operator. It follows that $K\in
L^{2}(\mathbb{R}^{n}\times\mathbb{R}^{n})$ as claimed. \textit{The condition
is sufficient}. Assume that the kernel $K$ of $\widehat{A}\in\mathcal{B}%
(L^{2}(\mathbb{R}^{n}))$ belongs to $L^{2}(\mathbb{R}^{n}\times\mathbb{R}%
^{n})$. Since $(\psi_{j}\otimes\psi_{i}^{\ast})$ is an orthonormal basis in
$L^{2}(\mathbb{R}^{n}\times\mathbb{R}^{n})$ we can find numbers $c_{ij}$ such
that $\sum_{i,j}|c_{ij}|^{2}<\infty$ and
\[
K(x,y)=\sum_{i,j}c_{ij}\psi_{j}(x)\otimes\psi_{i}^{\ast}(y).
\]
Define now the operator $\widehat{A}$ by the equality (\ref{akernel}); we
have
\[
\widehat{A}\psi(x)=\sum_{i,j}c_{ij}\psi_{j}(x)\int\psi_{i}^{\ast}%
(y)\psi(y)d^{n}y=\sum_{i,j}c_{ij}\langle\psi_{i}|\psi\rangle\psi_{j}(x)
\]
and hence, since the basis $(\psi_{j})_{j}$ is orthonormal,
\[
\widehat{A}\psi_{k}(x)=\sum_{i,j}c_{ij}\langle\psi_{i}|\psi_{k}\rangle\psi
_{j}(x)=\sum_{j}c_{kj}\psi_{j}(x)
\]
so that
\[%
{\textstyle\sum\nolimits_{k}}
||\widehat{A}\psi_{k}||^{2}=\sum\nolimits_{j,k}|c_{kj}|^{2}<\infty
\]
and $\widehat{A}$ is thus a Hilbert--Schmidt operator. (ii) In view of
property the algebra property (see (i) in Theorem \ref{ThmHS}) a trace class
operator is \textit{a fortiori} a Hilbert--Schmidt operator. The claim follows
in view of the statement (i).
\end{proof}

\section{The Phase Space Picture}

From now on we assume that the Hilbert space $\mathcal{H}$ is $L^{2}%
(\mathbb{R}^{n})$, the space of complex-valued square integrable functions on
$\mathbb{R}^{n}$ (we are thus dealing with quantum systems with $n$ degrees of freedom).

\subsection{The Weyl correspondence}

\subsubsection{Weyl operators and symbols}

Let us first explain what we mean by a Weyl symbol. Recall that a function
$K(x,y)$ defined on $\mathbb{R}^{n}\times\mathbb{R}^{n}$ is the kernel of an
operator $\widehat{A}$ if we have%
\begin{equation}
\widehat{A}\psi(x)=\int K(x,y)\psi(y)d^{n}y \label{kernel1}%
\end{equation}
for all $\psi\in L^{2}(\mathbb{R}^{n})$. A deep theorem from functional
analysis (\textquotedblleft Schwartz's kernel theorem\textquotedblright) tells
us that every continuous linear operator from spaces of test functions to the
tempered distributions\footnote{In this context we use the word
\textquotedblleft distribution\textquotedblright\ in the sense of L.
Schwartz's \textquotedblleft generalized functions\textquotedblright.} can be
represented in this way, the integral in (\ref{kernel1}) being possibly
replaced by a distributional bracket. By definition, the Weyl symbol of the
operator $\widehat{A}$ is the function%
\begin{equation}
a(x,p)=\int e^{-\frac{i}{\hbar}py}K(x+\tfrac{1}{2}y,x-\tfrac{1}{2}y)d^{n}y;
\label{symbol1}%
\end{equation}
this formula is easily inverted using an inverse Fourier transform in $p$,
yielding the expression of the kernel in terms of the symbol:
\begin{equation}
K(x,p)=\left(  \tfrac{1}{2\pi\hbar}\right)  ^{n}\int e^{\frac{i}{\hbar}%
p(x-y)}a(\tfrac{1}{2}(x+y),p)d^{n}p. \label{symbol2}%
\end{equation}
These two formulas uniquely define the kernel and the symbol in terms of each
other, and imply the \textquotedblleft Weyl correspondence (or
transform)\textquotedblright, which expresses unambiguously the operator
$\widehat{A}$ in terms of the symbol $a$:
\begin{equation}
\widehat{A}\psi(x)=\left(  \tfrac{1}{2\pi\hbar}\right)  ^{n}%
{\displaystyle\iint}
e^{\frac{i}{\hbar}p(x-y)}a(\tfrac{1}{2}(x+y),p)\psi(y)d^{n}yd^{n}p;
\label{Weyl1}%
\end{equation}
one often writes $\widehat{A}=\operatorname*{Op}_{\mathrm{W}}(a)$, and this
notation is unambiguous because the symbol of $\widehat{A}$ is uniquely
determined by (\ref{symbol1}). Formula (\ref{Weyl1}) can be rewritten in
several different ways; one common expression is%
\begin{equation}
\widehat{A}\psi(x)=\left(  \tfrac{1}{\pi\hbar}\right)  ^{n}%
{\displaystyle\iint}
a(x_{0},p_{0})\widehat{\Pi}(x_{0},p_{0})\psi(x)d^{n}x_{0}d^{n}p_{0}
\label{Weyl2}%
\end{equation}
where $\widehat{\Pi}(x_{0},p_{0})$ is the reflection (or parity) operator
\cite{Birk,Birkbis,gowig}%
\begin{equation}
\widehat{\Pi}(x_{0},p_{0})\psi(x)=e^{\frac{2i}{\hbar}p_{0}(x-x_{0})}%
\psi(2x_{0}-x). \label{parity}%
\end{equation}

The usefulness of the Weyl correspondence in quantum mechanics comes from the
fact that it associates to real symbols self-adjoint operators. In fact, more
generally:%
\[
\widehat{A}=\operatorname*{Op}\nolimits^{\mathrm{W}}(a)\text{ \ }%
\Longrightarrow\text{ \ }\widehat{A}^{\dag}=\operatorname*{Op}%
\nolimits^{\mathrm{W}}(a^{\ast}).
\]
We refer to de Gosson \cite{Birk,Birkbis,gowig} for detailed discussions of
the Weyl correspondence from the point of view outlined above; Littlejohn's
well-cited paper \cite{Littlejohn} contains a very nice review of the topic
with applications to semiclassical approximations.

\subsubsection{Twisted products and convolutions}

Composing Weyl operators leads in a natural way to the notion of \emph{twisted
product}, which is essential in the theory of deformation quantization: assume
that the two operators $\widehat{A}=\operatorname*{Op}_{\mathrm{W}}(a)$ and
$\widehat{B}=\operatorname*{Op}_{\mathrm{W}}(b)$ can be composed, and set
$\widehat{A}\widehat{B}=\widehat{C}=\operatorname*{Op}_{\mathrm{W}}(c)$. Then
the Weyl symbol $c$ is given by the expression $c=a\times b$ where
\begin{equation}
(a\times b)(z)=\left(  \tfrac{1}{4\pi\hbar}\right)  ^{2n}%
{\displaystyle\iint}
e^{\frac{i}{2\hbar}\sigma(z^{\prime},z^{\prime\prime})}a(z+\tfrac{1}%
{2}z^{\prime})b(z-\tfrac{1}{2}z^{\prime\prime})d^{2n}z^{\prime}d^{2n}%
z^{\prime\prime}. \label{star1}%
\end{equation}
Setting $u=z+\tfrac{1}{2}z^{\prime}$ and $v=z-\tfrac{1}{2}z^{\prime\prime}$ we
have we have $d^{2n}z^{\prime}d^{2n}z^{\prime\prime}=4^{2n}d^{2n}ud^{2n}v$ so
this formula can be rewritten%
\begin{equation}
(a\times b)(z)=\left(  \tfrac{1}{\pi\hbar}\right)  ^{n}%
{\displaystyle\iint}
e^{\frac{2i}{\hbar}\partial\sigma(u,z,v)}a(u)b(v)d^{2n}ud^{2n}v \label{star2}%
\end{equation}
where $\partial\sigma(u,z,v)$ is the coboundary of $\sigma(u,v)$ viewed as a
1-chain:%
\begin{equation}
\partial\sigma(u,z,v)=\sigma(u,z)-\sigma(u,v)+\sigma(z,v). \label{cocycle}%
\end{equation}
For detailed proofs of these properties see de Gosson \cite{Birkbis}, Chapter
10. The function
\begin{equation}
c=a\times b \label{twiprod}%
\end{equation}
is called the twisted product\footnote{It was defined by J. von Neumann
following the work of Weyl.} of $a$ and $b$. Thus, by definition,%
\[
\operatorname*{Op}\nolimits_{\mathrm{W}}(a\times b)=\operatorname*{Op}%
\nolimits_{\mathrm{W}}(a)\operatorname*{Op}\nolimits_{\mathrm{W}}(b).
\]

To the twisted product is associated the \emph{twisted convolution} $a\#b$ of
two symbols; it is defined by
\[
a\#b=F_{\sigma}(F_{\sigma}a\times F_{\sigma}b)
\]
or, equivalently, by
\[
F_{\sigma}(a\#b)=F_{\sigma}a\times F_{\sigma}b
\]
where $F_{\sigma}$ is the symplectic Fourier transform. The equivalence of
both definitions is due to the fact that the symplectic Fourier transform is
its own inverse. Explicitly (de Gosson, \cite{gowig}, Section 11.1):%
\begin{equation}
a\#b(z)=\left(  \tfrac{1}{2\pi\hbar}\right)  ^{n}\int e^{\frac{i}{2\hbar
}\sigma(z,z^{\prime})}a_{\sigma}(z-z^{\prime})b_{\sigma}(z^{\prime}%
)d^{2}z^{\prime}; \label{twisthoulahop}%
\end{equation}
an equivalent statement is to say that the twisted symbol $c_{\sigma}$ of the
product $\widehat{C}=\widehat{A}\widehat{B}$ is given by%
\begin{equation}
c_{\sigma}(z)=\left(  \tfrac{1}{2\pi\hbar}\right)  ^{n}\int e^{\frac{i}%
{2\hbar}\sigma(z,z^{\prime})}a_{\sigma}(z-z^{\prime})b_{\sigma}(z^{\prime
})d^{2n}z^{\prime}. \label{cecomp}%
\end{equation}

\subsection{The Wigner function of a density matrix}

The essential point to understand now is that the Wigner function
$\rho=W_{\widehat{\rho}}(x,p)$ we are going to define below is (up to an
unimportant constant factor) the \emph{Weyl symbol} of the operator
$\widehat{\rho}$: the Wigner function is thus a \emph{dequantization} of
$\widehat{\rho}$, that is a phase space\ function obtained from this
operator\footnote{It is perhaps a little bit daring to speak about
\textquotedblleft dequantization\textquotedblright\ in this context since the
Wigner functon is not really a classical object.}. Also notice that it is the
first time Planck's constant appears in a quite explicit way; we could have
\textit{a priori} replaced $\hbar$ with any other real parameter $\eta$: this
change wouldn't have consequence for the involved mathematics (but it would of
course change the physics!). We will come back to this essential point later, but

\subsubsection{Definition of the Wigner function of a density
matrix\label{secwft}}

To a density matrix $\widehat{\rho}$ on $L^{2}(\mathbb{R}^{n})$ one associates
in standard texts its Wigner function (also called Wigner distribution). It is
the function $W_{\widehat{\rho}}$ of the variables $x=(x_{1},...,x_{n})$ and
of the conjugate momenta $p=(p_{1},...,p_{n})$ usually defined in physics
texts by%
\begin{equation}
W_{\widehat{\rho}}(x,p)=\left(  \tfrac{1}{\pi\hbar}\right)  ^{n}\int
e^{-\frac{2i}{\hbar}px^{\prime}}\langle x+x^{\prime}|\widehat{\rho
}|x-x^{\prime}\rangle d^{n}x^{\prime} \label{wigner1}%
\end{equation}
where $|x\rangle$ is an eigenstate of the operator $\widehat{x}=(\widehat{x}%
_{1},...,\widehat{x}_{n})$ (where $\widehat{x}_{j}$ = multiplication by
$x_{j}$). Performing the change of variables $x\longmapsto y=2x^{\prime}$ we
can rewrite this definition in the equivalent form
\begin{equation}
W_{\widehat{\rho}}(x,p)=\left(  \tfrac{1}{2\pi\hbar}\right)  ^{n}\int
e^{-\frac{i}{\hbar}py}\langle x+\tfrac{1}{2}y|\widehat{\rho}|x-\tfrac{1}%
{2}y\rangle d^{n}y; \label{wigner2}%
\end{equation}
this has some practical advantages when one uses the Wigner--Weyl--Moyal
formalism. In spite of their formal elegance, formulas (\ref{wigner1}),
(\ref{wigner2}) are at first sight somewhat obscure and need to be clarified,
especially if one wants to work analytically with them. Assume first that
$\widehat{\rho}$ represents a pure state: $\widehat{\rho}=\widehat{\rho}%
_{\psi}=|\psi\rangle\langle\psi|$ where $\psi\in L^{2}(\mathbb{R}^{n})$ is
normalized. We get, using the relations $\psi(x)=\langle x|\psi\rangle$ and
$\psi^{\ast}(x)=\langle\psi|x\rangle$,
\begin{align*}
\langle x+x^{\prime}|\widehat{\rho}|x-x^{\prime}\rangle &  =\langle
x+x^{\prime}|\psi\rangle\langle\psi|x-x^{\prime}\rangle\\
&  =\psi(x+x^{\prime})\psi^{\ast}(x-x^{\prime})
\end{align*}
and hence $W_{\widehat{\rho}_{\psi}}(x,p)=W\psi(x,p)$ where
\begin{equation}
W\psi(x,p)=\left(  \tfrac{1}{\pi\hbar}\right)  ^{n}\int e^{-\frac{2i}{\hbar
}px^{\prime}}\psi(x+x^{\prime})\psi^{\ast}(x-x^{\prime})d^{n}x^{\prime}
\label{wigner3}%
\end{equation}
is the usual Wigner function (or Wigner distribution, or Wigner transform) of
$\psi\in L^{2}(\mathbb{R}^{n})$ (see Wigner \cite{Wigner}, Hillery \textit{et
al}. \cite{hilleryetal}; for the mathematical theory \cite{Birk,Birkbis,gowig}%
); equivalently%
\begin{equation}
W\psi(x,p)=\left(  \tfrac{1}{2\pi\hbar}\right)  ^{n}\int e^{-\frac{i}{\hbar
}py}\psi(x+\tfrac{1}{2}y)\psi^{\ast}(x-\tfrac{1}{2}y)d^{n}y. \label{wigner4}%
\end{equation}
In the general case, where $\widehat{\rho}=\sum_{j}\alpha_{j}|\psi_{j}%
\rangle\langle\psi_{j}|$ is a convex sum of operators of the type above one
immediately gets, by linearity, the expression%
\begin{equation}
W_{\widehat{\rho}}(x,p)=\sum_{j}\alpha_{j}W\psi_{j}(x,p). \label{wigner5}%
\end{equation}

A very important result we will prove later on (Theorem \ref{ThmWL2}), but use
immediately, is the following:

\begin{quote}
The Wigner function of a mixed state is square integrable: $W_{\widehat{\rho}%
}\in L^{2}(\mathbb{R}^{2n})$.
\end{quote}

One also often uses the cross-Wigner transform of a pair of square integrable
functions. It is given by
\begin{equation}
W(\psi,\phi)(x,p)=\left(  \tfrac{1}{2\pi\hbar}\right)  ^{n}\int e^{-\frac
{i}{\hbar}py}\psi(x+\tfrac{1}{2}y)\phi^{\ast}(x-\tfrac{1}{2}y)d^{n}y.
\label{crosswigner}%
\end{equation}
It naturally appears as an interference term when calculating the Wigner
function of a sum; in fact, using definition (\ref{wigner4}) one immediately
checks that%
\begin{equation}
W(\psi+\phi)=W\psi+W\phi+2\operatorname{Re}W(\psi,\phi). \label{winterf}%
\end{equation}
Notice that $W(\psi,\phi)$ is in general a complex number and that
\begin{equation}
W(\psi,\phi)^{\ast}=W(\phi,\psi). \label{wstar}%
\end{equation}
The cross-Wigner function has many applications; in particular it allows to
reformulate the notion of weak-value as an interference between the past and
the future in the time-symmetric approach to quantum mechanics (see de Gosson
and de Gosson \cite{Charlyne1}).

\subsubsection{The Weyl symbol of a density matrix}

In the case of density matrices we have:

\begin{theorem}
\label{Thm1}Let $\widehat{\rho}$ be a density matrix on $L^{2}(\mathbb{R}%
^{n})$:
\[
\widehat{\rho}=\sum_{j}\alpha_{j}|\psi_{j}\rangle\langle\psi_{j}|\text{ \ with
}\alpha_{j}\geq0\text{\ and }\sum_{j}\alpha_{j}=1.
\]
The Weyl symbol of $\widehat{\rho}$ is $a=(2\pi\hbar)^{n}\rho$ where
\begin{equation}
\rho(x,p)=W_{\widehat{\rho}}(x,p)=\sum_{j}\alpha_{j}W\psi_{j}(x,p)
\label{Weyl3}%
\end{equation}
is the Wigner function of $\widehat{\rho}$.
\end{theorem}

\begin{proof}
The action of the projection $\widehat{\rho}_{j}=$ $|\psi_{j}\rangle
\langle\psi_{j}|$ on a vector $\psi\in L^{2}(\mathbb{R}^{n})$ is given by
\[
\widehat{\rho}_{j}\psi(x)=\langle\psi_{j}|\psi\rangle\psi_{j}(x)=\int\psi
_{j}^{\ast}(y)\psi(y)\psi_{j}(x)d^{n}y
\]
hence the kernel of $\widehat{\rho}_{j}$ is the function
\[
K_{j}(x,y)=\psi_{j}(x)\psi_{j}^{\ast}(y).
\]
It follows, using formula (\ref{symbol1}), that the Weyl symbol of
$\widehat{\rho}_{j}$ is%
\begin{align*}
a_{j}(x,p)  &  =\int e^{-\frac{i}{\hbar}py}\psi_{j}(x+\tfrac{1}{2}y)\psi
_{j}^{\ast}(x-\tfrac{1}{2}y)d^{n}y\\
&  =(2\pi\hbar)^{n}W\psi_{j}(x,p).
\end{align*}
Formula (\ref{Weyl3}) follows by linearity.
\end{proof}

\subsubsection{Statistical interpretation of the Wigner function}

The importance of the Wigner function of a density matrix comes from the fact
that we can use it as a substitute for an ordinary probability density for
calculating averages (it is precisely for this purpose Wigner introduced his
eponymous transform in \cite{Wigner}). For all $\psi\in L^{2}(\mathbb{R}^{n})$
such that\footnote{The conditions $\psi\in L^{1}(\mathbb{R}^{n})$ and
$F\psi\in L^{1}(\mathbb{R}^{n})$ are necessary to ensure the convergence of
the $x$ and $p$ integrals. The condition $\psi\in L^{2}(\mathbb{R}^{n})$ is
not sufficient for the marginal properties to hold (see \textit{e.g.}
Daubechies \cite{daube})} $\psi\in L^{1}(\mathbb{R}^{n})$ and $F\psi\in
L^{1}(\mathbb{R}^{n})$ the marginal properties are
\begin{equation}
\int W\psi(x,p)d^{n}p=|\psi(x)|^{2}\text{ \ \textit{,} \ }\int W\psi
(x,p)d^{n}x=|F\psi(p)|^{2} \label{marginal1}%
\end{equation}
and hence, in particular,
\begin{equation}
\iint W\psi(x,p)d^{n}pd^{n}x=1\text{ \ \textit{if} \ }||\psi||=1.
\label{norm1}%
\end{equation}
In the second equality (\ref{marginal1})
\begin{equation}
F\psi(p)=\left(  \tfrac{1}{2\pi\hbar}\right)  ^{n/2}\int e^{-\frac{i}{\hbar
}px}\psi(x)d^{n}x \label{hft}%
\end{equation}
is the $\hbar$-Fourier transform of $\psi$. One should be aware of the fact
that while $W\psi$ is always real (and hence so is $\rho=W_{\widehat{\rho}}$)
as can be easily checked by taking the complex conjugates of both sides of the
equality (\ref{wigner4}), it takes negative values for all $\psi$ which are
not Gaussian functions. (This is the celebrated \textquotedblleft Hudson
theorem\textquotedblright\ \cite{Hudson}; also see Janssen \cite{Janssen} for
the multidimensional case.) A \textit{caveat}: this result is only true for
the Wigner function $W\psi$ of a single function $\psi$; the case of a general
distribution $\rho=\sum_{j}\alpha_{j}W\psi_{j}$ is much subtler, and will be
discussed later.

Let us introduce the following terminology: we call an observable
$\widehat{A}$ a \textquotedblleft good observable\textquotedblright\ for the
density matrix $\widehat{\rho}$ if its Weyl symbol $a$ (\textit{i.e.} the
corresponding classical observable) satisfies $a\rho\in L^{1}(\mathbb{R}%
^{2n})$, that is%
\begin{equation}
\int|a(z)\rho(z)|d^{2n}z<\infty\label{good1}%
\end{equation}
($\rho$ the Wigner function of $\widehat{\rho}$; we are using the shorthand
notation $z=(x,p)$, $d^{2n}z=d^{n}xd^{n}p$). We assume in addition that $a$ is
real so that $\widehat{A}$ is Hermitian. Notice that \textquotedblleft
goodness\textquotedblright\ is guaranteed if the symbol $a$ is square
integrable, because the Cauchy--Schwarz inequality then implies that%
\[
\left(  \int|a(z)\rho(z)|d^{2n}z\right)  ^{2}\leq\int\rho(z)^{2}d^{2n}%
z\int|a(z)|^{2}d^{2n}z<\infty
\]
since $\rho$ is square integrable (as mentioned above, see Theorem
\ref{ThmWL2}).

\begin{theorem}
\label{Thm2}Let $\widehat{\rho}$ be a density matrix on $L^{2}(\mathbb{R}%
^{n})$ and $\rho$ its Wigner function. The average value of every good
observable $\widehat{A}$ with respect to $\widehat{\rho}$ is then finite and
given by the formula%
\begin{equation}
\langle\widehat{A}\rangle_{\widehat{\rho}}=\int a(z)\rho(z)d^{2n}z.
\label{average2}%
\end{equation}

\end{theorem}

\begin{proof}
By linearity it suffices to consider the case where $\widehat{\rho}%
=|\psi\rangle\langle\psi|$ so that $\rho=W\psi$; this reduces the proof of
formula (\ref{average2}) to that of the simpler equality%
\begin{equation}
\langle\widehat{A}\rangle_{\psi}=\int a(x,p)W\psi(x,p)d^{n}xd^{n}p.
\label{average3}%
\end{equation}
Replacing in the equality above $W\psi(x,p)$ by its integral expression
(\ref{wigner4}) yields%
\[
\langle\widehat{A}\rangle_{\psi}=\left(  \tfrac{1}{2\pi\hbar}\right)
^{n}\iint a(x,p)\left(  \int e^{-\frac{i}{\hbar}py}\psi(x+\tfrac{1}{2}%
y)\psi^{\ast}(x-\tfrac{1}{2}y)d^{n}y\right)  d^{n}xd^{n}p.
\]
Since we assume that the \textquotedblleft goodness\textquotedblright%
\ assumption (\ref{good1}) is satisfied, we can use Fubini's theorem and
rewrite this equality as a double integral:%
\[
\langle\widehat{A}\rangle_{\psi}=\left(  \tfrac{1}{2\pi\hbar}\right)  ^{n}%
{\displaystyle\iint}
a(x,p)e^{-\frac{i}{\hbar}py}\psi(x+\tfrac{1}{2}y)\psi^{\ast}(x-\tfrac{1}%
{2}y)d^{n}yd^{n}xd^{n}p.
\]
Let us perform the change of variables $x^{\prime}=x+\tfrac{1}{2}y$ and
$y^{\prime}=x-\tfrac{1}{2}y$; we have $x=\frac{1}{2}(x^{\prime}+y^{\prime})$
and $y=x^{\prime}-y^{\prime}$ and hence, using definition (\ref{Weyl1}) of the
Weyl operator $\widehat{A}$,
\begin{align*}
\langle\widehat{A}\rangle_{\psi}  &  =\left(  \tfrac{1}{2\pi\hbar}\right)
^{n}%
{\displaystyle\iint}
e^{-\frac{i}{\hbar}p(x^{\prime}-y^{\prime})}a(\tfrac{1}{2}(x^{\prime
}+y^{\prime}),p)\psi(x^{\prime})\psi^{\ast}(y^{\prime})d^{n}y^{\prime}%
d^{n}x^{\prime}d^{n}p\\
&  =\left(  \tfrac{1}{2\pi\hbar}\right)  ^{n}\int\left(  \iint e^{-\frac
{i}{\hbar}p(x^{\prime}-y^{\prime})}a(\tfrac{1}{2}(x^{\prime}+y^{\prime
}),p)\psi^{\ast}(y^{\prime})d^{n}y^{\prime}d^{n}p\right)  \psi(x^{\prime
})d^{n}x^{\prime}%
\end{align*}
and hence%
\[
\langle\widehat{A}\rangle_{\psi}=\int\widehat{A}\psi^{\ast}(x^{\prime}%
)\psi(x^{\prime})d^{n}x^{\prime}=\langle\widehat{A}\psi|\psi\rangle
\]
which we set out to prove.
\end{proof}

We remark that the identity (\ref{average3}) can be extended to the
cross-Wigner function (\ref{crosswigner}); in fact, adapting the proof of
(\ref{average3}) one sees that if $\psi$ and $\phi$ are square integrable,
then%
\begin{equation}
\langle\psi|\widehat{A}\phi\rangle=\int a(x,p)W(\psi,\phi)(x,p)d^{n}xd^{n}p.
\label{average4}%
\end{equation}

\subsection{The displacement operator and the ambiguity function}

In this subsection we review a few properties of the Wigner function which are
perhaps not all so well-known in quantum mechanics; these properties are
important because they give an insight into some of the subtleties of the
Weyl--Wigner--Moyal transform. We also define a related transform, the
ambiguity function.

\subsubsection{Redefinition of the Wigner function}

Recall that the reflection operator (\ref{parity}) is explicitly given by the
formula
\[
\widehat{\Pi}(x_{0},p_{0})\psi(x)=e^{\frac{2i}{\hbar}p_{0}(x-x_{0})}%
\psi(2x_{0}-x).
\]
It can be used to define the Wigner function in a very concise way. In fact:
for every $\psi\in L^{2}(\mathbb{R}^{n})$ we have%
\begin{equation}
W\psi(x_{0},p_{0})=\left(  \tfrac{1}{\pi\hbar}\right)  ^{n}\langle
\psi|\widehat{\Pi}(x_{0},p_{0})\psi\rangle. \label{wgr}%
\end{equation}
This is easy to verify: we have, by definition (\ref{parity}) of
$\widehat{\Pi}(x_{0},p_{0})$,
\[
\langle\psi|\widehat{\Pi}(x_{0},p_{0})\psi\rangle=\int e^{\frac{2i}{\hbar
}p_{0}(x-x_{0})}\psi(2x_{0}-x)\psi^{\ast}(x)d^{n}x;
\]
setting $y=2(x_{0}-x)$ we have $x=x_{0}-\frac{1}{2}y$, $2x_{0}-x=x_{0}%
+\frac{1}{2}y$, and $d^{n}x=2^{-n}d^{n}y$ hence%
\[
\langle\psi|\widehat{\Pi}(x_{0},p_{0})\psi\rangle=2^{-n}\int e^{-\frac
{i}{\hbar}p_{0}y}\psi(x_{0}+\tfrac{1}{2}y)\psi^{\ast}(x_{0}-\tfrac{1}%
{2}y)d^{n}y
\]
which proves (\ref{wgr}), taking definition (\ref{wigner4}) of the Wigner
function into account. Formula (\ref{wgr}) shows quite explicitly that, up to
the factor $(\pi\hbar)^{-n}$, the Wigner function is the probability amplitude
for the state $|\psi\rangle$ to be in the state $|\widehat{\Pi}(x_{0}%
,p_{0})\psi\rangle$; this was actually already observed by Grossmann
\cite{Grossmann} and Royer \cite{Royer} in the mid 1970s.

\subsubsection{The Moyal identity\label{secmoyal}}

An important equality satisfied by the Wigner function is \emph{Moyal's
identity}\footnote{It is somtimes also called the \textquotedblleft
orthogonality relation\textquotedblright.}
\begin{equation}
\int W\psi(z)W\phi(z)d^{2n}z=\left(  \tfrac{1}{2\pi\hslash}\right)
^{n}|\langle\psi|\phi\rangle|^{2} \label{Moyal}%
\end{equation}
(see \cite{Birkbis,gowig} for a proof); it is valid for all square-integrable
functions $\psi$ and $\phi$ on $\mathbb{R}^{n}$. In particular:%
\begin{equation}
\int W\psi(z)^{2}d^{2n}z=\left(  \tfrac{1}{2\pi\hslash}\right)  ^{n}%
||\psi||^{4}. \label{Moyalpart}%
\end{equation}
This formula implies the following interesting fact which is not immediately
obvious: consider the spectral decomposition (\ref{speck1}) of a density
operator in Theorem \ref{ThmB}:%
\[
\widehat{\rho}\psi=\sum_{j}\lambda_{j}\langle\psi_{j}|\psi\rangle\psi_{j}%
\]
here the $\lambda_{j}$ are the eigenvalues of $\widehat{\rho}$ and the
corresponding eigenvectors $\psi_{j}$ form an orthonormal basis of
$\mathcal{H}$. When $\mathcal{H}=L^{2}(\mathbb{R}^{n})$ the corresponding
Wigner function is therefore%
\begin{equation}
\widehat{\rho}\psi=\sum_{j}\lambda_{j}W\psi_{j}. \label{lambdawigner}%
\end{equation}
It follows from Moyal's identity (\ref{Moyal}) that the $W\psi_{j}$ form an
orthonormal system of vectors in the Hilbert space $L^{2}(\mathbb{R}^{2n})$
(but not a basis as is easily seen by \textquotedblleft dimension
count\textquotedblright).

As a consequence of the Moyal identity we prove the fact, mentioned in Section
\ref{secwft}, that the Wigner function of a density matrix is square integrable.

\begin{theorem}
\label{ThmWL2}Let $\{(\psi_{j},\alpha_{j})\}$ be a mixed state ($\psi_{j}\in
L^{2}(\mathbb{R}^{n})$, $\alpha_{j}\geq0$, $\sum_{j}\alpha_{j}=1$). The Wigner
function $\rho=W_{\widehat{\rho}}$ is square integrable: $\rho\in
L^{2}(\mathbb{R}^{2n})$.
\end{theorem}

\begin{proof}
Since $L^{2}(\mathbb{R}^{2n})$ is a vector space it is sufficient to consider
the pure case, that is to prove that $W\psi\in L^{2}(\mathbb{R}^{2n})$ if
$\psi\in L^{2}(\mathbb{R}^{2n})$. But this immediately follows from Moyal's
identity (\ref{Moyalpart}).
\end{proof}

The Moyal identity can be extended to the cross-Wigner function
(\ref{crosswigner}); recall that for $\psi,\phi\in L^{2}(\mathbb{R}^{n})$ it
is defined by
\[
W(\psi,\phi)(x,p)=\left(  \tfrac{1}{2\pi\hbar}\right)  ^{n}\int e^{-\frac
{i}{\hbar}py}\psi(x+\tfrac{1}{2}y)\phi^{\ast}(x-\tfrac{1}{2}y)d^{n}y.
\]
In fact, for all $\psi,\psi^{\prime},\phi,\phi^{\prime}\in L^{2}%
(\mathbb{R}^{n})$ we have%
\begin{equation}
\int W(\psi,\psi^{\prime})^{\ast}(z)W(\phi,\phi^{\prime})(z)d^{2n}z=\left(
\tfrac{1}{2\pi\hslash}\right)  ^{n}\langle\psi|\phi\rangle\langle\psi^{\prime
}|\phi^{\prime}\rangle^{\ast} \label{CrossMoyal}%
\end{equation}
(see for instance de Gosson \cite{Birkbis,gowig}). Denoting by $\langle
\langle\cdot|\cdot\rangle\rangle$ the inner product on $L^{2}(\mathbb{R}%
^{2n})$ this identity can be written in the form%
\begin{equation}
\langle\langle W(\psi,\psi^{\prime})|W(\phi,\phi^{\prime})\rangle
\rangle=\left(  \tfrac{1}{2\pi\hslash}\right)  ^{n}\langle\psi|\phi
\rangle\langle\psi^{\prime}|\phi^{\prime}\rangle^{\ast} \label{moyal2}%
\end{equation}
In particular%
\begin{equation}
\int|W(\psi,\psi^{\prime})(z)|^{2}d^{2n}z=\left(  \tfrac{1}{2\pi\hslash
}\right)  ^{n}||\psi||^{2}\,||\psi^{\prime}||^{2}. \label{CrossMoyalbis}%
\end{equation}

An important remark: one can prove \cite{Birkbis,gowig}, using this
generalized Moyal identity, that if vectors $\psi_{j}$ form an orthonormal
basis of $L^{2}(\mathbb{R}^{n})$ then the vectors $(2\pi\hslash)^{n/2}%
W(\psi_{j},\psi_{k})$ form an orthonormal basis of the space $L^{2}%
(\mathbb{R}^{2n})$ (that these vectors are orthonormal is clear from
(\ref{moyal2})).

\subsubsection{The ambiguity function}

A transform closely related to the Wigner function and well-known from signal
analysis (especially radar theory) is the \emph{ambiguity function}
$\operatorname*{Amb}\psi$ (it is also called the \textquotedblleft
auto-correlation function\textquotedblright). It can be introduced in several
equivalent ways; we begin by defining it explicitly by a formula: for $\psi\in
L^{2}(\mathbb{R}^{n})$
\begin{equation}
\operatorname*{Amb}\psi(x,p)=\left(  \tfrac{1}{2\pi\hbar}\right)  ^{n}\int
e^{-\frac{i}{\hbar}py}\psi(y+\tfrac{1}{2}x)\psi^{\ast}(y-\tfrac{1}{2}x)d^{n}y.
\label{defamb1}%
\end{equation}
Comparing with the definition
\[
W\psi(x,p)=\left(  \tfrac{.1}{2\pi\hbar}\right)  ^{n}\int e^{-\frac{i}{\hbar
}py}\psi(x+\tfrac{1}{2}y)\psi^{\ast}(x-\tfrac{1}{2}y)d^{n}y
\]
of the Wigner function one cannot help being surprised by the similarity of
both definitions. In fact, it is easy to show by performing an elementary
change of variables that if $\psi$ is an even function (that is $\psi
(-x)=\psi(x)$ for all $x\in\mathbb{R}^{n}$) then $W\psi$ and
$\operatorname*{Amb}\psi$ are related by%
\begin{equation}
\operatorname*{Amb}\psi(x,p)=2^{-n}W\psi(\tfrac{1}{2}x,\tfrac{1}{2}p).
\label{ambeven}%
\end{equation}
There are two complementary \textquotedblleft natural\textquotedblright\ ways
to define the ambiguity function. The first is to observe that the Wigner
function and the ambiguity function are symplectic Fourier transforms of each
other:
\begin{equation}
\operatorname*{Amb}\psi=F_{\sigma}W\psi\text{ \ and \ }W\psi=F_{\sigma
}\operatorname*{Amb}\psi; \label{ambsft}%
\end{equation}
they are of course equivalent since $F_{\sigma}^{-1}=F_{\sigma}$. For a proof,
see de Gosson \cite{Birkbis,gowig}. There is still another way to define the
ambiguity function. Let $\widehat{D}(z_{0})=\widehat{D}(x_{0},p_{0})$ be the
\emph{Weyl displacement operator} (it is also called the Glauber--Sudarshan
displacement operator, or the Heisenberg operator, or the Heisenberg--Weyl
operator). It is defined by%
\begin{equation}
\widehat{D}(z_{0})\psi(x)=e^{\frac{i}{\hbar}(p_{0}x-\frac{1}{2}p_{0}x_{0}%
)}\psi(x-x_{0}). \label{HW1}%
\end{equation}
This operator is the time-one propagator for the Schr\"{o}dinger equation
associated with the classical translation Hamiltonian $\sigma(z,z_{0}%
)=x_{0}p-p_{0}x$ (see the discussions in de Gosson
\cite{Birk,Birkbis,SPRINGER} and Littlejohn \cite{Littlejohn}); this
observation motivates the notation%
\begin{equation}
\widehat{D}(z_{0})=e^{-\frac{i}{\hbar}\sigma(\widehat{z},z_{0})}=e^{-\frac
{i}{\hbar}(x_{0}\widehat{p}-p_{0}\widehat{x})} \label{dquant}%
\end{equation}
often found in the literature. We are using here the coordinate expression of
the displacement operator; we leave it to the reader as an exercise to check
that $\widehat{D}(z_{0})$ coincides with the operator
\[
D(\alpha)=\exp\left[  \frac{i}{\hbar}(\alpha a^{\dag}-\alpha^{\ast}a)\right]
\]
commonly used in quantum optics ($a$ and $a^{\dag}$ are the annihilation and
creation operators; see Poto\v{c}ek \cite{potocek} for a discussion of these
notational issues). The displacement operator is related to the reflection
operator $\widehat{\Pi}(z_{0})$ by the simple formula%
\begin{equation}
\widehat{\Pi}(z_{0})=\widehat{D}(z_{0})\Pi\widehat{D}(z_{0})^{\dag}
\label{deltapi}%
\end{equation}
where $\Pi$ is the parity operator $\Pi\psi(x)=\psi(-x)$. That the operators
$\widehat{D}(z_{0})$ correspond to translations in phase space quantum
mechanics is illustrated by the following important relation satisfied by the
Wigner transform:%
\begin{equation}
W(\widehat{D}(z_{0})\psi)(z)=W\psi(z-z_{0}) \label{transwig}%
\end{equation}
(it is easily proven by a direct computation, see de Gosson
\cite{Birk,Birkbis,gowig}, Littlejohn \cite{Littlejohn}). Using the
displacement operator, the ambiguity function is given by%
\begin{equation}
\operatorname*{Amb}\psi(z_{0})=\left(  \tfrac{1}{2\pi\hbar}\right)
^{n}\langle\widehat{D}(z_{0})\psi|\psi\rangle; \label{ambhw}%
\end{equation}
one verifies by a direct calculation using (\ref{HW1}) that one recovers the
first analytical definition (\ref{defamb1}).

The displacement operators play a very important role, not only in quantum
mechanics, but also in related disciplines such as harmonic analysis, signal
theory, and time-frequency analysis. They can be viewed as a representation of
the canonical commutation relations (the Schr\"{o}dinger representation of the
Heisenberg group); this is related to the fact these operators satisfy%
\begin{equation}
\widehat{D}(z_{0})\widehat{D}(z_{1})=e^{\frac{i}{\hslash}\sigma(z_{0},z_{1}%
)}\widehat{D}(z_{1})\widehat{D}(z_{0}) \label{commhw}%
\end{equation}
and also
\begin{equation}
\widehat{D}(z_{0}+z_{1})=e^{-\frac{i}{2\hslash}\sigma(z_{0},z_{1})}%
\widehat{D}(z_{0})\widehat{D}(z_{1}). \label{commhwbis}%
\end{equation}
The second formula shows that the displacement operators form a projective
representation of the phase space translation group. In addition to being used
to define the ambiguity function, the displacement operators allow one to
define Weyl operators in terms of their \textquotedblleft twisted
symbol\textquotedblright\ (sometimes also called \textquotedblleft covariant
symbol\textquotedblright), which is by definition the symplectic Fourier
transform%
\begin{equation}
a_{\sigma}(z)=F_{\sigma}a(z) \label{covariantsymbol}%
\end{equation}
of the ordinary symbol $a$. Let in fact $\widehat{A}=\operatorname*{Op}%
_{\mathrm{W}}(a)$, that is
\[
\widehat{A}=\left(  \tfrac{1}{\pi\hbar}\right)  ^{n}\int a(z_{0})\widehat{\Pi
}(z_{0})d^{2n}z_{0}%
\]
(formula \ref{Weyl2}). Using the displacement operator $\widehat{D}(z_{0})$ in
place of the reflection operator $\widehat{\Pi}(z_{0})$ we have%
\begin{equation}
\widehat{A}=\left(  \tfrac{1}{2\pi\hbar}\right)  ^{n}\int a_{\sigma}%
(z_{0})\widehat{D}(z_{0})d^{2n}z_{0} \label{atzo}%
\end{equation}
(see \cite{Birk,Birkbis,gowig,Littlejohn}). This formula has many
applications; it is essential in the study of the positivity properties of
trace class operators as we will see in a moment. Notice that formula
(\ref{atzo}) is Weyl's original definition \cite{Weyl} in disguise: making the
change of variables $z_{0}\longmapsto-Jz_{0}$ in this formula one gets, noting
that $a_{\sigma}(-Jz_{0})=Fa(z_{0})$ and $\widehat{D}(-Jz_{0})=e^{-\frac
{i}{\hbar}(x_{0}\widehat{x}+p_{0}\widehat{p})}e^{-\frac{i}{\hbar}%
(x_{0}\widehat{p}-p_{0}\widehat{x})}$
\[
\widehat{A}=\left(  \tfrac{1}{2\pi\hbar}\right)  ^{n}\iint Fa(x,p)e^{\frac
{i}{\hbar}(x\widehat{x}+p\widehat{p})}d^{n}pd^{n}x
\]
which is the formula originally proposed by Weyl \cite{Weyl}, in analogy with
the Fourier inversion formula (see the discussion in de Gosson
\cite{SPRINGER,gowig}).

\subsection{Calculating traces: rigorous theory\label{secalcul}}

\subsubsection{Kernels and symbols}

It is tempting to redefine the trace of a Weyl operator $\widehat{A}%
=\operatorname*{Op}\nolimits_{\mathrm{W}}(a)$ by the formula
\begin{equation}
\operatorname*{Tr}(\widehat{A})=\left(  \tfrac{1}{2\pi\hbar}\right)  ^{n}\int
a(z)d^{2n}z. \label{tralala}%
\end{equation}
But doing this one should not forget that even if the operator $\widehat{A}$
is of trace class, formula (\ref{tralala}) need not give the actual trace.
First, the integral in the right-hand side might not be convergent; secondly
even if it is we have to prove that it really yields the right result. We will
discuss the validity of formula (\ref{tralala}) and of other similar formulas
below, but let us first prove some intermediary results.

We begin by discussing the Weyl symbols of Hilbert--Schmidt and trace class operators.

\begin{theorem}
\label{ThmSymbol}Let $\widehat{A}=\operatorname*{Op}^{\mathrm{W}}(a)$ be a
Hilbert--Schmidt operator. Then $a\in L^{2}(\mathbb{R}^{2n})$ and we have%
\begin{equation}
\int|a(z)|^{2}d^{2n}z=\left(  2\pi\hbar\right)  ^{n/2}\iint K(x,y)d^{n}%
xd^{n}y\text{.} \label{aknorm}%
\end{equation}
Conversely, every Weyl operator with symbol $a\in L^{2}(\mathbb{R}^{2n})$ is a
Hilbert--Schmidt operator.
\end{theorem}

\begin{proof}
In view of Theorem \ref{ThmKernel} we have $K\in L^{2}(\mathbb{R}^{n}%
\times\mathbb{R}^{n})$. Let us prove formula (\ref{aknorm}) when
$K\in\mathcal{S}(\mathbb{R}^{n}\times\mathbb{R}^{n})$; it will then hold by
continuity for arbitrary $K\in L^{2}(\mathbb{R}^{n}\times\mathbb{R}^{n})$ in
view of the density of $\mathcal{S}(\mathbb{R}^{n}\times\mathbb{R}^{n})$ in
$L^{2}(\mathbb{R}^{n}\times\mathbb{R}^{n})$. In view of formula (\ref{symbol1}%
) relating the kernel and the symbol of a Weyl operator we have
\begin{equation}
\int|a(x,p)|^{2}dp=\left(  2\pi\hbar\right)  ^{n}\int|K(x+\tfrac{1}%
{2}y,x-\tfrac{1}{2}y)|^{2}d^{n}y. \label{planche10}%
\end{equation}
Integrating this equality with respect to the $x$ variables we get, using
Fubini's theorem%
\begin{align*}
\int|a(z)|^{2}d^{2n}z  &  =\left(  2\pi\hbar\right)  ^{n}\int\left(
\int|K(x+\tfrac{1}{2}y,x-\tfrac{1}{2}y)|^{2}d^{n}y\right)  d^{n}x\\
&  =\left(  2\pi\hbar\right)  ^{n}%
{\displaystyle\iint}
|K(x+\tfrac{1}{2}y,x-\tfrac{1}{2}y)|^{2}d^{n}xd^{n}y.
\end{align*}
Set now $x^{\prime}=x+\tfrac{1}{2}y$ and $y^{\prime}=x-\tfrac{1}{2}y$; we have
$d^{n}x^{\prime}d^{n}y^{\prime}=d^{n}xd^{n}y$ and hence%
\begin{equation}
\int|a(z)|^{2}d^{2n}z=\left(  2\pi\hbar\right)  ^{n}%
{\displaystyle\iint}
|K(x^{\prime},y^{\prime})|^{2}d^{n}x^{\prime}d^{n}y^{\prime} \label{a2k2}%
\end{equation}
which we set out to prove. The converse is obvious since the condition $a\in
L^{2}(\mathbb{R}^{2n})$\ is equivalent to $K\in L^{2}(\mathbb{R}^{n}%
\times\mathbb{R}^{n})$ in view of the inequality above.
\end{proof}

\subsubsection{Rigorous results}

We now specialize our discussion to the case\ $\mathcal{H}=L^{2}%
(\mathbb{R}^{n})$. Recall that we showed in Section \ref{sectionHS} that the
formula%
\begin{equation}
\langle\widehat{A}|\widehat{B}\rangle_{\mathrm{HS}}=\operatorname{Tr}%
(\widehat{A}^{\dag}\widehat{B}) \label{trab}%
\end{equation}
defines an inner product on the ideal $\mathcal{L}_{2}(\mathcal{H)}$ of
Hilbert--Schmidt operators in $\mathcal{H}$ the associated \textquotedblleft
trace norm\textquotedblright\ being defined by
\begin{equation}
||\widehat{A}||_{\mathrm{HS}}=\operatorname{Tr}(\widehat{A}^{\dag}%
\widehat{A})^{1/2} \label{tranorm}%
\end{equation}

Let us state and prove the following result:

\begin{theorem}
\label{ThmHSbis}Let $\widehat{A}=\operatorname*{Op}^{\mathrm{W}}(a)$ and
$\widehat{B}=\operatorname*{Op}^{\mathrm{W}}(a)$ be Hilbert--Schmidt
operators: $\widehat{A},\widehat{B}\in\mathcal{L}_{2}(\mathcal{H)}$. (i) The
trace class operator $\widehat{A}\widehat{B}$ has trace
\begin{equation}
\operatorname*{Tr}(\widehat{A}\widehat{B})=\left(  \tfrac{1}{2\pi\hbar
}\right)  ^{n}\int a(z)b(z)d^{2n}z; \label{trab12}%
\end{equation}
(ii) The Hilbert--Schmidt inner product is given by the convergent integral%
\begin{equation}
\langle\widehat{A}|\widehat{B}\rangle_{\mathrm{HS}}=\left(  \tfrac{1}%
{2\pi\hbar}\right)  ^{n}\int a^{\ast}(z)b(z)d^{2n}z \label{trab13}%
\end{equation}
and hence $||\widehat{A}||_{\mathrm{HS}}^{2}=\operatorname*{Tr}(\widehat{A}%
^{\dag}\widehat{A})$ is given by
\begin{equation}
||\widehat{A}||_{\mathrm{HS}}^{2}=\left(  \tfrac{1}{2\pi\hbar}\right)
^{n}\int|a(z)|^{2}d^{2n}z. \label{traa}%
\end{equation}

\end{theorem}

\begin{proof}
(i) We first observe that in view of Theorem \ref{ThmSymbol} we have $a\in
L^{2}(\mathbb{R}^{n})$ and $b\in L^{2}(\mathbb{R}^{n})$ hence the integrals in
(\ref{trab12}) and (\ref{trab13}) are absolutely convergent. Let $(\psi_{j})$
be an orthonormal basis of $L^{2}(\mathbb{R}^{n})$; by definition of the trace
we have%
\[
\operatorname*{Tr}(\widehat{A}\widehat{B})=\sum_{j}\langle\psi_{j}%
|\widehat{A}\widehat{B}\psi_{j}\rangle=\sum_{j}\langle\widehat{A}^{\dag}%
\psi_{j}|\widehat{B}\psi_{j}\rangle.
\]
Expanding $\widehat{B}\psi_{j}$ and $\widehat{A}^{\dag}\psi_{j}$ in the basis
$(\psi_{j})$ we get
\[
\widehat{B}\psi_{j}=\sum_{k}\langle\psi_{k}|\widehat{B}\psi_{j}\rangle\psi
_{k}\text{ \ , \ }\widehat{A}^{\dag}\psi_{j}=\sum_{\ell}\langle\widehat{A}%
\psi_{\ell}|\psi_{j}\rangle\psi_{\ell}%
\]
\ and hence, using the Bessel equality (\ref{Besseleq}),
\begin{equation}
\langle\widehat{A}^{\dag}\psi_{j}|\widehat{B}\psi_{j}\rangle=\sum_{k}%
\langle\widehat{A}^{\dag}\psi_{j}|\psi_{k}\rangle\langle\widehat{B}\psi
_{j}|\psi_{k}\rangle^{\ast} \label{abba}%
\end{equation}
In view of formula (\ref{average4}) we have
\begin{gather*}
\langle\widehat{A}^{\dag}\psi_{j}|\psi_{k}\rangle=\int a(z)W(\psi_{j},\psi
_{k})(z)d^{2n}z\\
\langle\widehat{B}\psi_{j}|\psi_{k}\rangle=\int b^{\ast}(z)W(\psi_{j},\psi
_{k})(z)d^{2n}z
\end{gather*}
where $W(\psi_{j},\psi_{k})$ is the cross-Wigner transform (\ref{crosswigner})
of $\psi_{j},\psi_{k}$; denoting by $\langle\langle\cdot|\cdot\rangle\rangle$
the inner product on $L^{2}(\mathbb{R}^{2n})$ these equalities can be
rewritten
\begin{gather*}
\langle\widehat{A}^{\dag}\psi_{j}|\psi_{k}\rangle=\langle\langle a^{\ast
}|W(\psi_{j},\psi_{k})\rangle\rangle\\
\langle\widehat{B}\psi_{j}|\psi_{k}\rangle=\langle\langle b|W(\psi_{j}%
,\psi_{k})\rangle\rangle
\end{gather*}
and hence It follows from the extended Moyal identity (\ref{CrossMoyal}) that%
\[
\operatorname*{Tr}(\widehat{A}\widehat{B})=\sum_{j,k}\langle\langle a^{\ast
}|W(\psi_{j},\psi_{k})\rangle\rangle\langle\langle b|W(\psi_{j},\psi
_{k})\rangle\rangle^{\ast}.
\]
Since $(\psi_{j})$ is an orthonormal basis the vectors $(2\pi\hbar
)^{n/2}W(\psi_{j},\psi_{k})$ also form an orthonormal basis (see the remark
following formula (\ref{CrossMoyalbis})), hence the Bessel identity
(\ref{Besseleq}) allows us to write the equality above as
\[
\operatorname*{Tr}(\widehat{A}\widehat{B})=\left(  \tfrac{1}{2\pi\hbar
}\right)  ^{n}\langle\langle a^{\ast}|b\rangle\rangle
\]
which is formula (\ref{trab12}). (ii) It immediately follows from
formula\ (\ref{trab12}) using (\ref{trab}) and (\ref{tranorm}), recalling that
if $\widehat{A}=\operatorname*{Op}^{\mathrm{W}}(a)$ then $\widehat{A}^{\dag
}=\operatorname*{Op}^{\mathrm{W}}(a^{\ast})$.
\end{proof}

Part (i) of the result above allows us -- at last! -- to express the trace of
a Weyl operator in terms of its symbol provided that the latter is absolutely integrable:

\begin{corollary}
Let $\widehat{A}=\operatorname*{Op}^{\mathrm{W}}(a)$ be a trace class
operator. If in addition we have $a\in L^{1}(\mathbb{R}^{n})$ then%
\begin{equation}
\operatorname*{Tr}(\widehat{A})=\left(  \tfrac{1}{2\pi\hbar}\right)  ^{n}\int
a(z)d^{2n}z. \label{TraceA}%
\end{equation}

\end{corollary}

\begin{proof}
(\textit{Cf}. Du and Wong \cite{duwong}, Theorem 2.4.) It is equivalent to
prove that the symplectic Fourier transform $a_{\sigma}=F_{\sigma}a$
satisfies
\begin{equation}
\operatorname*{Tr}(\widehat{A})=a_{\sigma}(0). \label{duwong2}%
\end{equation}
Writing $\widehat{A}=\widehat{B}\widehat{C}$ where $\widehat{B}$ and
$\widehat{C}$ are Hilbert--Schmidt operators we have%
\[
\operatorname*{Tr}(\widehat{A})=\left(  \tfrac{1}{2\pi\hbar}\right)  ^{n}\int
b(z)c(z)d^{2n}z
\]
hence it suffices to show that%
\[
a_{\sigma}(0)=\left(  \tfrac{1}{2\pi\hbar}\right)  ^{n}\int b(z)c(z)d^{2n}z.
\]
We have, in view of formula (\ref{cecomp}) giving the twisted symbol of the
product of two Weyl operators%
\[
a_{\sigma}(z)=\left(  \tfrac{1}{2\pi\hbar}\right)  ^{n}\int e^{\frac{i}%
{2\hbar}\sigma(z,z^{\prime})}b_{\sigma}(z-z^{\prime})c_{\sigma}(z^{\prime
})d^{2n}z^{\prime}%
\]
and hence, using the Plancherel identity (\ref{Plancherelsig}) for $F_{\sigma
}$,%
\begin{align*}
a_{\sigma}(0)  &  =\left(  \tfrac{1}{2\pi\hbar}\right)  ^{n}\int b_{\sigma
}(-z^{\prime})c_{\sigma}(z^{\prime})d^{2n}z^{\prime}\\
&  =\left(  \tfrac{1}{2\pi\hbar}\right)  ^{n}\int b(z)c(z)d^{2n}z
\end{align*}
which proves the formula (\ref{duwong2}).
\end{proof}

The following result is very much in the spirit of the $C^{\ast}$-algebraic
approach outlined in Section \ref{secalgebraic} (\textit{cf.}
Gracia-Bond\'{\i}a and Varilly \cite{Gracia1,Gracia2}). Let us denote by
$\operatorname*{Op}^{\mathrm{W}}$ denote the Weyl transform: $a\longmapsto
\widehat{A}=\operatorname*{Op}^{\mathrm{W}}(a)$; it associates to every symbol
the corresponding Weyl operator.

\begin{theorem}
(i) $\operatorname*{Op}^{\mathrm{W}}$ is an isomorphism of the Hilbert space
$L^{2}(\mathbb{R}^{n})$ onto the algebra $\mathcal{L}_{2}(L^{2}(\mathbb{R}%
^{n})\mathcal{)}$ of Hilbert--Schmidt operators on $L^{2}(\mathbb{R}^{n})$;
(ii) ...
\end{theorem}

\begin{proof}
(i) In view of Theorem \ref{ThmSymbol} a bounded operator $\widehat{A}$ on
$L^{2}(\mathbb{R}^{n})$ is Hilbert--Schmidt if and only if its Weyl symbol $a$
is in $L^{2}(\mathbb{R}^{n})$. The Weyl correspondence being linear and
one-to-one the statement follows.
\end{proof}

\section{Metaplectic Group and Symplectic Covariance}

For a complete study of the metaplectic group in quantum mechanics see our
book \cite{ICP}; on a slightly more general and technical level see
\cite{Birk,Birkbis}. An excellent introduction to the symplectic group is
given in Arvind \textit{et al.} \cite{Arvind}, also see Garc\'{\i}a-Bull\'{e}
\textit{et al}. \cite{Garcia}.

\subsection{The metaplectic representation}

\subsubsection{The generators of $\operatorname*{Sp}(n)$}

Recall that the symplectic form on phase space $\mathbb{R}^{2n}$ can be
defined by $\sigma(z,z^{\prime})=(z^{\prime})^{T}Jz$ where $J=%
\begin{pmatrix}
0_{n\times n} & I_{n\times n}\\
-I_{n\times n} & 0_{n\times n}%
\end{pmatrix}
$ is the standard symplectic matrix (we use the notation $z=(x,p)$,
$z^{\prime}=(x^{\prime},p^{\prime})$). By definition, the symplectic group
$\operatorname*{Sp}(n)$ consists of all real $2n\times2n$ matrices $S$ such
that $\sigma(Sz,Sz^{\prime})=\sigma(z,z^{\prime})$ for all vectors
$z,z^{\prime}$; such a matrix is called a \textit{symplectic matrix}.
Rewriting this condition as $(Sz^{\prime})^{T}JSz=(z^{\prime})^{T}Jz$ we thus
have $S\in\operatorname*{Sp}(n)$ if and only if $S^{T}JS=J$. It is an easy
exercise to show that if $S$ is symplectic then $S^{-1}$ and $S^{T}$ are
symplectic as well, hence this defining relation is equivalent to $SJS^{T}=J$.
The symplectic group plays an essential role in classical mechanics in its
Hamiltonian formulation; its role in quantum mechanics is no less important,
in association with its double covering, the metaplectic group
$\operatorname*{Mp}(n)$ which we briefly describe now.

There are several ways to introduce the metaplectic group. We begin by giving
a definition using the notion of free symplectic matrix and its generating
function (we are following here our presentation in \cite{Birk,Birkbis}). Let
\begin{equation}
S=%
\begin{pmatrix}
A & B\\
C & D
\end{pmatrix}
\label{block1}%
\end{equation}
be a symplectic matrix, where the \textquotedblleft blocks\textquotedblright%
\ $A,B,C,D$ are $n\times n$ matrices. It is easy to show that the relations
$SJS^{T}=S^{T}JS=J$ are equivalent to the two groups of conditions
\begin{align}
A^{T}C\text{, }B^{T}D\text{ \ \textit{are symmetric, and} }A^{T}D-C^{T}B  &
=I\label{cond12}\\
AB^{T}\text{, }CD^{T}\text{ \ \textit{are\ symmetric, and} }AD^{T}-BC^{T}  &
=I\text{.} \label{cond22}%
\end{align}
One says that the block-matrix (\ref{block1}) is a \emph{free symplectic
matrix }if $B$ is invertible, i.e. $\det B\neq0$. To a free symplectic matrix
is associated a generating function: it is the quadratic form%
\begin{equation}
\mathcal{A}(x,x^{\prime})=\frac{1}{2}DB^{-1}x\cdot x-B^{-1}x\cdot x^{\prime
}+\frac{1}{2}B^{-1}Ax^{\prime}\cdot x^{\prime}. \label{wfree}%
\end{equation}
The terminology comes from the fact that the knowledge of $\mathcal{A}%
(x,x^{\prime})$ uniquely determines the free symplectic matrix $S$: we have%
\[%
\begin{pmatrix}
x\\
p
\end{pmatrix}
=%
\begin{pmatrix}
A & B\\
C & D
\end{pmatrix}%
\begin{pmatrix}
x^{\prime}\\
p^{\prime}%
\end{pmatrix}
\Longleftrightarrow\left\{
\begin{array}
[c]{c}%
p=\nabla_{x}\mathcal{A}(x,x^{\prime})\\
p^{\prime}=-\nabla_{x^{\prime}}\mathcal{A}(x,x^{\prime})
\end{array}
\right.
\]
as can be verified by a direct calculation. The interest of the notion of free
symplectic matrix comes from the fact that such matrices generate the
symplectic group $\operatorname*{Sp}(n)$. More precisely every $S\in
\operatorname*{Sp}(n)$ can be written as a product $S=S_{\mathcal{A}%
}S_{\mathcal{A}^{\prime}}$ (we place the corresponding generating functions
$\mathcal{A}$ and $\mathcal{A}^{\prime}$ as subscripts).

Defining, for symmetric $P$ and invertible $L$, the symplectic matrices
$V_{-P}$ and $M_{L}$ by
\begin{equation}
V_{-P}=%
\begin{pmatrix}
I & 0\\
P & I
\end{pmatrix}
\text{ \ , \ }M_{L}=%
\begin{pmatrix}
L^{-1} & 0\\
0 & L^{T}%
\end{pmatrix}
\label{vpml}%
\end{equation}
a straightforward calculations shows that the free symplectic matrix
$S_{\mathcal{A}}$ can be factored as
\begin{equation}
S_{\mathcal{A}}=V_{-DB^{-1}}M_{B^{-1}}JV_{-B^{-1}A}. \label{savpml}%
\end{equation}
This implies that the symplectic group $\operatorname*{Sp}(n)$ is generated by
the set of all matrices $V_{-P}$ and $M_{L}$ together with $J$. It is easy to
deduce from this that the determinant of a symplectic matrix always is equal
to one since we obviously have $\det V_{-P}=\det M_{L}=\det J$.

\subsubsection{Generalized Fourier transforms}

Now, to every free symplectic matrix $S_{\mathcal{A}}$ we associate two
operators $\widehat{S}_{\mathcal{A},m}$ by the formula%
\begin{equation}
\widehat{S}_{\mathcal{A},m}\psi(x)=\left(  \tfrac{1}{2\pi\hbar}\right)
^{n/2}i^{m-n/2}\sqrt{|\det B^{-1}|}\int e^{\frac{i}{\hbar}\mathcal{A}%
(x,x^{\prime})}\psi(x^{\prime})d^{n}x^{\prime} \label{qft1}%
\end{equation}
where $m$ corresponds to a choice of argument for $\det B^{-1}$: $m=0$
$\operatorname{mod}2$ if $\det B^{-1}>0$ and $m=1$ $\operatorname{mod}2$ if
$\det B^{-1}<0$. It is not difficult to prove that the generalized Fourier
transforms $\widehat{S}_{\mathcal{A},m}$ are unitary operators on
$L^{2}(\mathbb{R}^{n})$. These operators generate a group: the metaplectic
group $\operatorname*{Mp}(n)$. One shows that, as for the symplectic group,
every $\widehat{S}\in\operatorname*{Mp}(n)$ can be written (non uniquely) as a
product $\widehat{S}_{\mathcal{A},m}\widehat{S}_{\mathcal{A}^{\prime
},m^{\prime}}$. This group is a double covering of $\operatorname*{Sp}(n)$,
the covering projection being simply defined by
\begin{equation}
\pi_{\operatorname*{Mp}}:\operatorname*{Mp}(n)\longrightarrow
\operatorname*{Sp}(n)\text{ \ , \ }\pi_{\operatorname*{Mp}}(\widehat{S}%
_{\mathcal{A},m})=S_{\mathcal{A}}. \label{pimp}%
\end{equation}
Here are three examples of free symplectic matrices and of their metaplectic
counterparts; these can be used to give an alternative definition of the
metaplectic group:

\begin{itemize}
\item \textit{The standard symplectic matrix} $J$ has as generating function
$\mathcal{A}(x,x^{\prime})=-x\cdot x^{\prime}$ and hence the two corresponding
metaplectic operators are $\pm\widehat{J}$ with%
\[
\widehat{J}\psi(x)=\left(  \tfrac{1}{2\pi\hbar i}\right)  ^{n/2}\int
e^{-\frac{i}{\hbar}x\cdot x^{\prime}}\psi(x^{\prime})d^{n}x^{\prime};
\]
observe that $\widehat{J}=i^{-n/2}F$ where $F$ is the usual Fourier transform;

\item \textit{The symplectic shear}\ $V_{-P}=%
\begin{pmatrix}
I & 0\\
P & I
\end{pmatrix}
$ ($P=P^{T}$) is not free, but
\[
U_{-P}=JV_{-P}J^{-1}=%
\begin{pmatrix}
I & -P\\
0 & I
\end{pmatrix}
\]
is if $\det P$ $\neq0$. In this case we have $\mathcal{A}(x,x^{\prime}%
)=-\frac{1}{2}P^{-1}x\cdot x+P^{-1}x\cdot x^{\prime}$ and the corresponding
metaplectic operators are hence $\pm\widehat{U}_{-P}$ with%
\[
\widehat{U}_{-P}=\left(  \tfrac{1}{2\pi\hbar}\right)  ^{n/2}i^{n/2}|\det
P|^{-1}e^{-\frac{i}{2\hbar}P^{-1}x\cdot x}\int e^{\frac{i}{\hbar}P^{-1}x\cdot
x^{\prime}}\psi(x^{\prime})d^{n}x^{\prime}.
\]

\item \textit{The symplectic rescaling matrix} $M_{L}=%
\begin{pmatrix}
L^{-1} & 0\\
0 & L^{T}%
\end{pmatrix}
$ is not free but the product
\[
M_{L}J=%
\begin{pmatrix}
0 & L^{-1}\\
L^{T} & 0
\end{pmatrix}
\]
is and has $\mathcal{A}(x,x^{\prime})=Lx\cdot x^{\prime}$ as generating
function; the corresponding metaplectic operator are $\widehat{M}%
_{L,m}\widehat{J}$ where%
\[
\widehat{M}_{L,m}\psi(x)=i^{m}\sqrt{|\det L|}\psi(Lx)
\]
the integer $m$ (the \textquotedblleft Maslov index\textquotedblright)
corresponding to a choice of $\arg\det L$.
\end{itemize}

It turns out that an easy calculation shows that, similarly to the
factorization (\ref{savpml}) of free symplectic matrices, the quadratic
Fourier transform (\ref{qft1}) can be written%
\[
\widehat{S}_{\mathcal{A},m}=\widehat{V}_{-B^{-1}A}\widehat{M}_{B^{-1}%
,m}\widehat{J}\widehat{V}_{-DB^{-1}}%
\]
where $\widehat{M}_{L,m}$ and $\widehat{J}$ are defined as above and%
\[
\widehat{V}_{-P}\psi(x)=e^{\frac{i}{2\hbar}Px^{2}}\psi(x)
\]
when $P=P^{T}$. It follows that the elementary operators $\widehat{V}%
_{-P},\widehat{M}_{L,m}$ and $\widehat{J}$ generate $\operatorname*{Mp}(n)$
(these operators are used in many texts to define the metaplectic group; our
approach using (\ref{qft1}) has some advantages since among other things it
makes immediately clear that metaplectic operators are generalized Fourier transforms).

The factorization $\widehat{S}=\widehat{S}_{\mathcal{A},m}\widehat{S}%
_{\mathcal{A}^{\prime},m^{\prime}}$ of a metaplectic operator is by no means
unique; for instance we can write the identity operator $I$ as $\widehat{S}%
_{\mathcal{A},m}\widehat{S}_{\mathcal{A},m}^{-1}$ $=\widehat{S}_{\mathcal{A}%
,m}\widehat{S}_{\mathcal{A}^{\ast},m^{\ast}}$ for every quadratic Fourier
transform $\widehat{S}_{\mathcal{A},m}$. There is however an invariant
attached to $\widehat{S}$: the Maslov index. Denoting by
$\operatorname*{Inert}R$ the index of inertia (= the number of negative
eigenvalues) of the real symmetric matrix $R$ we have:

\begin{proposition}
Let $\widehat{S}=\widehat{S}_{\mathcal{A},m}\widehat{S}_{\mathcal{A}^{\prime
},m^{\prime}}=\widehat{S}_{\mathcal{A}^{\prime\prime},m^{\prime\prime}%
}\widehat{S}_{\mathcal{A}^{^{\prime\prime\prime}},m^{\prime\prime\prime}}$. We
have%
\begin{equation}
m+m^{\prime}-\operatorname*{Inert}(P^{\prime}+Q)\equiv m^{\prime\prime
}+m^{\prime\prime\prime}-\operatorname*{Inert}(P^{\prime\prime\prime
}+Q^{\prime\prime})\text{ \ }\operatorname{mod}4.\tag{A8}\label{mm1}%
\end{equation}

\end{proposition}

\begin{proof}
See Leray \cite{Leray}, de Gosson \cite{AIF,Birk}.
\end{proof}

It follows from formula (\ref{mm1}) that the class modulo $4$ of the integer
$m+m^{\prime}-\operatorname*{Inert}(P^{\prime}+Q)$ does not depend on the way
we write $\widehat{S}\in\operatorname*{Mp}(n)$ as a product $\widehat{S}%
_{\mathcal{A},m}\widehat{S}_{\mathcal{A}^{\prime},m^{\prime}}$ of quadratic
Fourier transforms; this class is denoted by $m(\widehat{S})$ and called the
Maslov index of $\widehat{S}$. The mapping
\[
m:\operatorname*{Mp}(n)\in\widehat{S}\longrightarrow m(\widehat{S}%
)\in\mathbb{Z}_{4}\text{ \ }%
\]
is called the Maslov index on $\operatorname*{Mp}(n)$. We have $m(\widehat{S}%
_{\mathcal{A},m})=m$, $\operatorname{mod}4$ (\cite{Leray,AIF}). The theory of
the Maslov index has been further developed by Arnol'd, Leray, and by the
author (see the review \cite{CLM} by Cappell \textit{et al.}). To the Maslov
index $m$ is associated another integer index which plays an essential role in
quantum holography 

\subsubsection{The Weyl symbol of a metaplectic operator }

We define the following subset of $\operatorname*{Sp}(n)$:
\begin{equation}
\operatorname*{Sp}\nolimits_{0}(n)=\{S\in\operatorname*{Sp}(n):\det
(S-I)\neq0\}.\label{A9}%
\end{equation}

To $S\in\operatorname*{Sp}\nolimits_{0}(n)$ we associate the family of
operators $\widehat{R}_{\nu}(S)$ defined, for $\nu\in\mathbb{R}$, by
\begin{equation}
\widehat{R}_{\nu}(S)=\left(  \tfrac{1}{2\pi\hbar}\right)  ^{n}i^{\nu}%
\sqrt{|\det(S-I)|}\int\widehat{T}(Sz_{0})\widehat{T}(-z_{0})d^{2n}%
z_{0}.\label{rus1}%
\end{equation}
One verifies that for all $S\in\operatorname*{Sp}\nolimits_{0}(n)$ and $\nu
\in\mathbb{R}$ the operators $\widehat{R}_{\nu}(S)$ satisfy the intertwining
formula
\[
\widehat{T}(Sz_{0})=\widehat{R}_{\nu}(S)\widehat{T}(z_{0})\widehat{R}_{\nu
}(S)^{-1}.
\]
It follows, using the irreducibility of the Schr\"{o}dinger representation of
the Heisenberg group \cite{Folland}, that there exists a constant $c(S,\nu
)\in\mathbb{C}$ such that $\widehat{R}_{\nu}(S)=c(S,\nu)\widehat{S}$ where
$\pi^{\operatorname*{Mp}}(\widehat{S})=S$. It is moreover easy to check that
the operators are $\widehat{R}_{\nu}(S)$ unitary, hence $|c(S,\nu)|=1$. The
following result connects the integer $\nu$ in (\ref{rus1}) to the
Conley--Zehnder index when $\widehat{R}_{\nu}(S)$ is a true metaplectic operator:

\begin{proposition}
\label{propcz}Let $\Sigma=(S_{t})_{t\in I}$ be symplectic isotopy in
$\operatorname*{Sp}(n)$ leading from the identity to $S\notin%
\operatorname*{Sp}\nolimits_{0}(n)$. Let $\widehat{\Sigma}=(\widehat{S}%
_{t})_{t\in I}$ be the metaplectic isotopy covering $\Sigma$ and
$\widehat{S}\in\operatorname*{Mp}(n)$ be its endpoint (thus $S=\pi
^{\operatorname*{Mp}}(\widehat{S})$). We have $\widehat{S}=\widehat{R}%
_{\nu(\widehat{\Sigma})}(S)$ where $\nu(\widehat{\Sigma})=\nu(\Sigma)$
$\operatorname{mod}4$.
\end{proposition}

\begin{proof}
This results from the identity (\ref{iczmod4}) (see \cite{LMP} and
\cite{RMPCZ,Birk}).
\end{proof}

The statement above has the following consequences when the endpoint of the
symplectic isotopy $\Sigma$ is a free symplectic matrix $S_{\mathcal{A}}$ (a
symplectic isotopy is a $C^{1}$ mapping $\Sigma:t\longmapsto S_{t}%
\in\operatorname*{Sp}(n)$ such that $S_{0}=I_{\mathrm{d}}$):

\begin{corollary}
Let $\widehat{S}_{\mathcal{A},m}\in\operatorname*{Mp}(n)$ be such that
$S_{\mathcal{A}}=\pi^{\operatorname*{Mp}}(\widehat{S}_{\mathcal{A},m}%
)\notin\operatorname*{Sp}\nolimits_{0}(n)$. We then have
\begin{equation}
\widehat{S}_{\mathcal{A},m}=\widehat{R}_{m-\operatorname*{Inert}%
\mathcal{A}_{xx}}(S)\label{swmnu}%
\end{equation}
where $\operatorname*{Inert}\mathcal{A}_{xx}$ is the index of inertia of the
matrix $\mathcal{A}_{xx}$ of second derivatives of the quadratic form
$x\longmapsto\mathcal{A}(x,x)$ on $\mathbb{R}^{n}$.
\end{corollary}

This allows us to give a rigorous explicit formula for the twisted Weyl symbol
of $\widehat{S}_{\mathcal{A},m}$:

\begin{corollary}
The twisted Weyl symbol of $\widehat{S}_{\mathcal{A},m}$ with $S_{\mathcal{A}%
}\notin\operatorname*{Sp}\nolimits_{0}(n)$ is given by%
\begin{equation}
(s_{\mathcal{A}})_{\sigma}(z)=\frac{i^{m-\operatorname*{Inert}\mathcal{A}%
_{xx}}}{\sqrt{|\det(S_{\mathcal{A}}-I)|}}\exp\left(  \frac{i}{2\hbar
}M_{\mathcal{A}}z\cdot z\right)  \label{weylmp1}%
\end{equation}
where $M_{\mathcal{A}}$ is the symplectic Cayley transform of $S_{\mathcal{A}%
}$.
\end{corollary}

\begin{proof}
See de Gosson \cite{LMP,Birk}.
\end{proof}

Proposition \ref{propcz} and formula (\ref{weylmp1}) suggest that the
Conley--Zehnder index is related to a choice of argument of the square root of
the determinant of $S-I$. This is indeed the case:

\begin{proposition}
\label{propsw}Let $\widehat{S}_{\mathcal{A},m}\in\operatorname*{Mp}(n)$ have
projection $S_{\mathcal{A}}\notin\operatorname*{Sp}_{0}(n)$. We have%
\begin{equation}
\nu(\widehat{S}_{\mathcal{A},m})=n+\frac{1}{\pi}\arg\det(S_{\mathcal{A}%
}-I)\text{ \ }\operatorname{mod}2.\label{argdet1}%
\end{equation}
that is%
\begin{equation}
\nu(\widehat{S}_{\mathcal{A},m})=\left\{
\begin{array}
[c]{c}%
n\text{ \ }\operatorname{mod}2\text{ \ if }S_{\mathcal{A}}\in
\operatorname*{Sp}\nolimits_{+}(n)\\
n+2\text{ \ }\operatorname{mod}2\text{ \ if }S_{\mathcal{A}}\in
\operatorname*{Sp}\nolimits_{-}(n)
\end{array}
\right.  .\label{argdet2}%
\end{equation}

\end{proposition}

\begin{proof}
The projection $S_{\mathcal{A}}=\pi^{\operatorname*{Mp}}(\widehat{S}%
_{\mathcal{A},m})$ is a free symplectic matrix, in block-matrix form%
\[
S_{\mathcal{A}}=%
\begin{pmatrix}
A & B\\
C & D
\end{pmatrix}
\text{ \ , \ }\det B\neq0.
\]
A straightforward calculation yields the factorization
\[
S_{\mathcal{A}}-I=%
\begin{pmatrix}
0 & B\\
I & D-I
\end{pmatrix}%
\begin{pmatrix}
C-(D-I)B^{-1}(A-I) & 0\\
B^{-1}(A-I) & I
\end{pmatrix}
.
\]
Since $S_{\mathcal{A}}\in\operatorname*{Sp}(n)$ we have $C-DB^{-1}%
A=-(B^{T})^{-1}$ and hence
\[
C-(D-I)B^{-1}(A-I)=B^{-1}A+DB^{-1}-(B^{T})^{-1}=\mathcal{A}_{xx}%
\]
so that
\[
S_{\mathcal{A}}-I=%
\begin{pmatrix}
0 & B\\
I & D-I
\end{pmatrix}%
\begin{pmatrix}
\mathcal{A}_{xx} & 0\\
B^{-1}(A-I) & I
\end{pmatrix}
.
\]
It follows that
\[
\det(S_{\mathcal{A}}-I)=(-1)^{n}\det B\det\mathcal{A}_{xx}%
\]
and hence%
\[
\arg\det(S_{\mathcal{A}}-I)=n\pi+\arg\det B+\arg\det\mathcal{A}_{xx}\text{
\ }\operatorname{mod}2\pi.
\]
Noticing that $\arg\det\mathcal{A}_{xx}=\pi\operatorname{Inert}\mathcal{A}%
_{xx}$ and that this is%
\[
\arg\det(S_{\mathcal{A}}-I)=n\pi+\arg\det B+\pi\operatorname{Inert}%
\mathcal{A}_{xx}\text{ \ }\operatorname{mod}2\pi.
\]
In view of formula (\ref{mod4}) and (\ref{weylmp1}) we have $\arg\det(B)=m\pi$
(see Appendix A) and hence%
\[
\arg\det(S_{\mathcal{A}}-I)=(n+m-\operatorname{Inert}\mathcal{A}_{xx}%
)\pi\text{ \ }\operatorname{mod}2\pi
\]
that is
\[
\arg\det(S_{\mathcal{A}}-I)=(n+\nu(\widehat{S}_{\mathcal{A},m}))\pi\text{
\ }\operatorname{mod}2\pi
\]
which yields (\ref{argdet1}).
\end{proof}

\subsection{Products of metaplectic operators}

Each $\widehat{S}\in\operatorname*{Mp}(n)$ can be written as a product
$\widehat{S}_{\mathcal{A},m}\widehat{S}_{\mathcal{A}^{\prime},m^{\prime}}$
(Appendix A, Proposition \ref{propA1}). It turns out that $\widehat{S}%
_{\mathcal{A},m}$ and $\widehat{S}_{\mathcal{A}^{\prime},m^{\prime}}$ can be
chosen so that their projections $S_{\mathcal{A}}$ and $S_{\mathcal{A}%
^{\prime}}$ have no eigenvalue equal to one. This fact, together with the
composition formula\ (\ref{cecomp}) leads to a complete characterization of
the symbol of a metaplectic operator. When $\widehat{S}$ has projection
$S\notin\operatorname*{Sp}\nolimits_{0}(n)$ we have the following explicit result:

\begin{proposition}
Let $\widehat{S}\in\operatorname*{Mp}(n)$ be such that $\pi
^{\operatorname*{Mp}}(\widehat{S})\notin\operatorname*{Sp}\nolimits_{0}(n)$.

(i) There exist $\widehat{S}_{\mathcal{A},m}$ and $\widehat{S}_{\mathcal{A}%
^{\prime},m^{\prime}}$ such that $\widehat{S}=\widehat{S}_{\mathcal{A}%
,m}\widehat{S}_{\mathcal{A}^{\prime},m^{\prime}}$; moreover these operators
can be chosen so that $S_{\mathcal{A}}=\pi^{\operatorname*{Mp}}(\widehat{S}%
_{\mathcal{A},m})\notin\operatorname*{Sp}\nolimits_{0}(n)$ and $S_{\mathcal{A}%
^{\prime}}=\pi^{\operatorname*{Mp}}(\widehat{S}_{\mathcal{A}^{\prime
},m^{\prime}})\notin\operatorname*{Sp}\nolimits_{0}(n)$.

(ii) We have%
\begin{equation}
\widehat{S}=\widehat{R}_{\nu+\nu^{\prime}+\frac{1}{2}\operatorname*{sign}%
(M)}(S)=\widehat{R}_{\nu(\widehat{S})}(S)\label{sr}%
\end{equation}
where $M=M_{\mathcal{A}}+M_{\mathcal{A}^{\prime}}$ ($M_{\mathcal{A}}$ and
$M_{\mathcal{A}^{\prime}}$ the symplectic Cayley transforms of $S_{\mathcal{A}%
}$ and $S_{\mathcal{A}^{\prime}}$), and%
\begin{equation}
\nu=m-\operatorname*{Inert}\mathcal{A}_{xx}\text{ \ , \ }\nu^{\prime
}=m^{\prime}-\operatorname*{Inert}\mathcal{A}_{xx}^{\prime}\label{smm}%
\end{equation}
are the Conley--Zehnder indices of $\widehat{S}_{\mathcal{A},m}$ and
$\widehat{S}_{\mathcal{A}^{\prime},m^{\prime}}$;

(iii) The twisted Weyl symbol of $\widehat{S}$ is given by
\begin{equation}
s_{\sigma}(z)=\frac{i^{\nu(\widehat{S})}}{\sqrt{|\det(S-I)|}}\exp\left(
\frac{i}{2\hbar}Mz\cdot z\right)  \label{productmw}%
\end{equation}
with
\begin{equation}
\nu(\widehat{S})=\nu+\nu^{\prime}+\tfrac{1}{2}\operatorname*{sign}M.
\label{nus}%
\end{equation}

\end{proposition}

\begin{proof}
See Proposition 10 in \cite{LMP} or \cite{Birk}, \S 7.4 for detailed proofs.
That $\widehat{S}$ can always be factored as $\widehat{S}_{\mathcal{A}%
,m}\widehat{S}_{\mathcal{A}^{\prime},m^{\prime}}$ where $\widehat{S}%
_{\mathcal{A},m}$ and $\widehat{S}_{\mathcal{A}^{\prime},m^{\prime}}$ have
projections $S_{\mathcal{A}}$ and $S_{\mathcal{A}^{\prime}}$ not in
$\operatorname*{Sp}\nolimits_{0}(n)$ was proven in \cite{LMP}. For formula
(\ref{productmw}) the idea is to apply formula (\ref{cecomp}) to
(\ref{weylmp1}) and to use the Fresnel formula (\ref{Fresnel}), which yields,
after some calculations%
\[
c_{\sigma}(z)=\frac{i^{\nu+\nu^{\prime}+\frac{1}{2}\operatorname*{sign}(M)}%
}{\sqrt{|\det[(S_{\mathcal{A}}-I)(S_{\mathcal{A}^{\prime}}-I)M]|}}e^{\frac
{i}{2\hbar}Mz\cdot z}.
\]
A simple calculation taking into account the definition of the symplectic
Cayley transforms shows that%
\begin{equation}
(S_{\mathcal{A}}-I)(S_{\mathcal{A}^{\prime}}-I)M=S-I\label{identity}%
\end{equation}
($M$ is invertible in view of Lemma \ref{leminv}).
\end{proof}

We have seen in Proposition \ref{propsw} that the Conley--Zehnder index of a
quadratic Fourier transform $\widehat{S}_{\mathcal{A},m}$ is simply related to
a choice of argument for $\det(S_{\mathcal{A}}-I)$. Using the result above,
this observation can be generalized to the case of an arbitrary $\widehat{S}%
\in\operatorname*{Mp}(n)$ with projection $S\notin\operatorname*{Sp}%
\nolimits_{0}(n)$:

\begin{corollary}
Let $\widehat{S}\in\operatorname*{Mp}(n)$ with $S=\pi^{\operatorname*{Mp}%
}(\widehat{S})\notin\operatorname*{Sp}\nolimits_{0}(n)$. We have%
\begin{equation}
\nu(\widehat{S})=n+\frac{1}{\pi}\operatorname*{Arg}\det(S-I)\text{
\ }\operatorname{mod}2. \label{argdet3}%
\end{equation}

\end{corollary}

\begin{proof}
Writing $\widehat{S}=\widehat{S}_{\mathcal{A},m}\widehat{S}_{\mathcal{A}%
^{\prime},m^{\prime}}$ with $S_{\mathcal{A}}$ and $S_{\mathcal{A}^{\prime}}$
not in $\operatorname*{Sp}\nolimits_{0}(n)$ it follows from the identity
(\ref{identity}) that
\[
\det\left[  (S_{\mathcal{A}}-I)(S_{\mathcal{A}^{\prime}}-I)M\right]
=\det(S-I)
\]
with $M=M_{\mathcal{A}}+M_{\mathcal{A}^{\prime}}$ and hence
\[
\arg\det(S-I)=\arg\det(S_{\mathcal{A}}-I)+\arg\det(S_{\mathcal{A}^{\prime}%
}-I)+\arg\det M.
\]
Since $M=M_{\mathcal{A}}+M_{\mathcal{A}^{\prime}}$ is invertible (Lemma
\ref{leminv}) we have
\[
\arg\det M=\pi\operatorname{Inert}M=-\pi\operatorname{Inert}M\text{
\ }\operatorname{mod}2\pi
\]
and hence, using formulas (\ref{argdet1}) and (\ref{modprod}) together with
the relation $\operatorname*{sign}M=2(n-\operatorname{Inert}M)$,
\begin{align*}
\arg\det(S-I) &  =\nu(\widehat{S}_{\mathcal{A},m})\pi+\nu(\widehat{S}%
_{\mathcal{A},m})-\pi(n-\tfrac{1}{2}\operatorname*{sign}M)\text{\ \ }%
\operatorname{mod}2\pi\\
&  =\nu(\widehat{S}_{\mathcal{A},m})\pi+\nu(\widehat{S}_{\mathcal{A},m}%
)-n\pi+\tfrac{1}{2}\pi\operatorname*{sign}M\text{\ \ }\operatorname{mod}2\pi\\
&  =\nu(\widehat{S})\pi-n\pi\text{ \ }\operatorname{mod}2\pi
\end{align*}
proving formula (\ref{argdet3}).
\end{proof}

\subsection{Symplectic covariance properties}

\subsubsection{A conjugation property for the displacement and reflection
operators}

Everything here stems from the following observation: let $\widehat{S}%
\in\operatorname*{Mp}(n)$ have projection $S\in\operatorname*{Sp}(n)$ (the
symplectic matrix $S$ is thus \textquotedblleft covered\textquotedblright\ by
the two metaplectic operators $\pm\widehat{S}$). Then for every phase space
point $z_{0}=(x_{0},p_{0})$ the displacement operators $\widehat{D}(Sz_{0})$
and $\widehat{D}(z_{0})$ are related by the conjugation formula
\begin{equation}
\widehat{D}(Sz_{0})=\widehat{S}\widehat{D}(z_{0})\widehat{S}^{\dag}
\label{sycov1}%
\end{equation}
(recall that $\widehat{S}^{\dag}=\widehat{S}^{-1}$ since metaplectic operators
are unitary). This formula is most easily proven using the generators of
$\operatorname*{Sp}(n)$ and the corresponding generators of
$\operatorname*{Mp}(n)$; for a complete proof see for instance de Gosson
\cite{Birkbis}, \S 8.1.3, also \cite{Littlejohn}. It easily follows from
(\ref{sycov1}) that the reflection operator (\ref{parity}) satisfies a similar
relation:
\begin{equation}
\widehat{\Pi}(Sz_{0})=\widehat{S}\widehat{\Pi}(z_{0})\widehat{S}^{\dag}.
\label{sycov2}%
\end{equation}
In fact, recalling (formula (\ref{deltapi})) that $\widehat{\Pi}%
(z_{0})=\widehat{D}(z_{0})\Pi\widehat{D}(z_{0})^{\dag}$ we have%
\[
\widehat{\Pi}(Sz_{0})=\widehat{D}(Sz_{0})\Pi\widehat{D}(Sz_{0})^{\dag
}=\widehat{S}\widehat{D}(z_{0})(\widehat{S}^{\dag}\Pi\widehat{S}%
)\widehat{D}(z_{0})^{\dag}\widehat{S}^{\dag};
\]
to get (\ref{sycov2}) we have to show that $\widehat{S}^{\dag}\Pi
\widehat{S}=\Pi$ since we will then have%
\[
\widehat{\Pi}(Sz_{0})=\widehat{S}\widehat{D}(z_{0})\Pi\widehat{D}(z_{0}%
)^{\dag}\widehat{S}^{\dag}=\widehat{S}\widehat{\Pi}(z_{0})\widehat{S}^{\dag}.
\]
It suffices for that purpose to show that $\widehat{S}_{\mathcal{A},m}^{\dag
}\Pi\widehat{S}_{\mathcal{A},m}=\Pi$ since the generalized Fourier transforms
$\widehat{S}_{\mathcal{A},m}$ generate the metaplectic group. Now,
$\Pi\widehat{S}_{\mathcal{A},m}\psi(x)=\widehat{S}_{\mathcal{A},m}\psi(-x)$
hence, by (\ref{qft1}),
\[
\Pi\widehat{S}_{\mathcal{A},m}\psi(x)=\left(  \tfrac{1}{2\pi\hbar}\right)
^{n/2}i^{m-n/2}\sqrt{|\det B^{-1}|}\int e^{\frac{i}{\hbar}\mathcal{A}%
(-x,x^{\prime})}\psi(x^{\prime})d^{n}x^{\prime}.
\]
Noting that $\mathcal{A}(-x,x^{\prime})=\mathcal{A}(x,-x^{\prime})$ (cf.
formula (\ref{wfree})) we get, making the change of variables $x^{\prime
}\longmapsto-x^{\prime}$,%
\[
\Pi\widehat{S}_{\mathcal{A},m}\psi(x)=\left(  \tfrac{1}{2\pi\hbar}\right)
^{n/2}i^{m-n/2}\sqrt{|\det B^{-1}|}\int e^{\frac{i}{\hbar}\mathcal{A}%
(x,x^{\prime})}\psi(-x^{\prime})d^{n}x^{\prime}%
\]
that is $\Pi\widehat{S}_{\mathcal{A},m}=\widehat{S}_{\mathcal{A},m}\Pi$; it
follows that we have
\[
\widehat{S}_{\mathcal{A},m}^{\dag}\Pi\widehat{S}_{\mathcal{A},m}%
=\widehat{S}_{\mathcal{A},m}^{\dag}\widehat{S}_{\mathcal{A},m}\Pi=\Pi.
\]

\subsubsection{Symplectic covariance}

Collecting the facts above we have:

\begin{theorem}
\label{Thmsyco}Let $z=(x,p)$ be a point in the phase space $\mathbb{R}^{2n}$
and $\widehat{S}$ a metaplectic operator with projection $\pi
_{\operatorname*{Mp}}(\widehat{S})=S$ in $\operatorname*{Sp}(n)$. (i) We have%
\begin{equation}
W\psi(Sz)=W(\widehat{S}^{-1}\psi)(z)\text{ \ , \ }\operatorname*{Amb}%
\psi(Sz)=\operatorname*{Amb}(\widehat{S}^{-1}\psi)(z). \label{sycov3}%
\end{equation}
(ii) For every symbol $a$ we have
\begin{equation}
\operatorname*{Op}\nolimits_{\mathrm{W}}(a\circ S^{-1})=\widehat{S}%
\operatorname*{Op}\nolimits_{\mathrm{W}}(a)\widehat{S}^{\dag} \label{sycov4}%
\end{equation}
where $a\circ S^{-1}(z)=a(S^{-1}z)$.
\end{theorem}

\begin{proof}
(i) To prove the first identity (\ref{sycov3}) we recall (formula (\ref{wgr}))
that%
\[
W\psi(z)=\left(  \tfrac{1}{\pi\hbar}\right)  ^{n}\langle\psi|\widehat{\Pi
}(z)\psi\rangle
\]
and hence, using (\ref{sycov2}) and the unitarity of metaplectic operators,
\[
W\psi(Sz)=\left(  \tfrac{1}{\pi\hbar}\right)  ^{n}\langle\psi|\widehat{S}%
\widehat{\Pi}(z)\widehat{S}^{\dag}\psi\rangle=\left(  \tfrac{1}{\pi\hbar
}\right)  ^{n}\langle\widehat{S}^{\dag}\psi|\widehat{\Pi}(z)\widehat{S}^{\dag
}\psi\rangle
\]
which is precisely (\ref{sycov3}). The proof of the second identity
(\ref{sycov3}) is similar using the definition (\ref{ambhw}) of the ambiguity
function together with property (\ref{sycov1}). (ii) Recall (formula
(\ref{Weyl2})) that the Weyl operator $\widehat{A}=\operatorname*{Op}%
\nolimits_{\mathrm{W}}(a)$ can be written
\[
\widehat{A}=\left(  \tfrac{1}{\pi\hbar}\right)  ^{n}\int a(z)\widehat{\Pi
}(z)d^{2n}z
\]
and hence, using (\ref{sycov2}),%
\begin{align*}
\widehat{S}\widehat{A}\widehat{S}^{-1}  &  =\left(  \tfrac{1}{\pi\hbar
}\right)  ^{n}\int a(z)\widehat{S}\widehat{\Pi}(z)\widehat{S}^{-1}d^{2n}z\\
&  =\left(  \tfrac{1}{\pi\hbar}\right)  ^{n}\int a(z)\widehat{\Pi}(Sz)d^{2n}z;
\end{align*}
performing the change of variables $z^{\prime}=Sz$ we have, since $\det S=1$,%
\[
\widehat{S}\widehat{A}\widehat{S}^{-1}=\left(  \tfrac{1}{\pi\hbar}\right)
^{n}\int a(S^{-1}z)\widehat{\Pi}(z)d^{2n}z=\operatorname*{Op}%
\nolimits_{\mathrm{W}}(a\circ S^{-1})
\]
as claimed.
\end{proof}

Applying the machinery above to the density matrix we get:

\begin{corollary}
Let $\{(\psi_{j},\alpha_{j})\}$ be a mixed state with density matrix
$\widehat{\rho}$ and Wigner function $\rho$. Let $\widehat{S}\in
\operatorname*{Mp}(n)$. The mixed state $\{(\widehat{S}\psi_{j},\alpha_{j})\}$
has density matrix $\widehat{S}\widehat{\rho}\widehat{S}^{\dag}$ and Wigner
function $\rho(S^{-1}z)$, where $S=\pi_{\operatorname*{Mp}}(\widehat{S})$.
\end{corollary}

\begin{proof}
The Wigner distribution of $\{(\widehat{S}\psi_{j},\alpha_{j})\}$ is
\[
\sum_{j}\alpha_{j}W(\widehat{S}\psi_{j})(z)=\sum_{j}\alpha_{j}W\psi_{j}%
(S^{-1}z)
\]
because of the first formula (\ref{sycov3}). It follows that the Weyl symbol
of the density matrix corresponding to $\{(\widehat{S}\psi_{j},\alpha_{j})\}$
is $a(S^{-1}z)$ where $a=(2\pi\hbar)^{n}\rho$ is the Weyl symbol of
$\widehat{\rho}$; that $\widehat{S}\widehat{\rho}\widehat{S}^{\dag}$ is the
density matrix of $\{(\widehat{S}\psi_{j},\alpha_{j})\}$ follows from formula
(\ref{sycov4}).
\end{proof}

Note that the fact that $\widehat{S}\widehat{\rho}\widehat{S}^{\dag}$ is the
density matrix of $\{(\widehat{S}\psi_{j},\alpha_{j})\}$ can also be proven
directly using the definition (\ref{trc1}) of the density matrix in term of projectors.

\section{Variable Planck Constant}

We now address one of the central themes of this Review, namely the
mathematical consequences of possible changes in the value of Planck's
constant $h$.

\subsection{A consequence of Moyal's identity}

We begin by a few straightforward observations involving the Moyal identity
introduced in Section \ref{secmoyal}. Let $\eta$ be a real parameter; we
assume for the moment that $\eta>0$. This parameter will play the role of a
variable $\hbar=h/2\pi$. For any square integrable $\psi$ we define the $\eta
$-Wigner transform (or distribution) of $\psi$ by replacing $\hbar$ by $\eta$
in the usual definition:
\begin{equation}
W_{\eta}\psi(x,p)=\left(  \tfrac{1}{2\pi\eta}\right)  ^{n}\int e^{-\frac
{i}{\eta}py}\psi(x+\tfrac{1}{2}y)\psi^{\ast}(x-\tfrac{1}{2}y)d^{n}y.
\label{etawig}%
\end{equation}
Of course $W_{\hbar}\psi=W\psi$ (the usual Wigner transform). The mathematical
properties of $W_{\eta}\psi$ are of course the same as those of $W\psi$,
replacing $\hbar$ everywhere with $\eta$. In particular, replacing the $\hbar
$-Fourier transform (\ref{hft}) with the $\eta$-Fourier transform
\begin{equation}
F_{\eta}\psi(p)=\left(  \tfrac{1}{2\pi\eta}\right)  ^{n/2}\int e^{-\frac
{i}{\eta}px}\psi(x)d^{n}x \label{etaft}%
\end{equation}
the marginal properties (\ref{marginal1}) become%
\begin{equation}
\int W_{\eta}\psi(x,p)d^{n}p=|\psi(x)|^{2}\text{ \ , \ }\int W_{\eta}%
\psi(x,p)d^{n}x=|F_{\eta}\psi(p)|^{2} \label{marginal2}%
\end{equation}
($\psi\in L^{1}(\mathbb{R}^{n})\cap L^{2}(\mathbb{R}^{n})$).

An important equality satisfied by the Wigner function is Moyal's
identity\footnote{It is somtimes also called the \textquotedblleft
orthogonality relation\textquotedblright\ for the Wigner function.}
\begin{equation}
\int W_{\eta}\psi(z)W_{\eta}\phi(z)d^{2n}z=\left(  \tfrac{1}{2\pi\eta}\right)
^{n}|\langle\psi|\phi\rangle|^{2} \label{Moyaleta}%
\end{equation}
which is valid for all square integrable functions $\psi$ and $\phi$ (see de
Gosson \cite{gowig}). In particular%
\begin{equation}
\int W_{\eta}\psi(z)^{2}d^{2n}z=\left(  \tfrac{1}{2\pi\eta}\right)  ^{n}%
||\phi||^{4}. \label{Moyal2eta}%
\end{equation}

Let us now address the following question: for a given $\psi$, can we find
$\phi$ such that $W_{\eta}\phi=W\psi$ for $\eta\neq\hbar$? The answer is
\textquotedblleft no\textquotedblright! More generally:

\begin{theorem}
\label{Thm3}(i) A pure state $|\psi\rangle$ does not remain a pure state if we
vary $\hbar$: let $W\psi$ be the Wigner function of $|\psi\rangle$. There does
not exist any state $|\phi\rangle$ such that $W_{\eta}\phi=W\psi$ if $\eta
\neq\hbar$. (ii) Assume that $|\psi\rangle$ becomes a mixed state when $\hbar$
is replaced with $\eta.$ Then we must have $\eta\leq\hbar$.
\end{theorem}

\begin{proof}
(i) We have $\psi,\phi\in L^{2}(\mathbb{R}^{n})$. Assume that $W_{\eta}%
\phi=W\psi$; then
\[
\int W_{\eta}\phi(x,p)d^{n}p=\int W\psi(x,p)d^{n}p
\]
hence, using the marginal properties (\ref{marginal1}), $|\phi(x)|^{2}%
=|\psi(x)|^{2}$ so that $\phi$ and $\psi$ have same norm: $||\phi||=||\psi||$.
On the other hand, using the Moyal identity (\ref{Moyal2eta}),\ the equality
$W\psi=W_{\eta}\phi$ implies that%
\begin{align*}
\int W\psi(z)^{2}d^{2n}z  &  =\left(  \tfrac{1}{2\pi\hbar}\right)  ^{n}%
||\psi||^{4}\\
\int W_{\eta}\phi(z)d^{2n}z  &  =\left(  \tfrac{1}{2\pi\eta}\right)
^{n}||\phi||^{4}%
\end{align*}
hence $\eta=\hbar$. (ii) Assume that there exists a sequence $(\phi_{j})$ of
(normalized) functions in $L^{2}(\mathbb{R}^{n})$ and a sequence of positive
constants $\alpha_{j}$ summing up to one such that $W\psi=\sum_{j}\alpha
_{j}W_{\eta}\phi_{j}$. Proceeding as above we get, using again the marginal
properties,
\begin{equation}
||\psi||^{2}=\sum_{j}\alpha_{j}||\phi_{j}||^{2}. \label{psisqr}%
\end{equation}
On the other hand, squaring $W\psi$ we get%
\[
(W\psi)^{2}=\sum_{j,k}\alpha_{j}\alpha_{k}W_{\eta}\phi_{j}W_{\eta}\phi_{k}%
\]
hence, integrating and using respectively the Moyal identity for $(W\psi)^{2}$
and $W_{\eta}\phi_{j}W_{\eta}\phi_{k}$, and the Cauchy--Schwarz inequality we
get
\begin{align*}
\left(  \tfrac{1}{2\pi\hbar}\right)  ^{n}||\psi||^{4}  &  =\left(  \tfrac
{1}{2\pi\eta}\right)  ^{n}\sum_{j,k}\alpha_{j}\alpha_{k}|\langle\phi_{j}%
|\phi_{k}\rangle|^{2}\\
&  \leq\left(  \tfrac{1}{2\pi\eta}\right)  ^{n}\sum_{j,k}\alpha_{j}\alpha
_{k}||\phi_{j}||^{2}||\phi_{k}||^{2}\\
&  =\left(  \tfrac{1}{2\pi\eta}\right)  ^{n}\left(
{\textstyle\sum\nolimits_{j}}
\alpha_{j}||\phi_{j}||^{2}\right)  ^{2}\\
&  =\left(  \tfrac{1}{2\pi\eta}\right)  ^{n}||\psi||^{4}%
\end{align*}
which implies, using (\ref{psisqr}), that $\left(  \tfrac{1}{2\pi\hbar
}\right)  ^{n}\leq\left(  \tfrac{1}{2\pi\eta}\right)  ^{n}$, that is $\eta
\leq\hbar$ as claimed.
\end{proof}

A \textit{caveat}: property (ii) in the theorem above does not say that a pure
state automatically becomes a mixed state if we decrease Planck's constant. It
merely says that if a pure state becomes mixed, it can only happen if Planck's
constant has decreased. We will see later that this is related to the
uncertainty principle.

\subsection{The Quantum Bochner Theorem}

\subsubsection{Bochner's theorem}

We begin by recalling Bochner's theorem about the Fourier transform of a
probability density. That theorem says that a (complex valued) function $f$ on
$\mathbb{R}^{m}$, continuous at the origin and such that $f(0)=1$ is the
characteristic function of a probability density on $\mathbb{R}^{m}$ if and
only if it is of \textit{positive type}, that is, if for all choices of points
$z_{1},...,z_{N}\in\mathbb{R}^{m}$ the $N\times N$ matrix
\begin{equation}
F_{(N)}=(f(z_{j}-z_{k}))_{1\leq j,k\leq N}%
\end{equation}
is positive semidefinite (that is, the eigenvalues of $F_{(N)}$ are all
$\geq0$).

Let us introduce two modifications of the symplectic Fourier transform
$F_{\sigma}$. First we allow the latter to depend on an arbitrary parameter
$\eta\neq0$ and set
\begin{equation}
F_{\sigma,\eta}a(z)=a_{\sigma,\eta}(z)=\left(  \tfrac{1}{2\pi\eta}\right)
^{n}\int e^{-\frac{i}{\eta}\sigma(z,z^{\prime})}a(z^{\prime})d^{2n}z^{\prime}.
\label{etasft}%
\end{equation}
It coincides with $F_{\sigma}$ when $\eta=\hbar$. We next define the reduced
symplectic Fourier transform $F_{\Diamond}$ is by%
\begin{equation}
a_{\Diamond}(z)=F_{\Diamond}a(z)=\int e^{-i\sigma(z,z^{\prime})}a(z^{\prime
})d^{2n}z^{\prime}. \label{adiam}%
\end{equation}
Obviously $F_{\Diamond}a$ and $F_{\sigma,\eta}a=a_{\sigma,\eta}$ are related
by the simple formula
\begin{equation}
a_{\Diamond}(z)=(2\pi\eta)^{n}a_{\sigma,\eta}(\eta z). \label{diasig12}%
\end{equation}
With this notation Bochner's theorem on Fourier transforms of probability
measures can be restated in the following way: a real function $\rho$ on
$\mathbb{R}^{2n}$ is a probability density if and only if its reduced
symplectic Fourier transform $\rho_{\Diamond}$ is continuous, $\rho_{\Diamond
}(0)=1$, and for all choices of $z_{1},...,z_{N}\in\mathbb{R}^{2n}$ the
$N\times N$ matrix $\Lambda$ whose entries are the complex numbers
$\rho_{\Diamond}(z_{j}-z_{k})$ is positive semidefinite:
\begin{equation}
\Lambda=(\rho_{\Diamond}(z_{j}-z_{k}))_{1\leq j,k\leq N}\geq0 \label{bochner}%
\end{equation}
(A matrix is said to be positive semidefinite if all its eigenvalues are
$\geq0$).

When condition (\ref{bochner}) is satisfied one says that the reduced
symplectic Fourier transform $\rho_{\Diamond}$ is of \textit{positive type}.

\subsubsection{The notion of $\eta$-positivity}

The notion of $\eta$-positivity, due to Kastler \cite{Kastler}, generalizes
Bochner's notion: let $a$ $\in\mathcal{S}^{\prime}(\mathbb{R}^{2n})$ and
$\eta$ a real number; we say that $a_{\Diamond}$\textit{\ }is of $\eta
$\textit{-positive type }if for every integer $N$ the\textit{ }$N\times N$
matrix $\Lambda_{(N)}$ with entries
\[
\Lambda_{jk}=e^{-\frac{i\eta}{2}\sigma(z_{j},z_{k})}a_{\Diamond}(z_{j}-z_{k})
\]
is positive semidefinite (which we write \textquotedblleft$\geq0$%
\textquotedblright\ for short) for all choices of $(z_{1},z_{2},...,z_{N}%
)\in(\mathbb{R}^{2n})^{N}$:%
\begin{equation}
\Lambda_{(N)}=(\Lambda_{jk})_{1\leq j,k\leq N}\geq0. \label{fzjfzk}%
\end{equation}

The condition (\ref{fzjfzk}) is equivalent to the polynomial inequalities
\begin{equation}
\sum_{1\leq j,k\leq N}\lambda_{j}\lambda_{k}^{\ast}e^{-\frac{i\eta}{2}%
\sigma(z_{j},z_{k})}a_{\Diamond}(z_{j}-z_{k})\geq0 \label{polynomial1}%
\end{equation}
for all $N\in\mathbb{N}$, $\lambda_{j},\lambda_{k}\in\mathbb{C}$, and
$z_{j},z_{k}\in\mathbb{R}^{2n}$. If $a$ is of $\eta$-positive type then it is
also of of $(-\eta)$-positive type as is immediately seen by taking the
complex conjugate of the left-hand side of (\ref{polynomial1}).

When $\eta\neq0$ we can rewrite conditions (\ref{fzjfzk})--(\ref{polynomial1})
using the symplectic $\eta$-Fourier transform: replacing $(z_{j},z_{k})$ with
$\eta^{-1}(z_{k},z_{j})$ and noting that $\sigma(z_{k},z_{j})=-\sigma
(z_{j},z_{k})$ the conditions (\ref{fzjfzk}) are equivalent to%
\[
\Lambda_{(N)}^{\prime}=(\Lambda_{jk}^{\prime}(z_{j},z_{k}))_{1\leq j,k\leq
N}\geq0
\]
where
\begin{equation}
\Lambda_{jk}^{\prime}(z_{j},z_{k})=e^{\frac{i}{2\eta}\sigma(z_{j},z_{k}%
)}a_{\sigma,\eta}(z_{j}-z_{k}). \label{fzjfzkprime}%
\end{equation}
The polynomial conditions (\ref{polynomial1}) become in this case%
\begin{equation}
\sum_{1\leq j,k\leq N}\lambda_{j}\lambda_{k}^{\ast}e^{\frac{i}{2\eta}%
\sigma(z_{j},z_{k})}a_{\sigma,\eta}(z_{j}-z_{k})\geq0 \label{polynomial2}%
\end{equation}

\subsubsection{KLM condition and the quantum Bochner theorem}

We are now going to prove an essential result (the \textquotedblleft quantum
Bochner theorem\textquotedblright) originally due to Kastler \cite{Kastler},
and Loupias and Miracle-Sole \cite{LouMiracle1,LouMiracle2}; also see
Parthasarathy \cite{partha1,partha2} and Parthasarathy and Schmidt
\cite{parthaschmidt} for different points of view. The proof we will give is
simpler than that in \cite{Kastler,LouMiracle1,LouMiracle2}, which uses the
theory of $C^{\ast}$-algebras; our proof is partially based on the discussions
in \cite{Narcow3,Narconnell,Werner}. For this we will need a technical result
from linear algebra (\textquotedblleft Schur's Lemma\textquotedblright), which
says that the entrywise product of two positive semidefinite matrices is also
positive semidefinite:

\begin{lemma}
[Schur]\label{lemmaschur}Let $A=(A_{jk})_{1\leq j,k\leq N}$ and $B=(B_{jk}%
)_{1\leq j,k\leq N}$ be two symmetric matrices with the same finite dimension
$N$. Defining the Hadamard product of these matrices by%
\[
A\circ B=(A_{jk}B_{jk})_{1\leq j,k\leq N}%
\]
then if $A$ and $B$ both are positive semidefinite, then so is $A\circ B$.
\end{lemma}

For a proof of this result see for instance Bapat \cite{Bapat}.

\begin{theorem}
[Quantum Bochner]\label{Prop2}Let $\widehat{\rho}$ be a self-adjoint trace
class operator on $L^{2}(\mathbb{R}^{n})$:%
\[
\widehat{\rho}\psi=\sum_{j}\alpha_{j}\langle\psi_{j}|\psi\rangle\psi_{j}.
\]
Let $\rho=\sum_{j}\alpha_{j}W\psi_{j}$ be the Wigner function of
$\widehat{\rho}$. We have $\widehat{\rho}\geq0$ if and only if the two
following conditions hold: (i) The reduced symplectic Fourier transform
$\rho_{\Diamond}$ is continuous and $\rho_{\Diamond}(0)=1$; (ii)
$\rho_{\Diamond}$ is of $\eta$\textit{-positive type}.
\end{theorem}

\begin{proof}
Let us first show that the conditions (i)--(ii) are necessary. Assume that
$\widehat{\rho}\geq0$; then
\begin{equation}
\rho=\sum_{j}\alpha_{j}W_{\eta}\psi_{j}%
\end{equation}
\ for a family of normalized functions $\psi_{j}\in L^{2}(\mathbb{R}^{n})$,
the coefficients $\alpha_{j}$ being $\geq0$ and summing up to one. It is thus
sufficient to show that the Wigner transform $W_{\eta}\psi$ of an arbitrary
$\psi\in L^{2}(\mathbb{R}^{n})$ is of $\eta$-positive type. This amounts to
showing that for all $(z_{1},...,z_{N})\in(\mathbb{R}^{2n})^{N}$ and all
$(\lambda_{1},...,\lambda_{N})\in\mathbb{C}^{N}$ we have
\begin{equation}
I_{N}(\psi)=\sum_{1\leq j,k\leq N}\lambda_{j}\lambda_{k}^{\ast}e^{-\frac
{i}{2\eta}\sigma(z_{j},z_{k})}F_{\sigma,\eta}W_{\eta}\psi(z_{j}-z_{k})\geq0
\label{ineq121}%
\end{equation}
for every complex vector $(\lambda_{1},...,\lambda_{N})\in\mathbb{C}^{N}$ and
every sequence $(z_{1},...,z_{N})\in(\mathbb{R}^{2n})^{N}$ (see condition
(\ref{polynomial2})). Since the $\eta$-Wigner function $W_{\eta}\psi$ and the
$\eta$-ambiguity $\operatorname*{Amb}\nolimits_{\eta}$ function are obtained
from each other by the symplectic $\eta$-Fourier transform
\[
F_{\sigma,\eta}a(z)=\left(  \tfrac{1}{2\pi\eta}\right)  ^{n}\int e^{-\frac
{i}{\eta}\sigma(z,z^{\prime})}a(z^{\prime})d^{2n}z^{\prime}%
\]
we have%
\[
I_{N}(\psi)=\sum_{1\leq j,k\leq N}\lambda_{j}\lambda_{k}^{\ast}e^{-\frac
{i}{2\eta}\sigma(z_{j},z_{k})}\operatorname*{Amb}\nolimits_{\eta}\psi
(z_{j}-z_{k}).
\]
Let us prove that
\begin{equation}
I_{N}(\psi)=\left(  \tfrac{1}{2\pi\eta}\right)  ^{n}||%
{\textstyle\sum\nolimits_{1\leq j\leq N}}
\lambda_{j}\widehat{D}_{\eta}(z_{j})\psi||^{2}; \label{equin12}%
\end{equation}
the inequality (\ref{ineq121}) will follow. Taking into account the fact that
$\widehat{D}_{\eta}(-z_{k})^{\dag}=\widehat{D}_{\eta}(z_{k})$ and using the
relation (\ref{commhwbis}) which becomes here
\begin{equation}
\widehat{D}_{\eta}(z_{0})\widehat{D}_{\eta}(z_{1})=e^{\tfrac{i}{2\eta}%
\sigma(z_{0},z_{1})}\widehat{D}_{\eta}(z_{0}+z_{1}) \label{tzotzo12}%
\end{equation}
we have, expanding the square in the right-hand side of (\ref{equin12}),
\begin{align*}
||%
{\textstyle\sum\nolimits_{1\leq j\leq N}}
\lambda_{j}\widehat{D}_{\eta}(z_{j})\psi||^{2}  &  =\sum_{1\leq j,k\leq
N}\lambda_{j}\lambda_{k}^{\ast}\langle\widehat{D}_{\eta}(z_{k})\psi
|\widehat{D}_{\eta}(z_{j})\psi\rangle\\
&  =\sum_{1\leq j,k\leq N}\lambda_{j}\lambda_{k}^{\ast}\langle\widehat{D}%
_{\eta}(-z_{j})\widehat{D}_{\eta}(z_{k})\psi|\psi\rangle\\
&  =\sum_{1\leq j,k\leq N}\lambda_{j}\lambda_{k}^{\ast}e^{-\tfrac{i}{2\eta
}\sigma(z_{j},z_{k})}\langle\widehat{D}_{\eta}(z_{k}-z_{j})\psi|\psi\rangle;
\end{align*}
in view of formula (\ref{ambhw}), which becomes here%
\[
\operatorname*{Amb}\nolimits_{\eta}\psi(z)=\left(  \tfrac{1}{2\pi\eta}\right)
^{n}\langle\widehat{D}_{\eta}(z)\psi|\psi\rangle
\]
we thus have
\[
||%
{\textstyle\sum\nolimits_{1\leq j\leq N}}
\lambda_{j}\widehat{D}_{\eta}(z_{j})\psi||^{2}=\left(  2\pi\eta\right)
^{n}\sum_{1\leq j,k\leq N}\lambda_{j}\lambda_{k}^{\ast}e^{-\tfrac{i}{2\eta
}\sigma(z_{j},z_{k})}\operatorname*{Amb}\nolimits_{\eta}\psi(z_{j}-z_{k})
\]
proving the equality (\ref{equin12}). Let us now show that, conversely, the
conditions (i) and (ii) are sufficient, \textit{i.e.} that they imply that
$(\widehat{\rho}\psi|\psi)_{L^{2}}\geq0$ for all $\psi\in L^{2}(\mathbb{R}%
^{n})$; equivalently (see formula (\ref{average2}) in Theorem \ref{Thm2})
\begin{equation}
\int\rho(z)W_{\eta}\psi(z)d^{2n}z\geq0 \label{integraltoprove}%
\end{equation}
for $\psi\in L^{2}(\mathbb{R}^{n})$. Let us set, as above,%
\begin{equation}
\Lambda_{jk}^{\prime}=e^{\frac{i}{2\eta}\sigma(z_{j},z_{k})}a_{\sigma,\eta
}(z_{j}-z_{k})
\end{equation}
where $z_{j}$ and $z_{k}$ are arbitrary elements of $\mathbb{R}^{2n}$. To say
that $a_{\sigma,\eta}$ is of $\eta$\textit{-}positive type means that the
matrix $\Lambda^{\prime}=(\Lambda_{jk}^{\prime})_{1\leq j,k\leq N}$ is
positive semidefinite; choosing $z_{k}=0$ and setting $z_{j}=z$ this means
that every matrix $(a_{\sigma,\eta}(z))_{1\leq j,k\leq N}$ is positive
semidefinite. Setting
\begin{align*}
\Gamma_{jk}  &  =e^{\frac{i}{2\eta}\sigma(z_{j},z_{k})}F_{\sigma,\eta}W_{\eta
}\psi(z_{j}-z_{k})\\
&  =e^{\frac{i}{2\eta}\sigma(z_{j},z_{k})}\operatorname*{Amb}\nolimits_{\eta
}\psi(z_{j}-z_{k})
\end{align*}
the matrix $\Gamma_{(N)}=(\Gamma_{jk})_{1\leq j,k\leq N}$ is positive
semidefinite. Let us now write
\[
M_{jk}=\operatorname*{Amb}\nolimits_{\eta}\psi(z_{j}-z_{k})\rho_{\sigma,\eta
}(z_{j}-z_{k});
\]
we claim that the matrix $M_{(N)}=(M_{jk})_{1\leq j,k\leq N}$ is positive
semidefinite. In fact, $M$ is the Hadamard product of the positive
semidefinite matrices $M_{(N)}^{\prime}=(M_{jk}^{\prime})_{1\leq j,k\leq N}$
and $M_{(N)}^{\prime\prime}=(M_{jk}^{\prime\prime})_{1\leq j,k\leq N}$ where%
\begin{align*}
M_{jk}^{\prime}  &  =e^{\frac{i}{2\eta}\sigma(z_{j},z_{k})}\operatorname*{Amb}%
\nolimits_{\eta}\psi(z_{j}-z_{k})\\
M_{jk}^{\prime\prime}  &  =e^{-\frac{i}{2\eta}\sigma(z_{j},z_{k})}\rho
_{\sigma,\eta}(z_{j}-z_{k})\text{\ }%
\end{align*}
and Schur's Lemma \ref{lemmaschur} implies that $M_{(N)}$ is also positive
semidefinite. It follows from Bochner's theorem that the function $b$ defined
by
\[
b_{\sigma,\eta}(z)=\operatorname*{Amb}\nolimits_{\eta}\psi(z)\rho_{\sigma
,\eta}(-z)=F_{\sigma,\eta}W\psi(z)\rho_{\sigma,\eta}(-z)
\]
is a probability density; in particular we must have $b(0)\geq0$. Integrating
the equality above with respect to $z$ we get, using the Plancherel formula
(\ref{Plancherelsig})
\begin{equation}
\int F_{\sigma,\eta}a(z)F_{\sigma,\eta}b(-z)d^{2n}z=\int a(z)b(z)d^{2n}z
\label{Plancherel}%
\end{equation}
for the symplectic $\eta$-Fourier transform
\begin{align*}
(2\pi\eta)^{n}b(0)  &  =\int\operatorname*{Amb}\nolimits_{\eta}\psi
(z)\rho_{\sigma,\eta}(-z)d^{2n}z\\
&  =\int W_{\eta}\psi(z)\rho(z)d^{2n}z
\end{align*}
hence the inequality (\ref{integraltoprove}) since $b(0)\geq0$.
\end{proof}

\subsection{Application to quantum states}

\subsubsection{The covariance matrix}

Let $\rho(z)$ be a real function on phase space $\mathbb{R}^{2n}$. We assume
that this function is integrable and that
\begin{equation}
\int\rho(z)d^{2n}z=1. \label{cond15}%
\end{equation}
and the marginal identities (\ref{marginal1}) hold:
\[
\int\rho(x,p)d^{n}p=|\psi(x)|^{2}\text{ \ , \ }\int\rho(x,p)d^{n}%
x=|F\psi(p)|^{2}.
\]
We will in addition assume that $\rho$ decreases sufficiently fast at
infinity:%
\begin{equation}
\int(1+|z|^{2})|\rho(z)|d^{2n}z<\infty. \label{cond25}%
\end{equation}
Setting $z_{\alpha}=x_{\alpha}$ if $1\leq\alpha\leq n$ and $z_{\alpha
}=p_{\alpha-n}$ if $n+1\leq\alpha\leq2n$, the covariances and variances of the
variables $z=(x,p)$ associated with $\rho$ are the numbers
\begin{equation}
\Delta(z_{\alpha},z_{\beta})=\int(z_{\alpha}-\left\langle z_{\alpha
}\right\rangle )(z_{\beta}-\left\langle z_{\beta}\right\rangle )\rho(z)d^{2n}z
\label{cov1}%
\end{equation}
and%
\begin{equation}
(\Delta z_{\alpha})^{2}=\Delta(z_{\alpha},z_{\alpha})=\int(z_{\alpha
}-\left\langle z_{\alpha}\right\rangle )^{2}\rho(z)d^{2n}z. \label{cov2}%
\end{equation}
In the formulas (\ref{cov1}) and (\ref{cov2}) above $\langle z_{\alpha}%
\rangle$ is the average with respect to $\rho$ of $z_{\alpha}$; more generally
one defines for an integer $k\geq0$ the moments
\begin{equation}
\langle z_{\alpha}^{k}\rangle=\int z_{\alpha}^{k}\rho(z)d^{2n}z. \label{ave}%
\end{equation}
Notice that our condition (\ref{cond25}) guarantees the existence of both
$\langle z_{\alpha}\rangle$ and $\langle z_{\alpha}^{2}\rangle$ as follows
from the trivial inequalities%
\begin{align}
\left\vert \int z_{\alpha}\rho(z)d^{2n}z\right\vert  &  \leq\int%
(1+|z|^{2})|\rho(z)|d^{2n}z<\infty\label{conv1}\\
\left\vert \int z_{\alpha}z_{\beta}\rho(z)d^{2n}z\right\vert  &  \leq
\int(1+|z|^{2})|\rho(z)|d^{2n}z<\infty. \label{conv2}%
\end{align}
It follows that the quantities (\ref{cov1}) and (\ref{cov2}) are well-defined
in view of condition (\ref{cond25}). Since the integral of $\rho$ is equal to
one, formulae (\ref{cov1}) and (\ref{cov2}) can be rewritten as%
\begin{gather}
\Delta(z_{\alpha},z_{\beta})=\langle z_{\alpha}z_{\beta}\rangle-\langle
z_{\alpha}\rangle\langle z_{\beta}\rangle\label{cov1bis}\\
(\Delta z_{\alpha})^{2}=\Delta(z_{\alpha},z_{\alpha})=\langle z_{\alpha}%
^{2}\rangle-\langle z_{\alpha}\rangle^{2}. \label{cov2bis}%
\end{gather}

We will call the symmetric $2n\times2n$ matrix%
\[
\Sigma=\left(  \Delta(z_{\alpha},z_{\beta})\right)  _{1\leq\alpha,\beta\leq2n}%
\]
the covariance matrix associated with $\rho$. For instance, when $n=1$
\[
\Sigma=%
\begin{pmatrix}
\Delta x^{2} & \Delta(x,p)\\
\Delta(p,x) & \Delta p^{2}%
\end{pmatrix}
\]
with%
\begin{gather*}
\Delta x^{2}=\langle x^{2}\rangle-\langle x\rangle^{2}\text{ , }\Delta
p^{2}=\langle p^{2}\rangle-\langle p\rangle^{2}\\
\Delta(x,p)=\langle xp\rangle-\langle x\rangle\langle p\rangle.
\end{gather*}

Here is an example: let $\rho(z)$ be given by the formula%
\begin{equation}
\rho(z)=(2\pi)^{-n}\sqrt{\det\Sigma^{-1}}e^{-\frac{1}{2}\Sigma^{-1}z^{2}}%
\end{equation}
where $\Sigma$ is a symmetric positive definite real $2n\times2n$ matrix. A
straightforward calculation involving Gaussian integrals shows that $\Sigma$
is precisely the associated covariance matrix. We will see later that
$\rho(z)$ is the Wigner function of a quantum state if $\Sigma$ satisfies a
certain condition related to the uncertainty principle. However, $\rho(z)$ can
always be viewed as a classical probability distribution.

\subsubsection{Two lemmas}

We are going to state two preliminary results which will be used for the proof
of Theorem \ref{ThmCov} below characterizing the covariance matrix of a
quantum state. We will not prove these results here, and refer to the original
sources. That we need so much preparatory material is indicative of the
difficulty of the topic we address.

We begin with the well-known Williamson's symplectic diagonalization theorem.
Let us first recall the following terminology \cite{Birk,Birkbis}: suppose
that $\Sigma$ is a symmetric positive definite real $2n\times2n$ matrix. Then
the matrix $\Sigma J$ has the same eigenvalues as the
antisymmetric\footnote{It is antisymmetric because its transpose is
$\Sigma^{1/2}J^{T}\Sigma^{1/2}=-\Sigma^{1/2}J\Sigma^{1/2}$ since $J^{T}=-J$.}
matrix $\Sigma^{1/2}J\Sigma^{1/2}$ and they are therefore of the type $\pm
i\lambda_{1},...,\pm i\lambda_{n}$ where $\lambda_{j}>0$. These numbers
$\lambda_{j}$ are called the \emph{symplectic eigenvalues} of the matrix
$\Sigma$ (or sometimes also the \emph{Williamson invariants} of $\Sigma$).
Williamson's diagonalization result generalizes to the multidimensional case
the elementary observation that every $2\times2$ real symmetric matrix
$\Sigma=%
\begin{pmatrix}
a & b\\
b & c
\end{pmatrix}
$ with $a>0$ and $ac-b^{2}>0$ can be written $\Sigma=S^{T}DS$ where
\[
S=%
\begin{pmatrix}
\sqrt{a/d} & b/\sqrt{ad}\\
0 & \sqrt{d/a}%
\end{pmatrix}
\text{ \ , \ }D=%
\begin{pmatrix}
d & 0\\
0 & d
\end{pmatrix}
\]
and $d=\sqrt{ac-b^{2}}$.

\begin{lemma}
[Williamson]\label{LemmaWill}Let $\Sigma$ be a symmetric positive definite
real $2n\times2n$ matrix. There exists $S\in\operatorname*{Sp}(n)$ such that
$\Sigma=S^{T}DS$ where $D$ is the diagonal matrix
\[
D=%
\begin{pmatrix}
\Lambda & 0\\
0 & \Lambda
\end{pmatrix}
\]
with $\Lambda=\operatorname*{diag}(\lambda_{1},...,\lambda_{n})$, the positive
numbers $\lambda_{j}$ being the symplectic eigenvalues of $M$.
\end{lemma}

\begin{proof}
See for instance\ \cite{Birk,Birkbis} and the references therein.
\end{proof}

For this we will need, in addition to Williamson's symplectic diagonalization
result, the following technical result:

\begin{lemma}
[Narcowich]\label{LemmaNarcow}Let $f(z)$ be a twice differentiable function
defined on phase space. If $f$ is of $\eta$\textit{-positive type then }%
\begin{equation}
-f^{\prime\prime}(0)+\frac{i\eta}{2}J\geq0 \label{conarcow}%
\end{equation}
where $f^{\prime\prime}(0)$ is the matrix of second derivatives (the Hessian
matrix) of $f(z)$ at $z=0$.
\end{lemma}

\begin{proof}
See Narcowich \cite{Narcow}, Lemma 2.1.
\end{proof}

\subsubsection{A necessary (but not sufficient) condition for a state to be
\textquotedblleft quantum\textquotedblright}

We are going to prove an essential result, which goes back to Narcowich
\cite{Narcow}. It says that the covariance matrix of a quantum state must
satisfy a certain condition which implies -- but is stronger than -- the
Robertson--Schr\"{o}dinger uncertainty principle.

We recall that it is assumed that the Wigner function $\rho$ satisfies the
condition
\begin{equation}
\int(1+|z|^{2})|\rho(z)|d^{2n}z<\infty.
\end{equation}
This implies, among other things, that the Fourier transform (and hence also
the symplectic Fourier transform) of $\rho$ is twice continuously
differentiable. In fact, writing%
\[
F\rho(z)=\left(  \tfrac{1}{2\pi\hbar}\right)  ^{n}\int e^{-\frac{i}{\hbar
}zz^{\prime}}\rho(z^{\prime})d^{2n}z^{\prime}%
\]
we have%
\[
\partial_{z_{\alpha}}F\rho=-\frac{i}{\hbar}F[z_{\alpha}\rho]\text{ \ ,
\ }\partial_{z_{\alpha}}\partial_{z_{\beta}}F\rho=\left(  -\frac{i}{\hbar
}\right)  ^{2}F[z_{\alpha}z_{\beta}\rho]
\]
and hence%
\begin{align*}
|\partial_{z_{\alpha}}F\rho(z)|  &  \leq\frac{1}{\hbar}\left\vert \int
z_{\alpha}\rho(z)d^{2n}z\right\vert <\infty\\
|\partial_{z_{\alpha}}\partial_{z_{\beta}}F\rho(z)|  &  \leq\left(  \frac
{1}{\hbar}\right)  ^{2}\left\vert \int z_{\alpha}z_{\beta}\rho(z)d^{2n}%
z\right\vert <\infty
\end{align*}
in view of the inequalities (\ref{conv1}), (\ref{conv2}).

\begin{theorem}
\label{ThmCov}Suppose that the phase space function $\rho$ with associated
covariance matrix $\Sigma$ is the $\eta$-Wigner transform of a density matrix
$\widehat{\rho}$. (i) We have
\begin{equation}
\Sigma+\frac{i\eta}{2}J\geq0. \label{sigmaRS}%
\end{equation}
(i.e. the Hermitian matrix $\Sigma+\frac{i\eta}{2}J$ is positive
semidefinite); this condition is equivalent to the inequality
\begin{equation}
|\eta|\leq2\lambda_{\min} \label{eta2la}%
\end{equation}
where $\lambda_{\min}$ is the smallest symplectic eigenvalue of $\Sigma$. (ii)
If holds (\ref{sigmaRS})--(\ref{eta2la}) hold, then they hold for every
$\eta^{\prime}<\eta$.
\end{theorem}

\begin{proof}
(i) That $\Sigma+\frac{i\eta}{2}J$ is Hermitian is clear: since the adjoint of
$J$ is $-J$ we have
\[
(\Sigma+\tfrac{i\eta}{2}J)^{\dag}=\Sigma^{\dag}+(\tfrac{i\eta}{2}J)^{\dag
}=\Sigma+\tfrac{i\eta}{2}J.
\]
We next remark that $\Sigma=\Sigma_{0}$ where $\Sigma_{0}$ is the covariance
matrix of $\rho_{0}(z)=\rho(z+\left\langle z\right\rangle _{\rho})$: we have
$\left\langle z\right\rangle _{0}=0$ and hence%
\[
\Delta(x_{j},x_{k})_{0}=\int x_{j}x_{k}\rho_{0}(z)d^{2n}z=\Delta(x_{j}%
,x_{k})_{\rho};
\]
similarly $\Delta(x_{j},p_{k})_{0}=\Delta(x_{j},p_{k})_{\rho}$ and
$\Delta(p_{j},p_{k})_{0}=\Delta(p_{j},p_{k})_{\rho}$. It is thus sufficient to
prove the result for the density operator $\widehat{\rho}_{0}$. We are going
to use Lemma \ref{LemmaNarcow} with $f(z)=F_{\sigma}\rho_{0}(z)=\rho
_{0,\sigma}(z)$. Observing that the (symplectic) Fourier transform
$\rho_{0,\sigma}$ is twice continuously differentiable in view of the argument
preceding the statement of the theorem, we have%
\[
\hbar^{2}\rho_{0,\sigma}^{\prime\prime}(0)=(2\pi\hbar)^{-n}%
\begin{pmatrix}
-\Sigma_{0,pp} & \Sigma_{0,xp}\\
\Sigma_{0,px} & -\Sigma_{0,xx}%
\end{pmatrix}
\]
and hence%
\begin{equation}
\eta^{2}\rho_{0,\sigma}^{\prime\prime}(0)=\left(  \tfrac{1}{2\pi\eta}\right)
^{n}J\Sigma_{0}J. \label{cj}%
\end{equation}
Since $\widehat{\rho}$ is a density matrix we have
\[
M=-2\eta^{-1}J\Sigma J+iJ\geq0;
\]
the condition $M\geq0$ being equivalent to $J^{T}MJ\geq0$ the inequality
(\ref{sigmaRS}) follows. Let us finally show that the conditions
(\ref{eta2la}) and (\ref{sigmaRS}) indeed are equivalent. Let $\Sigma=S^{T}DS$
be a symplectic diagonalization of $\Sigma$ (Lemma \ref{LemmaWill}). Since
$S^{T}JS=J$ condition (\ref{sigmaRS}) is equivalent to
\[
D+\frac{i\eta}{2}J\geq0\text{ \ , \ }D=%
\begin{pmatrix}
\Lambda & 0\\
0 & \Lambda
\end{pmatrix}
.
\]
The characteristic polynomial of $D+\frac{i\eta}{2}J$ is%
\begin{align*}
P(\lambda)  &  =%
\begin{vmatrix}
\Lambda-\lambda I_{n} & \tfrac{i\eta}{2}I_{n}\\
-\tfrac{i\eta}{2}I_{n} & \Lambda-\lambda I_{n}%
\end{vmatrix}
\\
&  =\det\left[  (\Lambda-\lambda I_{n})^{2}-\tfrac{1}{4}\eta^{2}I_{n}\right]
;
\end{align*}
the matrix $\Lambda$ being diagonal, the zeroes $\lambda$ of $P(\lambda)$ are
the solutions of the $n$ equations $(\lambda_{j}-\lambda)^{2}-\tfrac{1}{4}%
\eta^{2}=0$ that is $\lambda_{j}-\lambda=\pm$ $\tfrac{1}{2}|\eta|$. Since
$\lambda\geq0$ we must have $\lambda_{j}\geq\tfrac{1}{2}|\eta|$ for all $j$
hence (\ref{eta2la}). Property (ii) immediately follows from (\ref{eta2la});
it can also be proved directly: setting $\eta^{\prime}=r\eta$ with $0<r\leq1$
we have
\[
\Sigma+\frac{i\eta^{\prime}}{2}J=(1-r)\Sigma+r(\Sigma+\frac{i\eta}{2}J)\geq0.
\]

\end{proof}

The relation of the result above with the uncertainty principle comes from the
following observation (see Narcowich \cite{Narcow}, de Gosson \cite{Birkbis},
Chapter 13, de Gosson and Luef \cite{PR}): the relation
\begin{equation}
\Sigma+\tfrac{1}{2}i\eta J\geq0; \label{unc112}%
\end{equation}
(which we will call the \textquotedblleft strong uncertainty
principle\textquotedblright) implies the Robertson--Schr\"{o}dinger
inequalities%
\begin{equation}
(\Delta x_{j})_{\widehat{\rho}}^{2}(\Delta p_{j})_{\widehat{\rho}}^{2}%
\geq(\Delta(x_{j},p_{j})_{\widehat{\rho}})^{2}+\tfrac{1}{4}\eta^{2}
\label{unc212}%
\end{equation}
($j=1,...,n$) and $(\Delta x_{j})_{\psi}^{2}(\Delta p_{k})_{\psi}^{2}\geq0$ if
$j\neq k$. It is however essential to note that (\ref{unc112}) and
(\ref{unc212}) are \emph{not equivalent}. Here is a counterexample in the case
$n=2$. Consider the symmetric matrix
\[
\Sigma=%
\begin{pmatrix}
1 & -1 & 0 & 0\\
-1 & 1 & 0 & 0\\
0 & 0 & 1 & 0\\
0 & 0 & 0 & 1
\end{pmatrix}
\]
Assume that $\Sigma$ is a covariance matrix; we thus have $(\Delta x_{1}%
)^{2}=(\Delta x_{2})^{2}=1$ and $(\Delta p_{1})^{2}=(\Delta p_{2})^{2}=1$, and
also $\Delta(x_{1},p_{1})=\Delta(x_{2},p_{2})=0$ so that the inequalities
(\ref{unc212}) are trivially satisfied for the choice $\eta=1$ (they are in
fact equalities). The matrix $\Sigma+iJ$ is nevertheless indefinite (its
determinant is $-1$); $\Sigma$ is not even invertible.

It is also essential to realize that condition (\ref{unc112}) is
\emph{necessary} but \emph{not sufficient} for a phase space function to be
the Wigner function of a quantum state. Let us check this on the following
example due Narcowich and O'Connell \cite{Narconnell}, further discussed in de
Gosson and Luef \cite{golubis}. Consider the function $f(x,p)$ defined for
$n=1$ by%
\[
f(x,p)=(1-\tfrac{1}{2}ax^{2}-\tfrac{1}{2}bx^{2})e^{-(a^{2}x^{4}-b^{2}p^{4})}%
\]
where $a$ and $b$ are positive constants such that $ab\geq\frac{1}{4}\hbar
^{2}$. Now define
\[
\rho(x,p)=\frac{1}{2\pi}\iint e^{-i(xx^{\prime}+pp^{\prime})}f(x^{\prime
},p^{\prime})dp^{\prime}dx^{\prime};
\]
since $f(x,p)$ is an even function $\rho(x,p)$ is real; in view of the Fourier
inversion formula we have
\[
f(x,p)=\iint e^{-i(xx^{\prime}+pp^{\prime})}\rho(x^{\prime},p^{\prime
})dp^{\prime}dx^{\prime}%
\]
and hence%
\[
\iint\rho(x,p)dpdx=f(0,0)=1
\]
so that $\rho(x,p)$ is a candidate for being the Wigner function of some
density matrix. Calculating the covariance matrix $\Sigma$ associated with
$\rho$, one finds after some tedious calculations (\cite{Narconnell}, pp.4--5)
that $\Sigma+\frac{i\hbar}{2}J\geq0$. However, $\rho$ cannot be the Wigner
function of a density matrix $\widehat{\rho}$ for if this were the case we
would have $\langle p^{4}\rangle_{\widehat{\rho}}\geq0$; but by definition of
$\rho$
\[
\iint p^{4}\rho(x,p)dpdx=\frac{\partial^{4}f}{\partial x^{4}}(0,0)=-24a^{2}%
<0.
\]

\subsection{Gaussian Bosonic states}

We focus now on the case where all involved states are Gaussian: we say that a
pure or mixed quantum state is Gaussian if its Wigner function is of the type%
\begin{equation}
\rho(z)=(2\pi)^{-n}\sqrt{\det\Sigma^{-1}}e^{-\frac{1}{2}\Sigma^{-1}%
(z-z_{0})\cdot(z-z_{0})}%
\end{equation}
where $\Sigma$ (the \textquotedblleft covariance matrix\textquotedblright) is
a positive definite $2n\times2n$ matrix satisfying a certain condition that
will be stated later (intuitively speaking $\Sigma$ can't be \textquotedblleft
too small\textquotedblright\ because then $\rho(z)$ would be too sharply
peaked and thus violate the uncertainty principle of quantum mechanics).
Gaussian states appear naturally in every quantum system which can be
described or approximated by a quadratic Bosonic Hamiltonian (Wolf \textit{et
al}. \cite{wogici}); because of their peculiarities they play an exceptionally
important role in quantum mechanics and optics (see Barnett and Radmore
\cite{BaRa}).

\subsubsection{Definition and examples}

We will define a \emph{generalized Gaussian} as any complex function on
$\mathbb{R}^{n}$ of the type
\begin{equation}
\psi_{M}^{\hbar}(x)=\left(  \tfrac{1}{\pi\hbar}\right)  ^{n/4}(\det
X)^{1/4}e^{-\tfrac{1}{2\hbar}M(x-x_{0})^{2}} \label{gausshud}%
\end{equation}
where $M=X+iY$ is a complex symmetric $2n\times2n$ invertible matrix; $X$ and
$Y$ are real matrices such that $X=X^{T}>0$ and $Y=Y^{T}$. The coefficient in
front of the exponential is chosen so that $\psi_{M}^{\hbar}$ is normalized to
unity: $||\psi_{M}^{\hbar}||=1$.

Suppose that $X=I$ and $Y=0$; then%
\begin{equation}
\psi_{M}^{\hbar}(x)=\phi_{0}^{\hbar}(x)=\left(  \tfrac{1}{\pi\hbar}\right)
^{n/4}e^{-\tfrac{1}{2\hbar}|x|^{2}} \label{standardco}%
\end{equation}
is the standard (or fiducial) coherent state

Let $\phi_{M}(x)=e^{-\tfrac{1}{2\hbar}Mx^{2}}$ where $M=X+iY$ is a symmetric
complex $n\times n$ matrix such that $X=\operatorname{Re}M>0$. The Fourier
transform
\[
F\phi_{M}(p)=\left(  \tfrac{1}{2\pi\hbar}\right)  ^{n/2}\int e^{-\frac
{i}{\hbar}px}\phi_{M}(x)d^{n}x
\]
is given by
\begin{equation}
F\phi_{M}(x)=(\det M)^{-1/2}\phi_{M^{-1}}(x) \label{fofol8}%
\end{equation}
where $(\det M)^{-1/2}$ is given by the formula%
\[
(\det M)^{-1/2}=\lambda_{1}^{-1/2}\cdot\cdot\cdot\lambda_{m}^{-1/2}%
\]
the numbers $\lambda_{1}^{-1/2},...,\lambda_{n}^{-1/2}$ being the square roots
with positive real parts of the eigenvalues $\lambda_{1}^{-1},...,\lambda
_{m}^{-1}$ of $M^{-1}$ (see \textit{e.g.} Folland \cite{Folland}, Appendix A).
It follows that the Fourier transform of $\psi_{M}^{\hbar}$ is given by the
formula%
\begin{equation}
F\psi_{M}^{\hbar}(p)=\left(  \tfrac{1}{\pi\hbar}\right)  ^{n/4}(\det
X)^{1/4}(\det M)^{-1/2}\phi_{M^{-1}}(x). \label{foutoir}%
\end{equation}

\subsubsection{The Wigner transform of $\psi_{M}^{\hbar}$}

We are following here almost verbatim our discussion in \cite{Birkbis}, \S 11.2.1.

\begin{theorem}
The Wigner transform $W\psi_{M}^{\hbar}$ is the phase space Gaussian
\begin{equation}
W\psi_{M}^{\hbar}(z)=\left(  \tfrac{1}{\pi\hbar}\right)  ^{n}e^{-\tfrac
{1}{\hbar}Gz^{2}} \label{phagauss}%
\end{equation}
where $G$ is the symplectic symmetric matrix%
\begin{equation}
G=%
\begin{pmatrix}
X+YX^{-1}Y & YX^{-1}\\
X^{-1}Y & X^{-1}%
\end{pmatrix}
; \label{gsym}%
\end{equation}
in fact $G=S^{T}S$ where%
\begin{equation}
S=%
\begin{pmatrix}
X^{1/2} & 0\\
X^{-1/2}Y & X^{-1/2}%
\end{pmatrix}
\label{bi}%
\end{equation}
is a symplectic matrix.
\end{theorem}

\begin{proof}
To simplify notation we set $C(X)=\left(  \pi\hbar\right)  ^{n/4}(\det
X)^{1/4}$. By definition of the Wigner transform we have%
\begin{equation}
W\psi_{M}^{\hbar}(z)=\left(  \tfrac{1}{2\pi\hbar}\right)  ^{n}C(X)^{2}\int
e^{-\frac{i}{\hbar}py}e^{-\frac{1}{2\hbar}F(x,y)}d^{n}y \label{wigcoh11}%
\end{equation}
where the phase $F$ is defined by%
\begin{align*}
F(x,y)  &  =(X+iY)(x+\tfrac{1}{2}y)^{2}+(X-iY)(x-\tfrac{1}{2}y)^{2}\\
&  =2Xx^{2}+2iYx\cdot y+\tfrac{1}{2}Xy^{2}%
\end{align*}
so we can rewrite (\ref{wigcoh11}) as%
\[
W\psi_{M}^{\hbar}(z)=\left(  \tfrac{1}{2\pi\hbar}\right)  ^{n}e^{-\frac
{1}{\hbar}Xx^{2}}C(X)^{2}\int e^{-\frac{i}{\hbar}(p+Yx)y}e^{-\frac{1}{4\hbar
}Xy^{2}}d^{n}y\text{.}%
\]
Using the Fourier transformation formula (\ref{fofol8}) above with $x$
replaced by $p+Yx$ and $M$ by $\frac{1}{2}X$ we get
\begin{multline*}
\int e^{-\frac{i}{\hbar}(p+Yx)y}e^{-\frac{1}{4\hbar}Xy^{2}}d^{n}y=\\
(2\pi\hbar)^{n/2}\left[  \det(\tfrac{1}{2}X)\right]  ^{-1/2}C(X)^{2}%
e^{-\tfrac{1}{\hbar}X^{-1}(p+Yx)^{2}}.
\end{multline*}
On the other hand we have
\[
(2\pi\hbar)^{n/2}\left[  \det(\tfrac{1}{2}X)\right]  ^{-1/2}C(X)^{2}=\left(
\tfrac{1}{\pi\hbar}\right)  ^{n}%
\]
and hence%
\[
W\psi_{M}^{\hbar}(z)=\left(  \tfrac{1}{\pi\hbar}\right)  ^{n}e^{-\tfrac
{1}{\hbar}Gz^{2}}%
\]
where
\[
Gz^{2}=(X+YX^{-1})x^{2}+2X^{-1}Yx\cdot p+X^{-1}p^{2}\mathbf{\ }%
\]
so that $G$ is given by (\ref{gsym}). One immediately verifies that $G=S^{T}S$
where $S$ is given by (\ref{bi}) and that $S^{T}JS=J$ hence $S\in
\operatorname*{Sp}(n)$ as claimed.
\end{proof}

In particular, when $\psi_{M}^{\hbar}$ is the standard coherent state
(\ref{standardco})\ we recover the well-known formula%
\begin{equation}
W\phi_{0}^{\hslash}(z)=\left(  \tfrac{1}{\pi\hbar}\right)  ^{n}e^{-\frac
{1}{\hbar}|z|^{2}}. \label{wii8}%
\end{equation}

\subsubsection{A necessary and sufficient condition}

We are going to discuss the following question: \textit{For which values of
}$\eta$\textit{ can the Gaussian function }$\rho$\textit{ be the }$\eta
$\textit{-Wigner function of a density operator?} Narcowich \cite{Narcow2} was
the first to address this question using techniques from harmonic analysis
using the approach in Kastler's paper \cite{Kastler}; we give here a new and
simpler proof using the multidimensional generalization of Hardy's uncertainty principle.

Let us begin with what is usually called in the literature \textquotedblleft
Hardy's uncertainty principle\textquotedblright. In what follows we denote by
$F_{\eta}\psi$ the $\eta$-Fourier transform given, for $\psi\in L^{2}%
(\mathbb{R}^{n})$, by
\begin{equation}
F_{\eta}\psi(p)=\left(  \tfrac{1}{2\pi\eta}\right)  ^{n}\int e^{-\frac{i}%
{\eta}px}\psi(x)d^{n}x. \label{etafourier}%
\end{equation}
An old result (1933) due to Hardy \cite{Hardy} quantifies the
\textquotedblleft folk theorem\textquotedblright\ following which a function
and its Fourier transform cannot be simultaneously arbitrarily sharply peaked.
In fact Hardy proved that if $\psi\in L^{2}(\mathbb{R})$ satisfies%
\[
|\psi(x)|\leq Ce^{-\tfrac{1}{2\eta}ax^{2}}\text{ \ and \ }|F_{\eta}%
\psi(p)|\leq Ce^{-\tfrac{1}{2\eta}bx^{2}}%
\]
then we must have $ab\leq1$. In particular, if $ab=1$ then $\psi
(x)=Ne^{-ax^{2}/2\eta}$ for some constant $N$ (we are thus here in the
presence of a particular quantum tomography result, which says that it
suffices, in the Gaussian case, to know the position and momentum
probabilities to determines the state). We will use the following
generalization to the multidimensional case of Hardy's result:

\begin{lemma}
[Hardy]\label{LemmaHardy}Let $A$ and $B$ be two real positive definite
matrices and $\psi\in L^{2}(\mathbb{R}^{n})$, $\psi\neq0$. Assume that
\begin{equation}
|\psi(x)|\leq Ce^{-\tfrac{1}{2}Ax^{2}}\text{ \ and \ }|F_{\eta}\psi(p)|\leq
Ce^{-\tfrac{1}{2}Bp^{2}} \label{AB}%
\end{equation}
for a constant $C>0$. Then: (i) The eigenvalues $\lambda_{j}$, $j=1,...,n$, of
the matrix $AB$ are all $\leq1/\eta^{2}$; (ii) If $\lambda_{j}=1/\eta^{2}$ for
all $j$, then $\psi(x)=ke^{-\frac{1}{2}Ax^{2}}$ for some complex constant $k$.
\end{lemma}

\begin{proof}
See de Gosson and Luef \cite{goluhardy,PR}, de Gosson \cite{Birkbis}.
\end{proof}

We will also need the following positivity result:

\begin{lemma}
\label{Lemmaposex} If $R$ is a symmetric positive semidefinite $2n\times2n$
matrix, then%
\begin{equation}
P_{(N)}=\left(  Rz_{j}\cdot z_{k}\right)  _{1\leq j,k\leq N} \label{P}%
\end{equation}
is a symmetric positive semidefinite $N\times N$ matrix for all $z_{1}%
,...,z_{N}\in\mathbb{R}^{2n}$.
\end{lemma}

\begin{proof}
There exists a matrix $L$ such that $R=L^{\ast}L$ (Cholesky decomposition).
Denoting by $\langle z|z^{\prime}\rangle=z\cdot\overline{z^{\prime}}$ the
inner product on $\mathbb{C}^{2n}$ we have, since the $z_{j}$ are real
vectors,
\[
L^{\ast}z_{j}\cdot z_{k}=\langle L^{\ast}z_{j}|z_{k}\rangle=\langle
z_{j}|Lz_{k}\rangle=z_{j}\cdot(Lz_{k})^{\ast}%
\]
hence $Rz_{j}\cdot z_{k}=Lz_{j}\cdot(Lz_{k})^{\ast}$. It follows that%
\[
\sum_{1\leq j,k\leq N}\lambda_{j}\lambda_{k}^{\ast}Rz_{j}\cdot z_{k}%
=\sum_{1\leq j\leq N}\lambda_{j}Lz_{j}\left(  \sum_{1\leq j\leq N}\lambda
_{j}Lz_{j}\right)  ^{\ast}\geq0
\]
hence our claim.
\end{proof}

We now have the tools needed to give a complete characterization of Gaussian
$\eta$-Wigner functions. Recall from Theorem \ref{ThmCov} that a necessary
condition for a matrix $\Sigma$ to be the covariance matrix of a quantum state
is that it satisfies the condition $\Sigma+\frac{i\eta}{2}J\geq0$. It turns
out that in the Gaussian case this condition is also \textit{sufficient}:

\begin{theorem}
\label{ThmGauss}The Gaussian function
\begin{equation}
\rho(z)=(2\pi)^{-n}\sqrt{\det\Sigma^{-1}}e^{-\frac{1}{2}\Sigma^{-1}z^{2}}
\label{71}%
\end{equation}
is the $\eta$-Wigner transform of a positive trace class operator if and only
if it satisfies
\begin{equation}
|\eta|\leq2\lambda_{\min} \label{etalambda}%
\end{equation}
where $\lambda_{\min}$ is the smallest symplectic eigenvalue of $\Sigma$;
equivalently
\begin{equation}
\Sigma+\frac{i\eta}{2}J\geq0. \label{sigmaj}%
\end{equation}

\end{theorem}

\begin{proof}
Let us give a direct proof of the necessity of condition (\ref{etalambda}) for
the Gaussian (\ref{71}) to be the $\eta$-Wigner transform of a positive trace
class operator. Let $\widehat{\rho}=(2\pi\eta)^{n}\operatorname*{Op}_{\eta
}^{\mathrm{W}}(\rho)$ and set $a(z)=(2\pi\eta)^{n}\rho(z)$. Let $\widehat{S}$
$\in\operatorname*{Mp}(n)$; the operator $\widehat{\rho}$ is of trace class if
and only if $\widehat{S}\widehat{\rho}\widehat{S}^{-1}$ is, in which case
$\operatorname{Tr}(\widehat{\rho})=\operatorname{Tr}(\widehat{S}\widehat{\rho
}\widehat{S}^{-1})$. Choose $\widehat{S}$ with projection $S\in
\operatorname*{Sp}(n)$ such that $\Sigma=S^{T}DS$ is a symplectic
diagonalization of $\Sigma$. This choice reduces the proof to the case
$\Sigma=D$, that is to%
\begin{equation}
\rho(z)=(2\pi)^{-n}(\det\Lambda^{-1})e^{-\frac{1}{2}(\Lambda^{-1}x^{2}%
+\Lambda^{-1}p^{2})}. \label{Gaussdiag}%
\end{equation}
Suppose now that $\widehat{\rho}$ is of trace class; then there exist
functions $\psi_{j}\in L^{2}(\mathbb{R}^{n})$ ($1\leq j\leq n$) such that
\[
\rho(z)=\sum_{j}\alpha_{j}W_{\eta}\psi_{j}(z)
\]
where the $\alpha_{j}>0$ sum up to one. Integrating with respect to the $p$
and $x$ variables, respectively, the marginal conditions satisfied by the
$\eta$-Wigner transform and formula (\ref{Gaussdiag}) imply that we have%
\begin{align*}
\sum_{j}\alpha_{j}|\psi_{j}(x)|^{2}  &  =(2\pi)^{-n/2}(\det\Lambda
)^{1/2}e^{-\frac{1}{2}\Lambda^{-1}x^{2}}\\
\sum_{j}\alpha_{j}|F_{\eta}\psi_{j}(p)|^{2}  &  =(2\pi)^{-n/2}(\det
\Lambda)^{1/2}e^{-\frac{1}{2}\Lambda^{-1}p^{2}}.
\end{align*}
In particular, since $\alpha_{j}\geq0$ for every $j=1,2,...,n$,%
\[
|\psi_{j}(x)|\leq C_{j}e^{-\frac{1}{4}\Lambda^{-1}x^{2}}\text{ \ , \ }%
|F_{\eta}\psi_{j}(p)|\leq C_{j}e^{-\frac{1}{4}\Lambda^{-1}p^{2}}%
\]
here $C_{j}=(2\pi)^{-n/4}(\det\Lambda)^{1/4}/\alpha_{j}^{1/2}$. Applying
Hardy's Lemma \ref{LemmaHardy} with $A=B=\frac{1}{2}\eta\Lambda^{-1}$ we must
have $|\eta|\leq2\lambda_{j}$ for all $j=1,...,n$ which is condition
(\ref{etalambda}); this establishes the sufficiency statement. Let us finally
show that, conversely, the condition (\ref{sigmaj}) is sufficient. It is again
no restriction to assume that $\Sigma$ is the diagonal matrix $D=%
\begin{pmatrix}
\Lambda & 0\\
0 & \Lambda
\end{pmatrix}
$; the symplectic Fourier transform of $\rho$ is easily calculated and one
finds that $\rho_{\Diamond}(z)=e^{-\frac{1}{4}Dz^{2}}$. Let $\Lambda_{(N)}$ be
the $N\times N$ matrix with entries%
\[
\Lambda_{jk}=e^{-\frac{i\eta}{2}\sigma(z_{j},z_{k})}\rho_{\Diamond}%
(z_{j}-z_{k});
\]
a simple algebraic calculation shows that we have
\[
\Lambda_{jk}=e^{-\frac{1}{4}Dz_{j}^{2}}e^{\frac{1}{2}(D+i\eta J)z_{j}\cdot
z_{k}}e^{-\frac{1}{4}Dz_{k}^{2}}%
\]
and hence%
\[
\Lambda_{(N)}=\Delta_{(N)}\Gamma_{(N)}\Delta_{(N)}^{\ast}%
\]
where $\Delta_{(N)}=\operatorname*{diag}(e^{-\frac{1}{4}Dz_{1}^{2}%
},...,e^{-\frac{1}{4}Dz_{N}^{2}})$ and $\Gamma_{(N)}=(\Gamma_{jk})_{1\leq
j,k\leq N}$ with $\Gamma_{jk}=e^{\frac{1}{2}(D+i\eta J)z_{j}\cdot z_{k}}$. The
matrix $\Lambda_{(N)}$ is thus positive semidefinite if and only if
$\Gamma_{(N)}$ is, but this is the case in view of Lemma \ref{Lemmaposex}.
\end{proof}

Setting $2\lambda_{\min}=\hslash$ and writing $\Sigma$ in the block-matrix
form $%
\begin{pmatrix}
\Sigma_{xx} & \Sigma_{xp}\\
\Sigma_{px} & \Sigma_{pp}%
\end{pmatrix}
$ where $\Sigma_{xx}=(\Delta(x_{j},x_{k}))_{1\leq j,k\leq n}$, $\Sigma
_{xp}=(\Delta(x_{j},p_{k}))_{1\leq j,k\leq n}$ and so on, one shows \cite{PR}
that (\ref{sigmaj}) implies the generalized uncertainty relations (the
\textquotedblleft Robertson--Schr\"{o}dinger inequalities\textquotedblright;
see de Gosson and Luef \cite{PR} for a detailed discussion of these
inequalities)%
\begin{equation}
\Delta x_{j}^{2}\Delta p_{j}^{2}\geq\Delta(x_{j},p_{k})^{2}+\tfrac{1}{4}%
\hbar^{2} \label{RS}%
\end{equation}
where, for $\leq j\leq n$, the $\Delta x_{j}^{2}=\Delta(x_{j},x_{j})$, $\Delta
p_{j}^{2}=\Delta(p_{j},p_{j})$ are viewed as variances and the $\Delta
(x_{j},p_{k})$ as covariances. We have given a detailed discussion of the
Robertson--Schr\"{o}dinger inequalities in de Gosson and Luef \cite{PR} from
the symplectic point of view.

\section{Some Speculations}

\subsection{The fine structure constant}

Dirac \cite{Dirac} already speculated in 1937 that physical constants such as
the gravitational constant or the fine structure constant might be subject to
change over time. This question has since been a very active area of research
(see the recent reviews \cite{Duff3,Uzan}). Some scientists have suggested
that the fine structure constant, $\alpha\approx1/137$, might not be constant,
but could vary over time and space. This dimensionless constant, introduced by
Sommerfeld in 1916, measures the strength of interactions between light and
matter, or equivalently, how strong electrical and magnetic forces are. It can
be expressed as a combination of three constants: the electron charge, the
speed of light, and Planck's constant $h$:
\[
\alpha=\frac{1}{4\pi\varepsilon_{0}}\frac{e^{2}}{\hbar c}.
\]
The quest for testing the non-constancy of $\alpha$ is ongoing. The Oklo
natural nuclear reactor is known to give limits on the variation of the fine
structure constant over the period since the reactor was running
(\symbol{126}1.8 billion years). In 1999, a team of astronomers using a
telescope in Hawaii reported that measurements of light absorbed by very
distant galaxy-like objects in space called quasars -- which are so far away
that we see them today as they looked billions of years ago -- suggest that
the value of the fine structure constant was once slightly different from what
it is today. Experiments can in principle only put an upper bound on the
relative change per year. For the fine structure constant, this upper bound is
comparatively low, at roughly $10^{-17}$ per year. That claim was
controversial, and still unproven. But if true, it must mean that at least one
of the three fundamental constants that constitute $\alpha$ must vary. The
possibility that some constants of Nature could vary in space-time has
remained a subject of fascination which has motivated numerous theoretical and
experimental researches \cite{oklo1,Dyson1}.

\subsection{Planck's constant}

Kentosh and Mohageg focused on $h$, and specifically on whether $h$ depends on
where (not when) you measure it. If $h$ changes from place to place, so do the
frequencies, and thus the \textquotedblleft ticking rate\textquotedblright, of
atomic clocks. And any dependence of $h$ on location would translate as a tiny
timing discrepancy between different GPS clocks. The physicist Freeman Dyson
has suggested (private communication) that the increasing precision of
measurements of time could lead to non-ambiguous results. Mohageg and his
student Kentosh \cite{KM1,KM3} have tested the constancy of $h$ using the
freely available data from GPS. Kentosh and Mohageg were actually motivated by
the fact that $h$ also appears in the fine structure constant, whose possible
variation is a very active area of research in experimental physics. After
careful analysis of the data from seven highly stable GPS satellites, Kentosh
and Mohageg concluded that $h$ is identical at different locations to an
accuracy of seven parts in a thousand. In other words, if $h$ were a one-metre
measuring stick, then two sticks in different places anywhere in the world
would not differ by more than seven millimeters.

At least as interesting is the possible time-variation of Planck's constant
(see Mangano \textit{et al}. \cite{mangano}). This deserves to be explored
because if true it could shed some light on the Early Universe, just after the
Big Bang. In fact, if the fine structure constant has been increasing since
the Big Bang, this could perhaps be due to a decrease of Planck's constant. If
such a variation could be experimentally detected, then it would mean,
following our discussion of the quantum Bochner theorem, that the early
Universe was much more \textquotedblleft quantum\textquotedblright\ than it is
now; this would of course have major implications in terms of entanglement.

\subsection{Units}

Testing the constancy of a physical parameter means going to extraordinary
lengths in terms of precision measurements, and is intimately related to
choices of \emph{unit systems}. The physicist Michael Duff \cite{Duff}
remarked in 2002 (also see Duff \cite{Duff2,Duff3}) that all the fundamental
physical dimensions could be expressed using only one: mass. Duff first
noticed the obvious, namely that lengths can be expressed as times using $c$,
the velocity of light, as a conversion factor. One can therefore take $c=1$,
and measure lengths in seconds. The second step was to use the relation
$E=h\nu$ which relates energy to a frequency, that is to the inverse of a
time. We can thus measure a time using the inverse of energy. But energy is
equivalent to mass as shown by Einstein, so that time can be measured by the
inverse of mass. Thus, setting $c=h=1$ we have reduced all the fundamental
dimensions to one: mass. A further step consists in choosing a reference mass
such that the gravitational constant (first measured by Cavendish in 1798) is
equal to one: $G=1$. Summarizing, we have obtained a theoretical system of
units in which $c=h=G=1$. Now, a very important physical parameter\ is,
without doubt, the fine structure constant $\alpha=e^{2}/2\varepsilon_{0}hc$
($\varepsilon_{0}$ the dielectric constant); it is a dimensionless number
whose approximate value is $1/137$. There are other ways to define irreducible
unit systems. Already Stoney, noting that electric charge is quantized,
derived units of length, time, and mass in 1881 by normalizing $G,c,$ and $e$
to unity; Planck suggested in 1898--1899 that it would suffice to use $G,c,$
and $h$ to define length, mass, and time units. His proposal led to what are
called today Planck's length\ $\ell_{P}=\sqrt{Gh/c^{3}}$ and Planck mass and
time $M_{P}=\sqrt{hc/G}$ and $T_{P}=\sqrt{Gh/c^{5}}$.

\begin{acknowledgement}
This work has been financed by the grant P27773 N23 of the Austrian Research
Foundation FWF (Fonds zur F\"{o}rderung der wissenschaftlichen Forschung). It
is my pleasure to thank Glen Dennis for a careful reading of the manuscript.
\end{acknowledgement}


\begin{thebibliography}{999}                                                                                              %


\bibitem {albe}S. Albeverio, R. H\o egh-Krohn, and S. Mazzucchi,
\textit{Mathematical theory of Feynman path integrals: an introduction}. Vol.
523. Springer Science \& Business Media, 2008

\bibitem {Arvind}Arvind, B. Dutta, N. Mukunda, and R. Simon, The real
symplectic groups in quantum mechanics and optics, \textit{Pramana Journal of
Physics}, 45(6), 471--497 (1995)

\bibitem {Bapat}R. Bapat, \textit{Nonnegative Matrices and Applications},
Cambridge University Press, 1997

\bibitem {BaRa}S. M. Barnett and P. M. Radmore, \textit{Methods in Theoretical
Quantum Optics}, New York: Oxford University press 1997; reprinted 2002

\bibitem {blabru}P. Blanchard and E. Br\"{u}ning, \textit{Mathematical Methods
in Physics: Distributions, Hilbert Space Operators, Variational Methods, and
Applications in Quantum Physics}. Vol. 69. Birkh\"{a}user, 2015

\bibitem {Brislawn}C. Brislawn, Kernels of trace class operators,
\textit{Proc. Amer. Math. Soc.} 104(4), 1181--1190 (1988)

\bibitem {buzek}V. Bu\v{z}ek, G. Adam, and G. Drobn\'{y}, Reconstruction of
Wigner Functions on Different Observation Levels, \textit{Ann. Phys}. 245,
37--97 (1996)

\bibitem {cogoni1}E. Cordero, M. de Gosson, and F. Nicola, On the
invertibility of Born--Jordan quantization, \textit{J. Math. Pure Appl.}
105(4), 537--557 (2016)

\bibitem {cogoni17}E. Cordero, M. de Gosson, and F. Nicola, Positivity of
trace class Operators and the Cohen Class;\ Application to Born--Jordan
Operators [preprint 2016]

\bibitem {oklo1}T. Damour and F. Dyson, The Oklo bound on the time variation
of the fine structure constant revisited, \textit{Nuclear Physics B} 480(1),
37--54 (1996)

\bibitem {dariano}G. M. D'Ariano,\ Universal quantum observables,
\textit{Phys. Lett. A} 300, 1--6 (2002)

\bibitem {damapa}G. M. D'Ariano,\ C. Macchiavello, and M. G. A. Paris,
Detection of the density matrix through optical homodyne tomography without
filtered back projection, \textit{Phys. Rev. A} 50(5), 4298--4303 (1994)

\bibitem {daube}I. Daubechies, Continuity statements and counterintuitive
examples in connection with Weyl quantization, \textit{J. Math. Phys.} 24(6),
1453--1461 (1983)

\bibitem {dipra}N. C. Dias and J. N. Prata, The Narcowich--Wigner spectrum of
a pure state, \textit{Rep. Math. Phys}. 63(1), 43--54 (2009)

\bibitem {digopra14}N. Dias, M. de Gosson, and J. Prata, Maximal covariance
group of Wigner transforms and pseudo-differential operators. \textit{Proc.
Amer. Math. Soc.} 142(9), 3183--3192 (2014)

\bibitem {dix}J. Dixmier, \textit{Les }$\mathit{C}$\textit{*-alg\`{e}bres et
leurs repr\'{e}sentations}, Gauthier--Villars, 1969

\bibitem {Dirac}P. A. M. Dirac, A New Basis for Cosmology. \textit{Proc. Royal
Soc. London A}. 165 (921), 199--208 (1938)

\bibitem {Dirac1}P. A. M. Dirac, Long range forces and broken symmetries,
\textit{Proc. Royal Soc. London A: Mathematical, Physical and Engineering
Sciences.} Vol. 333. No. 1595. The Royal Society, 1973

\bibitem {duwong}J. Du and M. W. Wong, A trace formula for Weyl transforms,
\textit{Approximation Theory and its Applications} 16(1), 41--45 (2000)

\bibitem {Duff}M. J. Duff, L. B. Okun, and G. Veneziano, Trialogue on the
number of fundamental constants, \textit{JHEP} 2002(03) 023 (2002)

\bibitem {Duff2}M. J. Duff, Comment on time-variation of fundamental
constants\qquad arXiv:hep-th/0208093

\bibitem {Duff3}M. J. Duff, How fundamental are fundamental constants?
\textit{Contemporary Physics}, 56(1), 35--47 (2015)

\bibitem {Dyson1}F. J. Dyson, in \textit{The fundamental constants and their
time variation, Aspects of quantum theory}, eds. J. E. Lannutti and E. P.
Wigner (CUP, Cambridge) 213--236 (1972)

\bibitem {Emch}G. G. Emch, Geometric dequantization and the correspondence
problem, \textit{International Journal of Theoretical Physics} 22(5) 397--420 (1983)

\bibitem {espo}G. Esposito, G. Marmo, G. Miele, and G. Sudarshan,
\textit{Advanced Concepts in Quantum Mechanics}, Cambridge University Press, 2015

\bibitem {Fano}U. Fano, Description of States in Quantum Mechanics by Density
Matrix and Operator Techniques, \textit{Rev. Mod. Phys.} 29, 71--93 (1957)

\bibitem {Folland}G. B. Folland, \textit{Harmonic Analysis in Phase space},
Annals of Mathematics studies, Princeton University Press, Princeton, N. J. 1989

\bibitem {Gracia1}J. M. Gracia-Bond\'{\i}a and J. C. Varilly, Algebras of
distributions suitable for phase-space quantum mechanics, I, \textit{J. Math.
Phys.} 29(4), 869--879 (1988)

\bibitem {Gracia2}J. M. Gracia-Bond\'{\i}a and J. C. Varilly, Nonnegative
mixed states in Weyl-Wigner-Moyal theory,\textit{ Phys. Lett. A} 128(1-2),
20--24 (1988)

\bibitem {Garcia}M. Garc\'{\i}a-Bull\'{e}, W. Lassner, and K. B. Wolf, The
metaplectic group within the Heisenberg--Weyl ring, \textit{J. Math. Phys}.
27(1), 29--36 (1986)

\bibitem {gena}I. M. Gel'fand and M. A. Naimark, On the imbedding of normed
rings into the ring of operators on a Hilbert space. \textit{Math. Sbornik.}
12(2), 197--217 (1943)

\bibitem {gia}G. Giachetta, L. Mangiarotti, and G. A.
Sardanashvili,\textit{\ Geometric and algebraic topological methods in quantum
mechanics. Singapore}, World Scientific, 2005

\bibitem {AIF}M. de Gosson, Maslov indices on the metaplectic group
$\operatorname*{Mp}(n)$. \textit{Ann. Inst. Fourier} 40(3), 537--555 (1990)

\bibitem {LMP}M. de Gosson, On the Weyl representation of metaplectic
operators. \textit{Lett. Math. Phys}. 72(2), 129--142 (2005)

\bibitem {Birk}M. de Gosson, \textit{Symplectic geometry and quantum
mechanics.} Vol. 166. Springer Science \& Business Media, 2006

\bibitem {Birkbis}M. de Gosson, \textit{Symplectic methods in harmonic
analysis and in mathematical physics.} Vol. 7. Springer Science \& Business
Media, 2011

\bibitem {Springer}M. de Gosson, \textit{Born-Jordan Quantization: Theory and
Applications}. Vol. 182. Springer, 2016

\bibitem {FOOP1}M. de Gosson, Born--Jordan quantization and the equivalence of
the Schr\"{o}dinger and Heisenberg pictures, \textit{Found. Phys}. 44(10),
1096--1106 (2014)

\bibitem {PR2}M. de Gosson, From Weyl to Born--Jordan quantization: the
Schr\"{o}dinger representation revisited, \textit{Phys. Reps}. 623, 1--58 (2016)

\bibitem {ICP}M. de Gosson, \textit{The Principles of Newtonian and Quantum
Mechanics, }Imperial College Press, London, 2001; 2d edition World Scientific, 2017

\bibitem {SPRINGER}M. de Gosson, \textit{Introduction to Born--Jordan
Quantization}, Springer-Verlag, series Fundamental Theories of Physics, 2016

\bibitem {gowig}M. de Gosson, \textit{The Wigner Transform}, World Scientific,
Advanced Textbooks in Mathematics (series), 2017 [in press]

\bibitem {Charlyne1}C. de Gosson and M. de Gosson, The Phase Space Formulation
of Time-Symmetric Quantum Mechanics, \textit{Quanta} 4(1) (2015)

\bibitem {goluhardy}M. de Gosson and F. Luef, Quantum states and Hardy's
formulation of the uncertainty principle: a symplectic approach, \textit{Lett.
Math. Phys.} 80(1), 69--82 (2007)

\bibitem {golubis}M. de Gosson and F. Luef, Remarks on the fact that the
uncertainty principle does not determine the quantum state, \textit{Phys.
Lett. A} 364(6), 453--457 (2007)

\bibitem {PR}M. de Gosson and F. Luef, Symplectic Capacities and the Geometry
of Uncertainty: the Irruption of Symplectic Topology in Classical
and\textit{\ }Quantum Mechanics, \textit{Phys. Reps.} 484, 131--179 (2009)

\bibitem {Grossmann}A. Grossmann, \textit{Parity operators and quantization of
}$\delta$\textit{-functions}, \textit{Commun. Math. Phys}. 48, 191--193 (1976)

\bibitem {Hardy}G. H.\ Hardy, A theorem concerning Fourier transforms,
\textit{J. London. Math. Soc.} 8, 227--231 (1933)

\bibitem {hilleryetal}M. O. S. Hillery, R. F. O'Connell, M. Scully, and E. P.
Wigner, Distribution Functions in Physics: Fundamentals, \textit{Phys. Rep}s.
106(3), 121--167 (1984)

\bibitem {Hudson}R. L. Hudson, When is the Wigner quasi-probability density
non-negative?, \textit{Rep. Math. Phys.} 6, 249--252 (1974)

\bibitem {Janssen}A. J. E. M. Janssen, A note on Hudson's theorem about
functions with nonnegative Wigner distributions,\textit{\ Siam. J. Math.}
\textit{Anal.} 15(1), 170--176 (1984)

\bibitem {kadison}R. V. Kadison and J. R. Ringrose, \textit{Fundamentals of
the Theory of Operator Algebras,} AMS, 1991; \textit{Fundamentals of the
Theory of Operator Algebras }Volume II: Advanced Theory. Vol. 16. American
Mathematical Society 2015

\bibitem {ib}A. Ibort, V. I. Man'ko, G. Marmo, A. Simoni, and F. Ventriglia,
An introduction to the tomographic picture of quantum mechanics, \textit{Phys.
Scr}. 79, 065013 (2009)

\bibitem {Kastler}D. Kastler, The $C^{\ast}$-Algebras of a Free Boson Field,
\textit{Commun. math. Phys}. 1, 14--48 (1965)

\bibitem {KM1}J. Kentosh and M. Mohageg, Global positioning system test of the
local position invariance of Planck's constant, \textit{Phys. Rev. Lett.}
108(11), 110801 (2012)

\bibitem {KM3}J. Kentosh and M. Mohageg, Testing the local position invariance
of Planck's constant in general relativity, \textit{Physics Essays} 28(2),
286--289 (2015)

\bibitem {land1}N. P. Landsman, 2006 Lecture Notes on Hilbert Spaces and
Quantum Mechanics, http://courses.daiict.ac.in/pluginfile.php/15332/mod\_resource/content/0/hilbert\_space\_why.pdf

\bibitem {landsman}N. P. Landsman, \textit{Mathematical topics between
classical and quantum mechanics}. Springer Science \& Business Media, 2012

\bibitem {lepa94}U. Leonhardt and H. Paul, Realistic optical homodyne
measurements and quasiprobability distributions, \textit{Phys. Rev. A} 48,
4598 (1993)

\bibitem {Leray}J. Leray, \textit{Lagrangian Analysis and Quantum
Mechanics},\textit{\ a mathematical structure related to asymptotic expansions
and the Maslov index,} MIT Press, Cambridge, Mass., (1981); translated from
\textit{Analyse Lagrangienne} RCP 25, Strasbourg Coll\`{e}ge de France (1976--1977)

\bibitem {Littlejohn}R. G. Littlejohn, The semiclassical evolution of wave
packets, \textit{Phys. Reps.} 138\textbf{(}4--5), 193--29 (1986)

\bibitem {LouMiracle1}G. Loupias and S. Miracle-Sole, $C^{\ast}$-Alg\`{e}bres
des syst\`{e}mes canoniques, I, \textit{Commun. math. Phys}, 2, 31--48 (1966)

\bibitem {LouMiracle2}G. Loupias and S. Miracle-Sole, $C^{\ast}$-Alg\`{e}bres
des syst\`{e}mes canoniques, II, \textit{Ann. Inst. Henri Poincar\'{e}} 6(1),
39--58 (1967)

\bibitem {lvra09}A. I. Lvovsky and M. G. Raymer, Continuous-variable optical
quantum-state tomography, \textit{Rev. Mod. Phys}. 8, 299--332 (2009)

\bibitem {mapa14}L. Maccone and A. K. Pati, Stronger Uncertainty Relations for
All Incompatible Observables, Phys. Rev. Lett. 113, 260401 (2014)

\bibitem {mancini}S. Mancini, V. I. Man'ko, and P. Tombesi, Symplectic
tomography as classical approach to quantum systems, \textit{Phys. Lett. A}
213, 1--6 (1996)

\bibitem {mangano}G. Mangano, F. Lizzi, and A. Porzio, Inconstant Planck's
constant, \textit{International Journal of Modern Physics A} 30(34), 1550209 (2015)

\bibitem {mankomanko}O. Man'ko and V. I. Man'ko, Quantum states in probability
representation and tomography, \textit{Journal of Russian Laser Research
}18(5), 407--444 (1997)

\bibitem {Marciano}W. J. Marciano, Time Variation of the Fundamental
"Constants" and Kaluza-Klein Theories, \textit{Phys. Rev. Lett.} 52(7),
489--491 (1984)

\bibitem {mofean}G. Mourgues, M. R. Feix, and J. C. Andrieux, Not necessary
but sufficient condition for the positivity of generalized Wigner functions,
\textit{J. Math. Phys.} 26\textbf{(}10) 2554--2555 (1985)

\bibitem {Narcow2}F. J. Narcowich, Conditions for the convolution of two
Wigner distributions to be itself a Wigner distribution, \textit{J. Math.
Phys.}, 29(9), 2036--2041 (1988)

\bibitem {Narcow3}F. J. Narcowich, Distributions of $\eta$-positive type and
applications, \textit{J. Math. Phys}., 30(11), 2565--2573 (1989)

\bibitem {Narcow}F. J. Narcowich, Geometry and uncertainty, \textit{J. Math.
Phys.} 31(2), 354--364 (1990)

\bibitem {Narconnell}F. J. Narcowich and R F. O'Connell, Necessary and
sufficient conditions for a phase-space function to be a Wigner
distribution,\textit{\ Phys. Rev. A,} 34(1), 1--6 (1986)

\bibitem {Narconnell88}F. J. Narcowich and R. F. O'Connell, A unified approach
to quantum dynamical maps and Gaussian Wigner distributions, \textit{Phys.
Lett. A}, 133\textbf{(}4), 167--170 (1988)

\bibitem {Nicola}F. Nicola, Convergence in $L^{p}$ for Feynman path integrals,
\textit{Advances in Mathematics} 294, 384--409 (2016)

\bibitem {Neumann}J. von Neumann, Wahrscheinlichkeitstheoretischer Aufbau der
Quantenmechanik, \textit{G\"{o}ttinger Nachrichten }1, 245--272 (1927)

\bibitem {pare}M. Paris and J. Reh\'{a}\v{c}ek, eds. \textit{Quantum State
Estimation}, Lecture Notes in Physics Vol. 649 Springer, Berlin, 2004

\bibitem {partha1}K. R. Parthasarathy, \textit{An introduction to quantum
stochastic calculus,} Springer Science \& Business Media, 2012

\bibitem {partha2}K. R. Parthasarathy, What is a Gaussian state?,
\textit{Commun. Stoch. Anal,} 4(2), 143--160 (2010)

\bibitem {parthaschmidt}K. R. Parthasarathy and K. Schmidt, \textit{Positive
definite kernels, continuous tensor products, and central limit theorems }of
Probability Theory, Springer LNM 272, Berlin, 1972

\bibitem {Peres}A. Peres, \textit{Quantum theory: concepts and methods}, Vol.
57. Springer Science \& Business Media, 2006

\bibitem {Pauli}P. Wolfgang, \textit{General principles of quantum mechanics},
Springer Science \& Business Media, 2012 [original title: \textit{Prinzipien
der Quantentheorie}, publ. in : Handbuch der Physik, v.5.1, 1958]

\bibitem {potocek}V. Poto\v{c}ek and S. M. Barnett, On the exponential form of
the displacement operator for different systems, \textit{Phys. Scr.} 90,
065208 (2015)

\bibitem {Radon}J. Radon, \"{U}ber die Bestimmung von Funktionen durch ihre
Integralwerte l\"{a}ngsgewisser Mannigfaltigkeiten, \textit{S\"{a}chs. Akad.
Wiss. Leipzig, Math. Nat. Kl.} 69, 262--277 (1917)

\bibitem {Rieffel}M. A. Rieffel, \textit{Quantization and }$C^{\ast}%
$\textit{-algebras}, Contemporary mathematics, Amer. Math. Soc.167, 67--97 (1994)

\bibitem {Royer}A. Royer, \textit{Wigner functions as the expectation value of
a parity operator}, Phys. Rev.\textit{ A} 15, 449--450 (1977)

\bibitem {Royer2}A. Royer, Measurement of quantum states and the Wigner
function, \textit{Found. Phys.} 19(1), 3--32 (1989)

\bibitem {sh87}M. A. Shubin, \textit{Pseudodifferential Operators and Spectral
Theory}, Springer-Verlag, 1987 [original Russian edition in Nauka, Moskva 1978]

\bibitem {BSimon}B. Simon, \textit{Trace Ideals and their Applications},
Cambridge U. P., Cambridge (1979)

\bibitem {simonetal1}R. Simon, N. Mukunda, and B. Dutta, Quantum-noise matrix
for multimotie systems: $U(n)$ invariance, squeezing, and normal forms,
\textit{Phys. Rev. A} 49, 1567 (1994)

\bibitem {sriwolf}M. D. Srinivas and E. Wolf, Some nonclassical features of
phase-space representations of quantum mechanics, \textit{Phys. Rev. D} 11(6),
1477--1485 (1975)

\bibitem {takta}L. A. Takhtajan, \textit{Quantum mechanics for mathematicians}%
. Vol. 95. American Mathematical Soc. 2008

\bibitem {the}G. S. Thekkadath, L. Giner, Y. Chalich, M. J. Horton, J. Banker,
and J. S. Lundeen, Direct Measurement of the Density Matrix of a Quantum
System, \textit{Phys. Rev. Lett.} 117, 120401 (2016)

\bibitem {Uzan}J.-P. Uzan, Varying Constants, Gravitation and Cosmology,
\textit{Living Rev. Relativity} 14, 2 (2011)

\bibitem {Vogel}K. Vogel and H. Risken, Determination of quasiprobability
distributions in terms of probability distributions for the rotated quadrature
phase,\textit{ Phys. Rev}. \textit{A} 40(5) 2847--2849 (1989)

\bibitem {waxu}Y. Wang and C. Xu, Density matrix estimation in quantum
homodyne tomography, \textit{Statistica Sinica }953--973 (2015)

\bibitem {oklo2}J. K. Webb, M. T. Murphy, V. V. Flambaum, V. A. Dzuba, J. D.
Barrow, C. W. Churchill, and A. M. Wolfe, Further evidence for cosmological
evolution of the fine structure constant. \textit{Phys. Rev, Lett.} 87(9),
091301 (2001)

\bibitem {Werner}R. Werner, Quantum harmonic analysis on phase space,
\textit{J. Math. Phys}. 25(5), 1404--1411 (1984)

\bibitem {Weyl}H. Weyl, Quantenmechanik und Gruppentheorie,
\textit{Zeitschrift f\"{u}r Physik} 46 (1927)

\bibitem {Wigner}E. Wigner, On the quantum correction for thermodynamic
equilibrium, \textit{Phys. Rev.} 40, 799--755 (1932)

\bibitem {wogici}M. M. Wolf, G. Giedke, and J. Cirac, Extremality of Gaussian
Quantum States, \textit{Phys. Rev. Lett.} 96, 080502 (2006)
\end{thebibliography}
\end{document}